\documentclass[12pt]{article}

\usepackage[margin=2cm]{geometry}

\usepackage{amsmath,amsthm,amsfonts,amscd,amssymb,bbm,mathrsfs,enumerate,url}
\usepackage{authblk}

\usepackage{braket}

\usepackage{forest}
\usepackage{comment}

\usepackage{graphicx,tikz}
\usepackage[colorlinks = true,
            linkcolor = black,
            urlcolor  = blue,
            citecolor = black,
            anchorcolor = black,
            pdfborder={0 0 0}]{hyperref}
\makeatletter
\newcommand*{\textlabel}[2]{%
  \edef\@currentlabel{#1}
  \phantomsection
  #1\label{#2}
}

\usetikzlibrary{matrix,arrows}
\usetikzlibrary{decorations.pathreplacing}
\usetikzlibrary{decorations.pathmorphing}
\usetikzlibrary{patterns,fadings}

\newcounter{OS}

\def\semicolon{;}
\def\applytolist#1{
    \expandafter\def\csname multi#1\endcsname##1{
        \def\multiack{##1}\ifx\multiack\semicolon
            \def\next{\relax}
        \else
            \csname #1\endcsname{##1}
            \def\next{\csname multi#1\endcsname}
        \fi
        \next}
    \csname multi#1\endcsname}

\def\calc#1{\expandafter\def\csname c#1\endcsname{{\mathcal #1}}}
\applytolist{calc}QWERTYUIOPLKJHGFDSAZXCVBNM;
\def\bbc#1{\expandafter\def\csname bb#1\endcsname{{\mathbb #1}}}
\applytolist{bbc}QWERTYUIOPLKJHGFDSAZXCVBNM;
\def\bfc#1{\expandafter\def\csname bf#1\endcsname{{\mathbf #1}}}
\applytolist{bfc}QWERTYUIOPLKJHGFDSAZXCVBNM;
\def\sfc#1{\expandafter\def\csname s#1\endcsname{{\sf #1}}}
\applytolist{sfc}QWERTYUIOPLKJHGFDSAZXCVBNM;
\def\rfc#1{\expandafter\def\csname r#1\endcsname{{\mathrm #1}}}
\applytolist{rfc}QWERTYUIOPLKJHGFDSAZXCVBNMqwertyuiopasdfghjklzxcvbnm;
\def\frfc#1{\expandafter\def\csname fr#1\endcsname{{\mathfrak #1}}}
\applytolist{frfc}QWERTYUIOPLKJHGFDSAZXCVBNMqwertyuiopasdfghjklzxcvbnm;
\def\scfc#1{\expandafter\def\csname sc#1\endcsname{{\mathscr #1}}}
\applytolist{scfc}QWERTYUIOPLKJHGFDSAZXCVBNMqwertyuiopasdfghjklzxcvbnm;

\DeclareMathOperator{\id}{id}
\DeclareMathOperator{\supp}{supp}

\DeclareMathOperator{\End}{End}

\DeclareMathOperator{\CP}{\bbC\bbP^1}

\DeclareMathOperator{\euclidean}{\rE_\mathrm{e}(2)}
\DeclareMathOperator{\PSLtwoC}{\mathrm{PSL}_2(\bbC)}

\def\h{\bar{h}}
\def\lat{{\mathrm{lat}}}
\def\z{\bar{z}}
\def\p{\bar{p}}

\def\ze{\zeta}

\def\zee{\bar{\zeta}}
\def\la{\lambda}
\def\si{\sigma}
\def\uz{\underline{z}}

\def\uze{\underline{\zeta}}

\def\uh{\underline{h}}
\def\Ld{\bar{L}}
\def\D{\bar{D}}
\def\al{\alpha}
\def\be{\beta}
\def\ep{\epsilon}
\def\ga{\gamma}
\def\vac{{\bf 1}}

\def\tw{{{I\hspace{-.1em}I}_{1,1}}}

\DeclareMathOperator{\om}{{\omega}}

\def\bb1{\mathbbm{1}}

\def\<{\langle}
\def\>{\rangle}

\newtheorem{theorem}{Theorem}[section]

\newtheorem{corollary}[theorem]{Corollary}
\newtheorem{proposition}[theorem]{Proposition}
\newtheorem{lemma}[theorem]{Lemma}
\theoremstyle{remark}
\newtheorem{remark}[theorem]{Remark}

\title{Osterwalder-Schrader axioms for unitary full vertex operator algebras}
\author[1]{Maria Stella Adamo\thanks{{\tt maria.stella.adamo@fau.de}}}
\author[2]{Yuto Moriwaki \thanks{{\tt moriwaki.yuto@gmail.com}}}
\author[3]{Yoh Tanimoto\thanks{{\tt hoyt@mat.uniroma2.it}}}

\affil[1]{Department Mathematik, Friedrich-Alexander-Universit\"at Erlangen-N\"urnberg \authorcr
Cauerstrasse 11, 91058 Erlangen, Germany}
\affil[2]{Interdisciplinary Theoretical and Mathematical Science Program (iTHEMS) \authorcr
Wako, Saitama 351-0198, Japan}
\affil[3]{Dipartimento di Matematica, Universit\`a di Roma Tor Vergata,\authorcr
   Via della Ricerca Scientifica 1, I-00133 Roma, Italy}
\date{}

\begin{document}
\maketitle
\begin{abstract}
 Full Vertex Operator Algebras (full VOA) are extensions of two commuting Vertex Operator Algebras,
 introduced to formulate compact two-dimensional conformal field theory.
 We define unitarity, polynomial energy bounds and polynomial spectral density for full VOA.
 Under these conditions and local $C_1$-cofiniteness of the simple full VOA, 
 we show that the correlation functions of quasi-primary fields define tempered distributions
 and satisfy a conformal version of the Osterwalder-Schrader axioms,
 including the linear growth condition.
 
 As an example, we show that a family of full extensions of the Heisenberg VOA satisfies all these assumptions.
\end{abstract}

\section{Introduction}
Two-dimensional conformal field theory (2d CFT) \cite{BPZ84} has been a very fruitful playground of mathematics (infinite-dimensional Lie groups, modular form, 
subfactors, knot invariants, quantum groups...) and physics (critical phenomena, string theory, integrable system, fixed point of renormalization group...).
Thanks to its large conformal symmetry in the two-dimensional space, it can be put into powerful frameworks
such as vertex operator algebras (VOAs) \cite{Borcherds86,FLM88VertexOperatorAlgebras} or conformal nets \cite{CKLW18} and interesting structural results have been obtained.
Even more categorical/functorial frameworks such as Segal's axioms have been also proposed \cite{Segal04, Huang97Two-dimensionalConformalGeometry}.

On the other hand, 2d CFTs should be placed in a wider context of quantum field theory (QFT) \cite{GJ87}, including
massive models and higher space(time) dimensions, because many interesting relations between CFTs and massive models
are expected (cf.\! renormalization, perturbation...).
Frameworks specific to 2d CFTs are not suitable to discuss such relations.
Fortunately, there are more general frameworks for QFT: the G{\aa}rding-Wightman axioms \cite{SW00}, the Osterwalder-Schrader axioms \cite{OS73, OS75},
and the Araki-Haag-Kaster axioms \cite{Haag96, Araki99}, where many examples of QFT (mostly massive in low dimensions)
have been constructed \cite{Summers12}.
It is therefore an important question whether 2d CFTs fit into such general
frameworks. In recent years, better understanding has been obtained for chiral components of 2d CFTs
(see e.g.\! \cite{CKLW18, RTT22}), yet the chiral components contain only partial information of the complete,
so-called full CFTs. In addition, algebraic and functorial approaches are often concerned with
Euclidean geometry, while analytic approaches such as conformal nets on $S^1$ inherits the local structure
of the Lorentzian geometry. 
As recently pointed out \cite[Introduction]{FSWY23AlgebraicStructures},
the relation between Euclidean and Lorentzian approaches has remained to be understood in full detail.
For a closely related question,
a precisely statement has not been known \cite[Section 1.2]{KR10AlgebraicStructures}.

In this work, we fill in this gap: we consider an algebraic framework for 2d full CFT
and define unitarity. Under technical (but often satisfied in examples) conditions,
we prove that the correlation functions of quasi-primary fields satisfy a conformal extension of the Osterwalder-Schrader axioms.
This shows that unitary 2d CFTs are indeed special cases of Euclidian QFT and opens the possibility to study their relations in a single framework.

Quantum fields as operator-valued distributions on a Hilbert space are formulated as the G{\aa}rding-Wightman axioms \cite{SW00}.
By the Wightman reconstruction theorem, they are equivalent to the Wightman axioms of correlation functions satisfying
the so-called Wightman axioms.
Osterwalder and Schrader axiomatized some properties of the analytic continuations of Wightman correlation functions
to Euclidean points and showed that the original Wightman correlation functions can be reconstructed from
their Euclidean counterparts \cite{OS73, OS75}. We note that there are versions of the Osterwalder-Schrader (OS) axioms
regarding the regularity condition. It turned out that a regularity condition that gives an equivalence
between the OS axioms and Wightman axioms is difficult to check in examples (see e.g.\! \cite[\S II.4]{Simon74}),
while there is a sufficient condition for the reconstruction, called the linear growth condition (see \cite[Section 9]{KQR21DistributionsII}
for a very readable account of it). As our purpose is to put 2d CFT in a general framework for QFT,
we will prove a variant of OS axioms that include conformal invariance and linear growth.

A mathematical formulation of 2d Euclidean CFT was made by Huang and Kong \cite{HK07FullFieldAlgebras}.
They introduced the notion of full field algebra and investigated rational CFTs. 
In this paper, we use the notion of full VOA \cite{Moriwaki21}, which introduced to investigate irrational CFTs based on the bootstrap hypothesis in physics \cite{Polyakov74Non-Hamiltonian}.
It has already been shown that from a full VOA, under local $C_1$-cofiniteness, correlation functions can be defined and are real analytic
on $\bbR^2$ and symmetric under permutations \cite{Moriwaki22Vertex, Moriwaki24SwissCheese}, and many such irrational CFTs have been constructed \cite{Moriwaki21Code,Moriwaki22WZW}.

In this paper, we introduce the notion of unitarity for full VOA, as a natural generalization of \cite{DL14Unitary} in order to link algebraic and analytical approaches.
We show that unitarity implies reflection positivity of \cite{OS73}
and clustering of the correlation functions follow essentially from the uniqueness of the vacuum, which turn out to be equivalent to the simplicity of unitary full VOA.

The most technical part of the OS axioms is the linear growth condition.
For this, we assume a certain bound on quasi-primary fields in the full VOA analogous to polynomial energy bounds in \cite{CKLW18}
and bounds on the density of the spectrum. These conditions are often easy to check in examples.
We show that polynomial energy bounds and polynomial spectral density imply the linear growth condition,
and in particular, that the correlation functions define tempered distributions.
We also show that the conformal invariance of correlation functions in a sense similar to \cite{LM75} follows from the Virasoro symmetry of a full VOA.
In this way, we fully prove the conformal version of the OS axioms that enables the reconstruction of Wightman fields.

We expect that our assumptions are satisfied by virtually all compact unitary 2d CFTs.
In order to verify that our definition of unitarity is the right one and our bounds are not too strict,
we exhibit the case of extensions of the Heisenberg algebra, also known as the Narain CFTs \cite{Narain86NewHeterotic},
which are constructed as a deformation family of full VOAs in \cite{Moriwaki21}.
We show that those VOAs satisfy unitarity, polynomial energy bounds and polynomial spectral density, showing that
all quasi-primary fields there correspond to Wightman fields.
Note that previously we have directly constructed Wightman fields for 2d extensions of the $\rU(1)$-current (the Heisenberg algebra with rank $1$) \cite{AGT23Pointed}, and comparing these is a future work.

Our results not only clarify the relations between various approaches to 2d CFT, but put it in a larger context of
general QFT that can host further developments. For example, it is conjectured that certain massive integrable models
can be obtained by deforming CFTs \cite{Zamolodchikov89-1}. One can study such conjectures by putting both models
in a common ground. Although some attempts have been made in that direction in specific cases \cite{JT23Towards},
the Euclidean approach should give much more powerful tools.
Another interesting question concerns an extension of conformal nets of operator algebras to Riemann surfaces.
Our work makes it evident that such an extension must be done in the Euclidean geometry,
extending \cite{Schlingemann99Euclidean}. For this, it would be crucial to have a functional integral measure \cite{GJ87}
that gives a CFT, cf.\! \cite{GKR24ProbabilisticLiouville}.

This paper is organized as follows.
In Section \ref{preliminaries} we present our setting.
After briefly reviewing vector-valued formal series in Section \ref{notations},
we recall the definition of full VOA in Section \ref{full}.
Here we introduce our main assumptions: unitarity, polynomial energy bounds, polynomial spectral density and local $C_1$-cofiniteness.
Under local $C_1$-cofiniteness, one can define correlation functions $S^{\pmb{a}}_n$ on the Riemann sphere $\CP$. We recall that they are symmetric.
After recalling the action of $\PSLtwoC$ on $\CP$ in Section \ref{riemann},
we show that the correlation functions satisfy global conformal invariance in Section \ref{global}.
We then recall the OS axioms in Section \ref{OS}.
The proof of OS axioms is given in Section \ref{proofOS}:
In Section \ref{invariance} we define linear functionals $S^{\pmb{a}}_n$ (denoted by the same symbol)
on compactly supported functions vanishing in a neighbourhood of the set of coinciding points
and prove that it is invariant under $\PSLtwoC$ locally.
In Section \ref{LG} we show that $S^{\pmb{a}}_n$ are bounded by a Schwartz norm under
polynomial energy bounds and polynomial spectral density, and they satisfy the linear growth condition.
In particular, $S^{\pmb{a}}_n$ are tempered distributions.
We show that $\{S^{\pmb{a}}_n\}$ satisfy reflection positivity under unitarity in Section \ref{RP}
clustering under uniqueness of vacuum in Section \ref{clustering}.
We exhibit some examples satisfying all our assumptions in Section \ref{examples}.
They are full VOAs that extend Heisenberg algebras as chiral components, known as Narain CFTs, parametrized by
a quotient of the orthogonal group.
We conclude the paper with some outlook in Section \ref{outlook}.

\section{Preliminaries}\label{preliminaries}
\subsection{Notations}\label{notations}
We assume that the base field of vector spaces is $\bbC$ unless otherwise stated. 
Throughout this paper, $z, \z, w, \bar w$ are independent formal variables,
while $\ze, \omega$ are complex numbers and $\zee, \bar \omega$ are their complex conjugate.

We will use the notation $\uz$ for the pair $(z,\z)$ and $|z|^2$ for $z\z$.
For a vector space $V$,
we denote by $V[[z^\bbR,\z^\bbR]]$ the set of formal sums 
$$\sum_{r,s \in \bbR} a_{r,s}z^{r} \z^{s},$$
and
by $V[[z,\z,|z|^\bbR]]$
the subspace of $V[[z^\bbR,\z^\bbR]]$
such that
\begin{itemize}
\item
$a_{r,s}=0$ unless $r-s \in \bbZ$.
\end{itemize}
We also denote by
$V((z,\z,|z|^\bbR))$ the subspace of $V[[z,\z,|z|^\bbR]]$ spanned by the series satisfying
\begin{itemize}
\item
There exists $N \in \bbR$ such that
$a_{r,s}=0$ unless $r,s \geq N$;
\item
For any $H\in \bbR$,
$$\{(r,s) \;|\; a_{r,s}\neq 0 \text{ and }r+s \leq H \}$$
is a finite set.
\end{itemize}

We use the roman $\ri$ for the imaginary unit, while the italic $i$ for a general index.
Similarly, the roman $\re$ stands for Napier's number, while the italic $e$ stands for other things.

\subsection{Full vertex operator algebras, unitarity and polynomial energy bounds}\label{full}

\subsubsection{Full vertex operator algebra}

Let $F=\bigoplus_{h,\h\in \bbR}F_{h,\h}$ be an $\bbR^2$-graded vector space
and $L(0),\Ld(0):F\rightarrow F$ linear maps defined by
$L(0)|_{F_{h,\h}}=h \id_{F_{h,\h}}$ and 
$\Ld(0)|_{F_{h,\h}}=\h \id_{F_{h,\h}}$
for any $h,\h\in \bbR$.
We assume that:
\begin{enumerate}
\item[FO1)]
$F_{h,\h}=0$ unless $h-\h \in \bbZ$;
\item[FO2)]
$F_{h,\h}=0$ unless $h\geq 0$ and $\h \geq 0$;
\item[FO3)]
For any $H\in \bbR$,
$\bigoplus_{h+\h \leq H}F_{h,\h}$ is finite-dimensional.
\end{enumerate}
Set 
\begin{align*}
F^\vee =\bigoplus_{h,\h\in\bbR} F_{h,\h}^*,
\end{align*}
where $F_{h,\h}^*$ is the dual vector space.

A  \textbf{full vertex operator} on $F$ is a linear map
\begin{align*}
Y(\bullet, z,\z):F \rightarrow \mathrm{End}(F)[[z,\z,|z|^\bbR]],\; a\mapsto Y(a,z,\z)=\sum_{r,s \in \bbR}a(r,s)z^{-r-1}\z^{-s-1}
\end{align*}
such that:
\begin{align}
\begin{split}
[L(0),Y(a,\uz)]&= \frac{d}{dz}Y(a,\uz) + Y(L(0)a,\uz),\\
[\Ld(0),Y(a,\uz)]&= \frac{d}{d\z}Y(a,\uz) + Y(\Ld(0)a,\uz).
\label{eq_L0_cov}
\end{split}
\end{align}
Then, by (FO1), (FO2) and (FO3), $Y(a,\uz)b \in F((z,\z,|z|^\bbR))$ (see \cite[Proposition 1.5]{Moriwaki21}).
It is possible to substitute the vector $a \in F$ by
a formal series $\sum_{r', s'} a'_{r',s'} (z')^{r'} (z')^{s'}$ with
\textit{another set of formal variables} $z', \z', |z'|^\bbR$.
Then the result is a formal series in $z,\z,|z|^\bbR, z',\z',|z'|^\bbR$.

By \eqref{eq_L0_cov} (for more detail, see \cite[Lemma 1.6]{Moriwaki21}), for $u \in F_{h_0,\h_0}^\vee$ and $a_i \in F_{h_i,\h_i}$ we have
\begin{align}
u(Y(a_1,\uz_1)Y(a_2,\uz_2)a_3) \in z_2^{h_0-h_1-h_2-h_3}\z_2^{\h_0-\h_1-\h_2-\h_3}\bbC\Bigl(\Bigl(\frac{z_2}{z_1},\frac{\z_2}{\z_1},\Bigl|\frac{z_2}{z_1}\Bigr|^\bbR\Bigr)\Bigr),\label{eq_conv_rad2}\\
u(Y(Y(a_1,\uz_0)a_2,\uz_2)a_3) \in z_2^{h_0-h_1-h_2-h_3}\z_2^{\h_0-\h_1-\h_2-\h_3}\bbC\Bigl(\Bigl(\frac{z_0}{z_2},\frac{\z_0}{\z_2},\Bigl|\frac{z_0}{z_2}\Bigr|^\bbR\Bigr)\Bigr),
\label{eq_conv_rad}
\end{align}
where the left-hand side of \eqref{eq_conv_rad} is a formal series in $z_0,\z_0,|z_0|^\bbR, z_2,\z_2,|z_2|^\bbR$
but contains only terms of the form $z_0^n z_2^{-n}, \z_0^m \z_2^{-m}, z^s\z_0^s \z_2^{-s}\z_2^{-s}$ up to the factor
in front.

A \textbf{full vertex algebra} is an $\bbR^2$-graded $\bbC$-vector space
$F=\bigoplus_{h,\h \in \bbR^2} F_{h,\h}$ equipped with a
full vertex operator $Y(\bullet,\uz):F \rightarrow \mathrm{End}(F)[[z^\pm,\z^\pm,|z|^\bbR]]$
and an element $\vac \in F_{0,0}$ satisfying the following conditions:
\begin{enumerate}
\item[FV1)]
For any $a \in F$, $Y(a,\uz)\vac \in F[[z,\z]]$ and $\displaystyle{\lim_{\uz \to 0}Y(a,\uz)\vac = a(-1,-1)\vac=a}$.
\item[FV2)]
$Y(\vac,\uz)=\mathrm{id} \in \End F$;
\item[FV3)]
For any $a_i \in F_{h_i,\h_i}$ and $u \in F_{h_0,\h_0}^*$, \eqref{eq_conv_rad2} and \eqref{eq_conv_rad} are absolutely convergent in $\{|\zeta_1|>|\zeta_2|\}$ and $\{|\zeta_0|<|\zeta_2|\}$, respectively, and there exists a real analytic function $\mu: Y_2(\bbC)\rightarrow \bbC$ such that:
\begin{align*}
u(Y(a,\uze_1)Y(b,\uze_2)c) &= \mu(\zeta_1,\zeta_2)|_{|\zeta_1|>|\zeta_2|}, \\
u(Y(Y(a,\uze_0)b,\uze_2)c) &= \mu(\zeta_0+\zeta_2,\zeta_2)|_{|\zeta_2|>|\zeta_0|},\\
u(Y(b,\uze_2)Y(a,\uze_1)c)&=\mu(\zeta_1,\zeta_2)|_{|\zeta_2|>|\zeta_1|}
\end{align*}
where $Y_2(\bbC)=\{(\zeta_1,\zeta_2)\in \bbC^2\mid \zeta_1\neq \zeta_2,\zeta_1\neq 0,\zeta_2\neq 0\}$.
\end{enumerate}

Let $F$ be a full vertex algebra
and $D$ and $\D$ denote the endomorphism of $F$
defined by $Da=a(-2,-1)\vac$ and $\D a=a(-1,-2)\vac$ for $a\in F$,
i.e., $$Y(a,z)\vac=a+Daz+\D a\z+\dots.$$

Then, we have (see \cite[Proposition 3.7, Lemma 3.11, Lemma 3.13]{Moriwaki21}):
\begin{proposition}
\label{prop_translation}
For $a \in F$, the following properties hold:
\begin{enumerate}
\item
$Y(Da,\uz)=\frac{d}{dz} Y(a,\uz)$ and $Y(\D a,\uz)=\frac{d}{d\z} Y(a,\uz)$;
\item
$D\vac=\D\vac=0$;
\item
$[D,\D]=0$;
\item
$Y(a,\uz)b=\exp(zD+\z\D)Y(b,-\uz)a$;
\item
$Y(\D a,\uz)=[\D,Y(a,\uz)]$ and $Y(Da,\uz)=[D,Y(a,\uz)]$.
\item
If $\D a=0$, then for any $n\in \bbZ$ and $b\in F$,
\begin{align*}
[a(n,-1),Y(b,\uz)]&= \sum_{j \geq 0} \binom{n}{j} Y(a(j,-1)b,\uz)z^{n-j},\\
Y(a(n,-1)b,\uz)&= 
\sum_{j \geq 0} \binom{n}{j}(-1)^j a(n-j,-1)z^{j}Y(b,\uz) \\
&\qquad -Y(b,\uz)\sum_{j \geq 0} \binom{n}{j}(-1)^{j+n} a(j,-1)z^{n-j}.
\end{align*}
\item
If $\D a =0$ and $D b=0$, then $[Y(a,\uz),Y(b,\uz)]=0$.
\end{enumerate}
\end{proposition}

Proposition \ref{prop_translation} implies $\ker D$ is a vertex algebra and $\ker \D$ is also a vertex algebra with the formal variable $\z$.

A \textbf{full vertex operator algebra} is a pair of a full vertex algebra and distinguished vectors
$\nu \in F_{2,0}$ and $\bar\nu\in F_{0,2}$ such that
\begin{enumerate}[{FVOA}1{)}]
\item
$\D \nu=0$ and $D \bar\nu=0$;
\item
There exist scalars $c, \bar{c} \in \bbC$ such that
$\nu(3,-1)\nu=\frac{c}{2} \vac$,
$\bar\nu(-1,3)\bar\nu=\frac{\bar{c}}{2} \vac$ and
$\nu(k,-1)\nu=\bar\nu(-1,k)\bar\nu=0$
for any $k=2$ or $k\in \bbZ_{\geq 4}$.
\item
$\nu(0,-1)=D$ and $\bar\nu(-1,0)=\D$;
\item
$\nu(1,-1)|_{F_{h,\h}}=h$ and
$\bar\nu(-1,1)|_{F_{h,\h}}=\h$ for any $h,\h \in \bbR$.
\item
For any $H >0$, $\bigoplus_{h+\h<H} F_{h,\h}$ is finite-dimensional.
\end{enumerate}

We remark that $\{ \nu(n,-1)\}_{n\in \bbZ}$ and $\{ \bar\nu(-1,n)\}_{n\in \bbZ}$
satisfy the commutation relation of Virasoro algebra by Proposition \ref{prop_translation}.
Then, we have \cite[Proposition 3.18 and Proposition 3.19]{Moriwaki21}:
\begin{proposition}\label{ker_hom}
Let $(F,\nu,\bar\nu)$ be a full vertex operator algebra.
Then, $(\ker \D,\nu)$ and $(\ker D,\bar\nu)$ are vertex operator algebras and the linear map
\begin{align*}
\ker \D \otimes \ker D \rightarrow F, a\otimes b \mapsto a(-1,-1)b
\end{align*}
is a full vertex operator algebra homomorphism.
Moreover, $F$ is a $\ker \D \otimes \ker D$-module and the vertex operator $Y(\bullet,\uz)$ is an intertwining operator of $\ker \D \otimes \ker D$-modules.
\end{proposition}

Let $V$ be a vertex operator algebra and $M$ a $V$-module.
For any $n \in \bbZ_{>0}$, set 
\begin{align*}
C_n(M)= \{a(-n)m\mid m\in M \text{ and }a \in\bigoplus_{k\geq 1}V_k \}.
\end{align*}
A $V$-module $M$ is called \textbf{$C_n$-cofinite} if $M/C_n(M)$ is a finite-dimensional vector space.

Since $(L(-1)a)(-n)=na(-n-1)$ for any $a \in V$ and $n\in \bbZ_{>0}$,
$C_{n+1}(M) \subset C_{n}(M)$. Hence, if $M$ is $C_{n+1}$-cofinite,
then $M$ is $C_{n}$-cofinite.
Note that any vertex operator algebra is of itself $C_1$-cofinite.

\subsubsection{Unitarity of full VOA}

An invariant bilinear form of a vertex operator algebra was introduced by \cite{Li94SymmetricInvariant}.
A unitary vertex operator algebra was introduced by Dong and Lin \cite{DL14Unitary} and \cite{CKLW18}
(cf.\! \cite[Section 8.3]{Gui19Unitarity2}).

Let $F$ be a full vertex operator algebra.
A \textbf{bilinear form} $(\bullet,\bullet):F\otimes F \rightarrow \bbC$ is called \textbf{invariant} if
\begin{align}
(u,Y(a,\uz)v) = (Y(\exp(L(1)z+\Ld(1)\z) (-1)^{L(0)-\Ld(0)}z^{-2L(0)}\z^{-2\Ld(0)} a,\uz^{-1})u,v)
\label{eq_bilinear_inv}
\end{align}
for any $a,u,v \in F$. Here, since $L(0)-\Ld(0) \in 2\bbZ$ on $F$, $ (-1)^{L(0)-\Ld(0)}$ is well-defined.

The following proposition \cite[Proposition 3.2 and Corollary 3.2]{Moriwaki20Consistency} is an analogue of the result in the chiral case \cite{Li94SymmetricInvariant}:
\begin{proposition}
\label{prop_inv_bilinear}
Let $F$ be a simple full vertex operator algebra. Assume
\begin{enumerate}
\item
$F_{0,0}=\bbC$;
\item
$F_{h,\h}=0$ if $h<0$ or $\h <0$;
\item
$L(1)F_{1,0}=0$ and $\Ld(1)F_{0,1}=0$;
\item
$L(-1)F_{0,n}=0$ and $\Ld(-1)F_{n,0}=0$ for any $n \geq 0$.
\end{enumerate}
Then, there exists a unique (up to constant) non-degenerate invariant bilinear form on $F$.
\end{proposition}

Let $(F, Y,\vac,\nu,\bar\nu)$ be a full vertex operator algebra. An anti-linear automorphism $\phi$ of $F$ is an
anti-linear map $\phi:F \rightarrow F$ such that $\phi(\vac) = \vac, \phi(\nu)=\nu, \phi(\bar\nu)=\bar\nu$ and 
$\phi(a(r,s)b) = \phi(a)(r,s)\phi(b)$ for any $a,b\in F$ and $r,s \in \bbR$.

Let $(F, Y,\vac,\nu,\bar\nu)$ be a full vertex operator algebra with an invariant bilinear form $(\bullet,\bullet)$
and $\phi:F \rightarrow F$ be an anti-linear involution, i.e.\! an anti-linear automorphism of order 2. The pair $(F,\phi)$ is called \textbf{unitary} if the sesquilinear form $\<\bullet,\bullet \> = (\phi(\bullet),\bullet)$ is positive-definite.
For a unitary full vertex operator algebra, we will normalize the invariant form $(\bullet, \bullet)$ on $F$ by $(\vac,\vac)=1$.

By invariance of $(\bullet, \bullet)$, for any $a, u, v \in F$ we have
\begin{align}
\<u,Y(\phi(a),\uz)v\>=\<Y(\exp(L(1)z+\Ld(1)\z) (-1)^{L(0)-\Ld(0)}z^{-2L(0)}\z^{-2\Ld(0)} a,\uz^{-1})u,v\>,
\label{eq_Hermite_formal}
\end{align}
where $z$ and $\z$ are regarded as formal variables, and in particular,
$\<z \bullet, \bullet\> = z\<\bullet, \bullet\>$ even though the scalar product is
anti-linear on the left.

Set $F_\bbR = \{a\in F\mid \phi(a)=a\}$, which is a real subalgebra of $F$. Then, the restriction of $\<\bullet,\bullet\>$ on $F_\bbR$ is a real-valued invariant bilinear form of the (real) full VOA $F_\bbR$.
It is easy to show that $\<\bullet, \bullet\>$ is positive-definite if and only if
the restriction of $(\bullet, \bullet)$ to $F_\bbR$ is positive-definite.

For $h,\h \in \bbR$, set
\begin{align*}
\mathrm{QF}_{h,\h}= \{a\in F_{h,\h}\mid L(1)a=\Ld(1)a=0\}.
\end{align*}
A vector in $\mathrm{QF}_{h,\h}$ is called a \textbf{quasi-primary} vector (of conformal weight $(h,\h)$).
We call a vector $a \in F_{h,\h}$ \textbf{Hermite} if $\phi(a)=(-1)^{h-\h} a$.
Here, $h - \bar h \in \bbZ$ by assumption, hence $(-1)^{h - \bar h} = \pm 1$,
depending on whether $h - \bar h$ is even or odd. Let $\mathrm{QHF}_{h,\h}$ be the real vector space of quasi-primary Hermite vectors of conformal weight $(h,\h)$.

If $a$ is a quasi-primary vector, then the invariance property reduces to the following:
\begin{align}\label{eq:hermite}
 (u,Y(a,\uz)v)= (-1)^{h-\bar h}z^{-2h}\z^{-2\bar h}(Y(a,\uz^{-1})u,v).
\end{align}

For a full VOA satisfying the hypotheses of Proposition \ref{prop_inv_bilinear}, we can take a set of vectors $a \in \mathrm{QF}_{h,\h}$ for
some $h, \h \in \bbR$ that generate $F$ (in the sense of Lemma \ref{lemma_quasiprimary}).
As $\phi$ is an (antilinear) automorphism, it commutes with all $L(m), \Ld(m)$,
and hence $\phi$ maps $\mathrm{QF}_{h,\h}$ to $\mathrm{QF}_{h,\h}$.
We can modify the generating set of quasi-primary vectors with the following way:
for $a$ in that set,
\begin{itemize}
 \item if $h - \bar h$ is even, then we take $a + \phi(a)$ and $\ri(a - \phi(a))$. They are both in $\mathrm{QF}_{h,\h}$.
 \item if $h - \bar h$ is even, then we take $\ri(a + \phi(a))$ and $a - \phi(a)$. They are both in $\mathrm{QF}_{h,\h}$.
\end{itemize}
Let us call this modified set of homogeneous quasi-primary vectors $\cQ$.
It is clear that $\cQ$ still generates $F$. Furthermore, if $a \in \cQ \cap \mathrm{QF}_{h,\h}$, then we have
\begin{align}\label{eq:hermite-sp}
 \<u,Y(a,\uz)v\>=  \<Y(z^{-2h}\z^{-2\bar h} a,\uz^{-1})u,v\> =  z^{-2h}\z^{-2\bar h} \<Y(a,\uz^{-1})u,v\>.
\end{align}
Note that the equation \eqref{eq:hermite-sp} holds as formal power series,
however, one need to be careful when one evaluates the formal variable $z$ by some complex number $\ze \in \bbC^\times$,
since $\<\bullet, \bullet\>$ is anti-linear on the first vector (see \eqref{eq:conjugatephi} below).

\begin{lemma}
\label{lemma_quasiprimary}
Let $F$ be a full VOA satisfying the assumptions in Proposition \ref{prop_inv_bilinear}. Then,
\begin{align*}
F = \bigoplus_{h,\h\in \bbR} \bbC[L(-1),\Ld(-1)]\mathrm{QHF}_{h,\h}.
\end{align*}
\end{lemma}
\begin{proof}
Under the assumptions in Proposition \ref{prop_inv_bilinear}, by
\cite[Proposition 3.5]{Moriwaki20Consistency},
\begin{align*}
F = \bigoplus_{h,\h\in \bbR} \bbC[L(-1),\Ld(-1)]\mathrm{QF}_{h,\h}.
\end{align*} 
For any $a\in \mathrm{QF}_{h,\h}$,
\begin{itemize}
 \item if $h - \bar h$ is even, then we take $a + \phi(a)$ and $\ri (a - \phi(a))$. They are both in $\mathrm{QHF}_{h,\h}$.
 \item if $h - \bar h$ is odd, then we take $\ri (a + \phi(a))$ and $a - \phi(a)$. They are both in $\mathrm{QHF}_{h,\h}$.
\end{itemize}
Hence, the assertion holds.
\end{proof}

\begin{proposition}
\label{prop_unitary_spec}
Let $F$ be a unitary full VOA. Then, the following conditions hold:
\begin{enumerate}
\item
$F_{h,\h}=0$ if $h<0$ or $\h<0$;
\item
$L(1)F_{1,0}=0$ and $\Ld(1)F_{0,1}=0$;
\item
$\ker L(-1)= \bigoplus_{n\geq 0}F_{0,n}$ and $\ker \Ld(-1)= \bigoplus_{n\geq 0}F_{n,0}$.
\end{enumerate}
\end{proposition}
\begin{proof}
For any $v\in \mathrm{QF}_{h,\h}$, we have
\begin{align}
h||v|| = \<v,L(0)v\>= 1/2\<v,[L(1),L(-1)]v\>= 1/2 ||L(-1)v|| \geq 0, \label{eq_L0_pos}
\end{align}
which implies (1) (see \cite[Proposition 1.7]{Gui19Unitarity1}).
$\ker L(-1) \subset \bigoplus_{n\geq 0}F_{0,n}$ holds for any full VOA,
and the equality holds by \eqref{eq_L0_pos}.
For $v\in F_{1,0}$,
$0 \leq \<L(1)v,L(1)v\> = \<v,L(-1)L(1)v\>=0$ by (3). Hence, (2) holds.
\end{proof}

Recall that an \textbf{ideal} of a full VOA $F$ is a subspace $I \subset F$ such that
\begin{align*}
a(r,s)v \in I \text{ for any }a \in F, v\in I\text{ and }r,s\in \bbR.
\end{align*}
Since $L(0)$ and $\Ld(0)$ act on $I$, $I$ is an $\bbR^2$-graded subspace.
If $I$ is an ideal of $F$, then, by Proposition \ref{prop_translation} (4), $v(r,s)a \in I$ for any $a\in F$, $v\in I$ and $r,s\in \bbR$.
Thus, a left ideal is automatically two-sided ideal (see \cite[Section 3]{Moriwaki21}).
A full VOA $F$ is called \textbf{simple} if $F$ does not have any proper ideal.
\begin{proposition}
\label{prop_unitary_vacuum}
Let $F$ be a unitary full VOA. Then, $F$ is simple if and only if $F_{0,0}=\bbC\vac$.
\end{proposition}
\begin{proof}
Recall that there is a bijection between invariant bilinear forms on $F$ and
\begin{align}
\mathrm{Hom}_\bbC(F_{0,0} / L(1)F_{1,0}+\Ld(1)F_{0,1},\bbC) \label{eq_bilinears_all}
\end{align}
by \cite[Proposition 3.2]{Moriwaki20Consistency}. By Proposition \ref{prop_unitary_spec}, \eqref{eq_bilinears_all} is just $F_{0,0}^*$
and $F_{0,0}=\ker L(-1) \cap \ker \Ld(-1)$.
Since $\Ld(-1)$ is a VOA by Proposition \ref{ker_hom},
$F_{0,0}$ is a unital commutative associative $\bbC$-algebra by
\begin{align*}
F_{0,0}\otimes F_{0,0} \rightarrow F_{0,0},\quad a\otimes b \mapsto a(-1,-1)b.
\end{align*}

First, assume that $F$ is simple.
Let $J \subset F_{0,0}$ be an ideal of $\bbC$-algebra such that $J \neq F_{0,0}$.
Take a non-zero dual vector $p:F_{0,0} \rightarrow \bbC$ such that $J \subset \ker (p)$ and let $(-,-)_p:F\otimes F \rightarrow \bbC$ be the unique associated invariant bilinear form in \cite[Proposition 3.2]{Moriwaki20Consistency},
which satisfies $(\vac, a)_p= p(a)$ for any $a \in F_{0,0}$. 
Then, by the invariance,  for $a\in F_{0,0}$ and $v \in J$, 
\begin{align*}
(a,v)_p = (\vac, a(-1,-1)v)_p=p(a(-1,-1)v) =0.
\end{align*}
Hence, the radical of the bilinear form $I_p = \{v \in F\mid (a,v)_p=0 \text{ for any }a\in F \}$ contains $J$.
By the invariance, $I_p \subsetneq F$ is an ideal of a full VOA. Hence, $I_p=0$ by the assumption. Hence, $J =0$ and $F_{0,0}$ is a simple commutative $\bbC$-algebra,
thus, one-dimensional.

Next, assume that $F_{0,0}=\bbC\vac$.
Let $I \subset F$ be an ideal such that $I \neq F$.
Let $a\in F_{h,\h}$ and $v\in I \cap F_{h,\h}$ with $(h,\h)\neq (0,0)$.
Then, by $I = \bigoplus_{\substack{h,\h\geq 0\\ (h,\h)\neq (0,0)}} I \cap F_{h,\h}$,
\begin{align*}
(a,v) = (a(-1,-1)\vac,v) = \sum_{n,m\geq 0} \frac{(-1)^{h-\h}}{n!m!}(\vac, (L(1)^n\Ld(1)^m a) (h-n-1,\h-m-1) v) =0.
\end{align*}
Hence, $v=0$, and thus, $I=0$, i.e., $F$ is simple. 
\end{proof}

\subsubsection{Correlation functions of full vertex operator algebra}
Let us denote $\zeta_{[n]}$ the $n$-tuple $(\zeta_1, \cdots, \zeta_n)$.
Similarly, $a_{[n]} \in F^n$.
The variables $z,\z$ in the vertex operator are formal variables, and the arguments of the correlation function $C_n(u,a_{[n]};\zeta_{[n]})$ are complex numbers. Often it is important to distinguish between them, so in this paper we will use $\zeta$ when we use complex numbers and $z,\z$ when we think they are formal. Denote by $\zee$ the complex conjugate of $\ze \in \bbC$ and set $\ze_{i,j} = \ze_i-\ze_j$, throughout of this paper.

In the definition of full vertex operator algebra, we assumed that the compositions of two vertex operators are convergent to the same real analytic function regardless of the orders and parentheses. In general, it is nontrivial whether the composition of $n$ vertex operators converges. If it converges, it is a physical quantity called a {\it correlation function}, which is a real analytic function on the configuration space
\begin{align*}
X_n(\bbC)=\{(\zeta_1,\dots,\zeta_n)\in\bbC^n \mid \zeta_i \neq \zeta_j \text{ for }i\neq j\}.
\end{align*}
We have shown the existence and several important properties of the correlation functions under the assumption of \textit{local $C_1$-cofiniteness}. In this section we briefly review the results of \cite{Moriwaki24SwissCheese} necessary for this paper.
(The existence of the correlation function was shown by Huang-Kong under the assumption that a VOA is rational $C_2$-cofinite \cite{HK07FullFieldAlgebras}.)

Let $V,W$ be vertex operator algebras.
A full vertex operator algebra is called \textbf{locally $C_1$-cofinite} over $V$ and $W$
if there are $C_1$-cofinite $V$-modules $M_i$
and $C_1$-cofinite $W$-modules $\overline{M}_i$ indexed by some countable set $I$ such that:
\begin{itemize}
\item
$V$ is a subalgebra of $\ker L(-1)$ and $W$ is a subalgebra of $\ker \Ld(-1)$;
\item
$F$ is isomorphic to $\bigoplus_{i \in I} M_i\otimes \overline{M}_i$ as a $V\otimes W$-module;
\item
For any $i,j \in I$, there exists finite subset $I(i,j) \subset I$ such that:
\begin{align}
Y(\bullet,\uz)\bullet \in \bigoplus_{i,j \in I} \bigoplus_{k \in I(i,j)} I\binom{M_k}{M_iM_j} \otimes I\binom{{\overline{M}}_k}{{\overline{M}}_i{\overline{M}}_j},
\label{eq_int_C1_decomp}
\end{align}
where $I\binom{M_k}{M_iM_j}$ and $I\binom{{\overline{M}}_k}{{\overline{M}}_i{\overline{M}}_j}$ are the space of intertwining operators of $V$ and $W$, respectively.
\end{itemize}

Let $\frS_n$ be the symmetric group. For $\si \in \frS_n$, set
\begin{align*}
U_n^\si= \{(\zeta_1,\dots,\zeta_n)\in X_n(\bbC) \mid |\zeta_{\si (1)}|>|\zeta_{\si (2)}|>\cdots >|\zeta_{\si (n)}| \},
\end{align*}
which is an open domain in $X_n(\bbC)$.
Denote $U_n^{\iota}$ by $U_n$ for short, where $\iota \in \frS_n$ is the unit element.
Then, we have \cite[Theorem 3.11 and Corollary 3.8]{Moriwaki24SwissCheese}:
\begin{theorem}
\label{thm_A}
Let $(F,Y,\vac,\nu,\bar\nu)$ be a full vertex operator algebra and assume that $F$ is locally $C_1$-cofinite over some vertex operator algebras.
Then, for any $u\in F^\vee$ and $a_{[n]} = (a_1,\cdots, a_n)\in F^{n}$, the following power series
\begin{align}
\langle u, Y(a_1,\uze_1)Y(a_2,\uze_2)\dots Y(a_r,\uze_r)\vac \rangle,
\label{eq_full_cor_n}
\end{align}
which is given by replacing formal variables $z, \bar z$ by complex numbers $\zeta, \bar \zeta$
in the formal series $\langle u, Y(a_1,\uz_1)Y(a_2,\uz_2)\dots Y(a_r,\uz_r)\vac \rangle$,
is absolutely and locally uniformly convergent in $U_n$.
Moreover, there is a unique family of linear maps for $n\geq 1$
\begin{align}
C_n:F^\vee\otimes F^{n} \rightarrow C^{\om} (X_n(\bbC)), \label{eq_Cor_n_def}
\end{align}
where $C^{\om} (X_n(\bbC))$ is a space of real analytic functions on $X_n(\bbC)$,
such that:
\begin{align*}
C_n(u,a_{[n]};\zeta_{[n]})\Bigl|_{U_{n}^\si}=\langle u, Y(a_{\si (1)},\uze_{\si (1)})Y(a_{\si (2)},\uze_{\si (2)})\dots Y(a_{\si( n)},\uze_{\si (n)})\vac \rangle
\end{align*}
for any $u\in F^\vee$, $a_{[n]} \in F^{n}$ and $\sigma \in \frS_n$ as real analytic functions.
\end{theorem}

We have shown in \cite[Theorem 3.11]{Moriwaki24SwissCheese} that for any parentheses and orders of compositions of the vertex operators, the formal power series converges on some explicitly given open domain in $X_n(\bbC)$, all of which are series expansions in different domains of a single real analytic function $C_n$ \eqref{eq_Cor_n_def}.
The following proposition is one of such examples
and will be used to show the cluster decomposition property in Section \ref
{clustering} (see \cite[Theorem 3.11 and Proposition 1.8]{Moriwaki24SwissCheese}):
\begin{proposition}
\label{prop_cluster_tree}
For any $m,n>0$, let
\begin{align*}
{U}_{m,n}
&=\left\{(\ze_1,\dots,\ze_{m+n})\in X_{m+n}(\bbC) \middle\vert
\begin{array}{l}
|\ze_{i+1,m}|<|\ze_{i,m}|, |\ze_{j+1,m+n}|<|\ze_{j,m+n}|, \\
|\ze_{1,m}|+|\ze_{m+1,m+n}|<|\ze_{m,m+n}|\\
\text{ for }1 \leq i \leq m-2, m+1 \leq j\leq m+n-2
\end{array}\right \}. 
\end{align*}
Assume that $F$ is locally $C_1$-cofinite. Then, the following formal power series
\begin{align*}
\Big\langle u,\exp(L(-1)z_{m+n}+\Ld(-1)\z_{m+n})
Y\Bigl(Y(a_1,\uz_{1,m})Y(a_2,\uz_{2,m})\dots Y(a_{m-1},\uz_{m-1,m})a_m,\uz_{m,m+n}\Bigr)\quad&\\
Y(a_{m+1},\uz_{m+1,m+n})\dots Y(a_{m+n-1},\uz_{m+n-1,m+n})a_{m+n} \Big\rangle&
\end{align*}
is absolutely and locally uniformly convergent in $U_{m,n}$ for any $u\in F^\vee$ and $a_1,\dots,a_{m+n} \in F$
and coincides with the expansion of $C_{m+n}(u,a_{[m+n]};\zeta_{[m+n]})$ on $U$ after substituting $z_{i,j}$ (resp. $\z_{i,j}$) with $\ze_i-\ze_j$ (resp. $\zee_i-\zee_j$).
\end{proposition}

We also use the following results \cite[Theorem 3.11]{Moriwaki24SwissCheese}:
\begin{theorem}
\label{thm_B}
Under the assumption of Theorem \ref{thm_A}, the family of linear maps $C_n$ satisfies
\begin{description}
\item[(Symmetry)]
For any permutation $\sigma \in \frS_n$, $u \in F^\vee$ and $a_1,\dots,a_n \in F$,
\begin{align*}
C_n(u,a_n,\dots,a_n;\ze_1,\dots,\ze_n)=C_n(u,a_{\sigma(1)},\dots,a_{\sigma(n)};\ze_{\sigma(1)},\dots,\ze_{\sigma(n)}).
\end{align*}
\item[(Vacuum)]
For any $u \in F^\vee$ and $a_1,\dots,a_n \in F$,
\begin{align*}
C_{n+1}(u,a_1,\dots,a_n,\vac;\ze_1,\dots,\ze_n,\ze_{n+1})=C_n(u,a_1,\dots,a_n;\ze_1,\dots,\ze_n).
\end{align*}
\item[(Infinitesimal conformal covariance)]
\begin{align*}
C_n(u,L(-1)_i a_{[n]};\ze_{[n]}) &= \frac{d}{dz_i} C_n(u, a_{[n]};\ze_{[n]})\\
C_n(u,\Ld(-1)_i a_{[n]};\ze_{[n]}) &= \frac{d}{d\z_i} C_n(u, a_{[n]};\ze_{[n]})\\
C_n(L(-1)^* u,a_{[n]};\ze_{[n]}) &= \sum_{i=1}^n C_n(u, L(-1)_i a_{[n]};\ze_{[n]})\\
C_n(\Ld(-1)^* u,a_{[n]};\ze_{[n]}) &= \sum_{i=1}^n C_n(u, \Ld(-1)_i a_{[n]};\ze_{[n]})\\
C_n(L(0)^* u,a_{[n]};\ze_{[n]}) &= \sum_{i=1}^n C_n(u, (\ze_i\frac{d}{dz_i}+L(0)_i) a_{[n]};\ze_{[n]})\\
C_n(\Ld(0)^* u,a_{[n]};\ze_{[n]}) &= \sum_{i=1}^n C_n(u, (\zee_i\frac{d}{d\zee_i}+\Ld(0)_i) a_{[n]};\ze_{[n]})\\
C_n(L(1)^* u,a_{[n]};\ze_{[n]}) &= \sum_{i=1}^n C_n(u, (\ze_i^2\frac{d}{dz_i}+\ze_iL(0)_i+L(1)_i) a_{[n]};\ze_{[n]})\\
C_n(\Ld(1)^* u,a_{[n]};\ze_{[n]}) &= \sum_{i=1}^n C_n(u, (\zee_i^2\frac{d}{d\zee_i}+\zee_i\Ld(0)_i+\Ld(1)_i) a_{[n]};\ze_{[n]})
\end{align*}
\end{description}
\end{theorem}

Hereafter, we assume that a full vertex operator algebra is
\begin{itemize}
\item
simple;
\item
locally $C_1$-cofinite;
\item
unitary with anti-linear involution $\phi:F \rightarrow F$.
\end{itemize}

In this case, 
by Proposition \ref{prop_unitary_vacuum}, $F_{0,0}=\bbC\vac$.
Let $\<\vac| \in F_{0,0}^\vee$ be the unique dual vector such that $\<\vac|\vac \>=1$.
Then, by Proposition \ref{prop_unitary_spec},
\begin{align*}
L(1)F_{1,0}+\Ld(1)F_{0,1}=0,
\end{align*}
which implies that
\begin{align}
L(n)^* \<\vac| = \Ld(n)^*\<\vac|=0 \label{eq_dual_vac}
\end{align}
for all $n =-1,0,1$.

One of the most important observable in quantum field theory is the \textit{n-point correlation function} (or the \textit{vacuum expectation value}) is defined as
\begin{align*}
S_n^{\pmb{a}}(\pmb{\zeta})=C_{[n]}(\<\vac|,a_{[n]},\ze_{[n]}),
\end{align*}
where we use the following notation:
\begin{align*}
\pmb{a}=(a_1,\cdots,a_n), \qquad \pmb{\zeta}=(\zeta_1,\cdots,\zeta_n).
\end{align*}

By the infinitesimal conformal invariance  and \eqref{eq_dual_vac}, we have:
\begin{corollary}[infinitesimal conformal invariance]
\label{cor_inf_conf}
Assume $a_{i} \in F_{h_i,\h_i}$ ($i=1,2,\dots,n$) are quasi-primary. Then,
\begin{align*}
(\sum_{i=1}^n \frac{d}{d\ze_i}) S_n^{\pmb{a}}(\pmb{\zeta})&=0\\
 (\sum_{i=1}^n \frac{d}{d\zee_i})S_n^{\pmb{a}}(\pmb{\zeta})&=0\\
 (\sum_{i=1}^n \ze_i\frac{d}{d\ze_i}+h_i) S_n^{\pmb{a}}(\pmb{\zeta})&=0\\
 (\sum_{i=1}^n \zee_i\frac{d}{d\zee_i}+\h_i) S_n^{\pmb{a}}(\pmb{\zeta})&=0\\
 (\sum_{i=1}^n \ze_i^2\frac{d}{d\ze_i}+h_i\ze_i) S_n^{\pmb{a}}(\pmb{\zeta})&=0\\
 (\sum_{i=1}^n \zee_i^2\frac{d}{d\zee_i}+\h_i\zee_i) S_n^{\pmb{a}}(\pmb{\zeta})&=0.
\end{align*}
\end{corollary}

\subsubsection{Energy bounds and spectral density}
For a unitary VOA, an estimate of the field norm, called polynomial energy bound, was introduced in \cite{CKLW18}.

Let $(F,\<\bullet,\bullet\>)$ be a unitary full VOA.
\begin{itemize}
\item[(PEB)]
We say that $F$ satisfies \textbf{polynomial energy bounds} if the following condition holds:
For any $a\in F$, there exist positive integers $p_a, q_a$ and a constant $M_a > 0$ such that, for all $r,s\in \bbR$ and all $b\in F$
\begin{align}\label{eq:PEB}
||a(r,s)b|| \leq M_a(|r|+|s| + 1)^{p_a} ||(L(0)+\Ld(0)+\bb1)^{q_a} b||.
\end{align}
\item[(PSD)]
We say that $F$ satisfies \textbf{polynomial spectral density} if
there exists $L$ such that for any $n\in \bbZ_{\geq 0}$,
\begin{align}\label{eq:density}
\#\{(h + \h \in \bbR^2 : N \leq h + \h < N+1, F_{h,\h}\neq 0 \} < C(N+2)^L.
\end{align}
\end{itemize}

\begin{lemma}\label{lm:PEBestimate}
 Let us assume (PEB) and let $\Upsilon \subset F$ be a finite set.
 Then for any $n$ and $\pmb{a} = (a_1,\cdots a_n) \in \Upsilon^n$, there are $M_\Upsilon, Q_\Upsilon > 0$ such that
 \begin{align*}
  &|\<\vac, a_1(r_1, s_1) \cdots a_n(r_n, s_n)\vac\>| \\
  &\le M_\Upsilon^n (|r_1 + s_1|+1 +\cdots |r_n + s_n|+1)^{nQ_\Upsilon} (|r_1 - s_1|+1 +\cdots |r_n - s_n|+1)^{nQ_\Upsilon}
 \end{align*}
\end{lemma}
\begin{proof}
 We calculate, using $|r| + |s| \le |r-s| + |r+s|$,
 \begin{align}
  & |\<\vac, a_1(r_1, s_1) \cdots a_n(r_n, s_n)\vac\>| \nonumber \\
  &\le \|a_1(r_1, s_1) \cdots a_n(r_n, s_n)\vac\| \nonumber \\
  &\le M_{a_1}(|r_1|+|s_1|+1)^{p_{a_1}} \|(L_0+\bar L_0 + \bb1)^{q_{a_1}}a_1(r_2, s_2) \cdots a_n(r_n, s_n)\vac\| \nonumber \\
  &\le M_{a_1}(|r_1|+|s_1|+1)^{p_{a_1}} (|r_2 + \cdots + r_2| + |s_2 + \cdots + s_n| + 1)^{q_{a_1}} 
  \|a_2(r_2, s_2) \cdots a_n(r_n, s_n)\vac\| \nonumber \\
  &\le M_\Upsilon^n \prod_{j=1}^n (|r_j|+|s_j|+1)^{p_\Upsilon} (|r_j+\cdots +r_j| + |s_1+\cdots+s_n|+1)^{q_\Upsilon} \nonumber \\
  &\le M_\Upsilon^n (|r_1|+|s_1|+1 +\cdots |r_n|+|s_n|+1)^{nq_\Upsilon} \prod_{j=1}^n (|r_j|+|s_j|+1)^{p_\Upsilon} \nonumber \\
  &\le M_\Upsilon^n (|r_1 + s_1|+1 +\cdots |r_n + s_n|+1)^{nQ_\Upsilon} (|r_1 - s_1|+1 +\cdots |r_n - s_n|+1)^{nQ_\Upsilon} \label{eq:fullFourier}
\end{align}
\end{proof}

Let us remark that there are cases where our technical conditions can be checked by
looking at the chiral components.
\begin{proposition}
Let $F$ be a full vertex operator algebra. Assume that the canonical vertex operator subalgebras $\ker L(-1)$ and $\ker \Ld(-1)$ are $C_2$-cofinite. Then, $F$ satisfies the polynomial spectral density and the local $C_1$-cofiniteness.
\end{proposition}
\begin{proof}
Note that a $C_2$-cofinite vertex operator algebra has only finitely many irreducible modules \cite{Zhu96ModularInvariance}.
Thus, $F$ is finitely generated module over $\ker \Ld(-1)\otimes \ker L(-1)$.
In particular, there are finitely many real numbers $(\Delta_i,\bar{\Delta}_i) \in \bbR^2$ $(i=1,\dots,N)$ such that
\begin{align*}
F = \bigoplus_{i=1}^N \bigoplus_{n,m\in \bbZ_{\geq 0}}F_{n+\Delta_i,m+\bar{\Delta}_i}.
\end{align*}
Hence, $\#\{(h,\h) \in \bbR^2\mid n \leq h \leq n+1, m \leq \h \leq m+1, F_{h,\h}\neq 0 \} \leq N$
and from this it is straightforward that $F$ satisfies the polynomial spectral density.
It is shown in \cite{ABD04Rationality} that any finitely generated module of a $C_2$-cofinite vertex operator algebra is $C_2$-cofinite,
and thus, $C_1$-cofinite. Hence, $F$ satisfies the local $C_1$-cofiniteness.
\end{proof}

On the other hands, with the current techniques, polynomial energy bounds (of intertwining operators)
need to be checked case by case, cf.\! \cite{Gui19, CT23EnergyBounds}.

\subsection{The Riemann sphere and the stereographic projection}\label{riemann}
Until now, we considered the correlation functions defined in a unitary full VOA under local $C_1$-cofiniteness
on $\bbC$.
Under the identification $\bbC \cong \bbR^2$, the Riemann sphere $\CP = \bbC \cup \{\infty\}$ includes the two-dimensional Euclidean space $\bbR^2$.
Moreover, the connected component $\euclidean$ of the unit element of the Euclidean group $\rE(2)$ can be considered as a subgroup
of the conformal group $\PSLtwoC$. The latter acts on $\CP$ by the linear fractional transformations
\begin{align}\label{eq:lft}
 \left(\begin{array}{cc} a & b \\ c & d \end{array}\right) \cdot \zeta = \frac{a\zeta + b}{c\zeta + d}.
\end{align}

In particular, $\euclidean$ is generated by
\begin{itemize}
 \item Translations: for $a \in \bbC$,
 $\zeta \mapsto \left(\begin{array}{cc} 1 & a \\ 0 & 1 \end{array}\right) \cdot \zeta = \zeta + a$.
 \item Rotations: for $\lambda \in \bbR/2\pi\bbZ$,
 $\zeta \mapsto \left(\begin{array}{cc} \re^{\ri\lambda/2} & 0 \\ 0 & \re^{-\ri\lambda/2} \end{array}\right) \cdot \zeta = \re^{\ri\lambda}\zeta$.
\end{itemize}

On the other hand, $\PSLtwoC$ in addition contains
\begin{itemize}
 \item Dilations: for $\lambda \in \bbR$,
 $\zeta \mapsto \left(\begin{array}{cc} \re^{\lambda/2} & 0 \\ 0 & \re^{-\lambda/2} \end{array}\right) \cdot \zeta = \re^{\lambda}\zeta$.
 \item NS-dilations: for $\lambda \in \bbR$,
 $\zeta \mapsto \left(\begin{array}{cc} \cosh (\lambda/2) & -\sinh (\lambda/2) \\ -\sinh(\lambda/2) & \cosh (\lambda/2) \end{array}\right) \cdot \zeta
 = \frac{\zeta \cosh(\lambda/2) - \sinh(\lambda/2)}{-\zeta \sinh(\lambda/2) + \cosh(\lambda/2)}$.
\end{itemize}
We call the last one-parameter group ``NS-dilations'' for the following reason (cf.\! \cite[Section 3.1.5]{Rychkov17EPFL}:
The Riemann sphere $\CP$ is mapped to the sphere $S^2 \subset \bbR^3$ by the stereographic projection
\begin{align*}
 \zeta = \tau + \ri\xi \mapsto \left(\frac{2\tau}{1 + \tau^2 + \xi^2}, \frac{2\xi}{1 + \tau^2 + \xi^2}, \frac{- 1 + \tau^2 + \xi^2}{1 + \tau^2 + \xi^2}\right).
\end{align*}
Under this map, the point $0, \infty$ are mapped to $(0,0,-1), (0,0,1) \in \bbR$, respectively,
and the region $|\zeta| > 1$ is mapped to the upper hemisphere.
It is more convenient to compose this with
$\left(\begin{array}{cc} \frac1{\sqrt2} & -\frac1{\sqrt2} \\ \frac1{\sqrt2} & \frac1{\sqrt2} \end{array}\right)$
which maps three points $1, \infty, -1$ to $0, 1, \infty$ and the region $|\zeta| > 0$ to
the region $\tau > 0$. It holds that
\begin{align*}
 \left(\begin{array}{cc} \cosh (\lambda/2) & -\sinh (\lambda/2) \\ -\sinh(\lambda/2) & \cosh (\lambda/2) \end{array}\right) =
 \left(\begin{array}{cc} \frac1{\sqrt2} & -\frac1{\sqrt2} \\ \frac1{\sqrt2} & \frac1{\sqrt2} \end{array}\right)^{-1}
 \left(\begin{array}{cc} \re^{\lambda/2} & 0 \\ 0 & \re^{-\lambda/2} \end{array}\right)
 \left(\begin{array}{cc} \frac1{\sqrt2} & -\frac1{\sqrt2} \\ \frac1{\sqrt2} & \frac1{\sqrt2} \end{array}\right),
\end{align*}
so that it gives the dilation flowing from the North Pole to the South Pole.

\subsection{Global conformal invariance}\label{global}
In Corollary \ref{cor_inf_conf}, we considered the infinitesimal conformal invariance.
Since the global conformal group $\rS\rO_e(3,1)$ does not act on $X_n(\bbC)$, we need to treat this global action more carefully.
Set
\begin{align*}
X_n(\CP) =\{(\ze_1,\dots,\ze_n)\in (\CP)^n\mid \ze_i\neq \ze_j\}.
\end{align*}
Then, the global conformal symmetry $\rS\rO_e(3,1) \cong \PSLtwoC = \mathrm{SL}_2 (\bbC)/\{\pm 1\}$ acts on $X_n(\CP)$ by the linear fractional transformation. 
The aim of this section is to show that
\begin{align}
  S_n^{\pmb{a}}(\ze_1, \cdots, \ze_n) &= \prod_{j=1}^n  \left(\frac{d\gamma_j}{d\ze_j}(\ze_j)\right)^{h_{a_j}} \left(\overline{\frac{d\gamma_j}{d\ze_j}(\ze_j)}\right)^{\bar h_{a_j}} 
  S_n^{\pmb{a}}({\gamma(\ze_1)}, \cdots, {\gamma(\ze_n)}) \nonumber \\
  &=\prod_{j=1}^n  \left(c\ze_j+d\right)^{-2h_{a_j}} \left(c\zee_j+d)\right)^{-2\bar h_{a_j}} 
  S_n^{\pmb{a}}({\gamma(\ze_1)}, \cdots, {\gamma(\ze_n)}), \label{eq_fractional}
 \end{align}
where $h_a, \bar h_a$ are conformal dimensions of the quasi-primary vector $a$.
To state this claim we prepare a function space on $X_r(\CP)$.

Let $h_1,\h_1,\dots,h_n,\h_n \in \bbR_{\geq 0}$ with $h_i-\h_i \in \bbZ$.
Denote\footnote{The upper index $\omega$ indicates real analyticity and has nothing to do with
variables $\omega_i$ that appear below.}
by $C_{\pmb{h}}^{\om}(X_n(\CP))$ the vector space consisting of continuous functions $f$ on $X_n(\CP)$ such that:
\begin{itemize}
\item
$f$ is a real analytic function on $X_n(\bbC) \subset X_n(\CP)$;
\item
For any $i \in \{1,\cdots, n\}$ and $\al = (\al_1,\dots, \al_n) \in X_n(\CP)$ with $\al_i = \infty$,
the function $\ze_i^{-2h_i}\zee_i^{-2\h_i} f(\ze_1,\dots,\ze_i,\dots,\ze_n)$ has a real analytic continuation at $\alpha$.
\end{itemize}

\begin{remark}
\label{rem_analytic_inf}
Take a real analytic coordinate 
\begin{align*}
(\omega_1,\dots,\omega_n) = (\ze_1-\al_1,\dots,\ze_{i-1}-\al_{i-1}, \ze_i^{-1},\ze_{i+1}-\al_{i+1},\dots,\ze_n-\al_n)
\end{align*}
around $\al$. The second condition is equivalent to $f$ having the following convergent expansion around $\al$
\begin{align*}
\omega_i^{2h_i}\bar{\omega}_i^{2\h_i} f(\omega_1+\al_1,\dots,\omega_{i-1}+\al_{i-1},\omega_i^{-1},\omega_{i-1}+\al_{i-1},\dots,\omega_n+\al_n) 
\in \bbC[[\omega_1,\bar{\omega}_1,\dots,\omega_n,\bar{\omega}_n]].
\end{align*}
\end{remark}

Let $\ga=\begin{pmatrix}
a & b \\
c & d \\
\end{pmatrix}
 \in \mathrm{SL}_2 (\bbC)$.
Note that if $h-\h \in \bbZ$, then
\begin{align*}
\left(c\ze+d\right)^{-2h} \left(c\zee+d\right)^{-2\bar h} &=
|c\ze+d|^{-2h}\left(c\zee+d\right)^{2(h-\h)}=\left(-c\ze-d\right)^{-2h} \left(-c\zee-d\right)^{-2\bar h}.
\end{align*}
Hence, $\left(c\ze+d\right)^{-2h} \left(c\zee+d\right)^{-2\bar h}$ is well-defined for the element of $\PSLtwoC$.
\begin{lemma}
\label{lem_action_global}
For any $f\in C_{\pmb{h}}^{\om}(X_n(\CP))$ and $\ga \in \mathrm{SL}_2 (\bbC)$,
\begin{align*}
\ga \cdot_{\pmb{h}} f=
\prod_{j=1}^n  \left(c\ze_j+d\right)^{-2h_{j}} \left(c\zee_j+d)\right)^{-2\bar h_{j}}f({\gamma(\ze_1)}, \cdots, {\gamma(\ze_n)})
\end{align*}
is in $C_{\pmb{h}}^{\om}(X_n(\CP))$.
\end{lemma}
\begin{proof}
First, assume that $c=0$. Then, $\ga \cdot f =f(a\ze+b)$ is clearly real analytic on $X_n(\bbC)$.
Since $\ze=\frac{\omega^{-1}-b}{a}$ is a local coordinate at $\zeta=\infty$,
\begin{align*}
(\omega_i^{-1}-b)^{-2h_i}(\bar{\omega_i}^{-1}-b)^{-2\h_i} f(\dots,\omega_i^{-1},\dots)
=(1-b\omega_i)^{-2h_i}(1-b\omega_i)^{-2\h_i} \omega_i^{2h_i} \bar{\omega}_i^{2\h_i}f(\dots,\omega_i^{-1},\dots)
\end{align*}
is real analytic by Remark \ref{rem_analytic_inf}.

Next, assume that $c\neq 0$ and consider the case of $\al \in X_n(\bbC)$ with $\al_i = -\frac{d}{c}$.
Take the local coordinate $\ze=\frac{dw^{-1}-b}{-cw^{-1}+a}$ at $\ze=-\frac{d}{c}$, since
\begin{align*}
&(c\ze_i+d)^{-2h_i}(c\zee_i+d)^{-2\h_i} \cdots f(\cdots,\textstyle{\frac{a\ze_i+b}{c\ze_i+d}},\dots) \\
&=(-c+a\omega_i)^{-2h_i}(-c+a\bar{\omega}_i)^{-2\h_i}\omega_i^{2h_i} \bar{\omega}_i^{2\h_i}f(\dots,\omega_i^{-1},\dots),
\end{align*}
$\ga \cdot_{\pmb{h}} f$ is real analytic on $X_n(\bbC)$.
Finally, consider the case of $\al \in X_n(\CP)$ with $\al_i=\infty$.
Since
\begin{align*}
&\ze_i^{2h_i}\zee_i^{2\h_i}(c\ze_i+d)^{-2h_i}(c\zee_i+d)^{-2\h_i} \cdots f(\cdots,\textstyle{\frac{a\ze_i+b}{c\ze_i+d}},\dots) \\
&=(c+d\ze_i^{-1})^{-2h_i}(c+d\zee_i^{-1})^{-2\h_i} f(\dots, \textstyle{\frac{a+b\ze_i^{-1}}{c+d\ze_i^{-1}}},\dots),
\end{align*}
the assertion holds.
\end{proof}

\begin{proposition}[Global conformal invariance]
\label{prop_global_conf}
Suppose that $a_i$ are Hermite and quasi-primary vectors.
Then, the correlation function $S_n^{\pmb{a}}(\ze_1, \cdots, \ze_n)$ is in $C_{\pmb{h_a}}^{\om}(X_n(\CP))$ and 
is invariant with respect to the $\cdot_{\pmb{h_a}}$ action of $\PSLtwoC$, that is,
\begin{align}\label{eq:invarianceC-function}
  S_n^{\pmb{a}}(\ze_1, \cdots, \ze_n) &=\prod_{j=1}^n  \left(\frac{d\gamma_j}{d\ze_j}(\ze_j)\right)^{h_{a_j}} \left(\overline{\frac{d\gamma_j}{d\ze_j}(\ze_j)}\right)^{\bar h_{a_j}} S_n^{\pmb{a}}({\gamma(\ze_1)}, \cdots, {\gamma(\ze_n)})
  \end{align}
  for any $\ga \in \PSLtwoC$.
\end{proposition}

To prove the above proposition, we will use the following lemma:
\begin{lemma}
\label{lem_inverse_cor}
Assume that $a_i$ are Hermite and quasi-primary.
Then the following identity holds as real analytic functions:
\begin{align*}
S_n^{\pmb{a}}(\ze_1, \cdots, \ze_n) = \prod_{j=1}^n (-1)^{h_j-\h_j}\ze_j^{-2h_j}\zee_j^{-2\h_j}
S_n^{\pmb{a}}(\ze_1^{-1}, \cdots, \ze_n^{-1}).
\end{align*}
\end{lemma}
\begin{proof}
First, by noting that $z, \bar z$ as formal variables, $\<zu,v\> = z\<u,v\>$,
\begin{align*}
&\<\vac, Y(a_1,\uz_1) Y(a_2,\uz_2)\dots Y(a_n,\uz_n) \vac\> \\
&=(\phi\vac, Y(a_1,\uz_1) Y(a_2,\uz_2)\dots Y(a_n,\uz_n) \vac)  & \text{ (definition of } \<\bullet, \bullet\>\text{)}\\
&=(\vac, Y(a_1,\uz_1) Y(a_2,\uz_2)\dots Y(a_n,\uz_n) \vac)  & (\phi\vac = \vac) \\
&=(Y(a_1,\uz_1) Y(a_2,\uz_2)\dots Y(a_n,\uz_n) \vac, \vac)  & \text{(symmetry of } (\bullet, \bullet) \text{)} \\
&= (-1)^{h_1 - \h_1} z_1^{-2h_1} \z_1^{-2\h_1} (Y(a_2,\uz_2)\dots Y(a_n,\uz_n)\vac, Y(a_1,\uz_1^{-1})\vac) & \text{(} a_1 \text{ is Hermite)}\\
&\cdots\\
&=\prod_{j=1}^n (-1)^{h_j-\h_j} z_j^{-2h_j} \z_j^{-2\h_j} \cdot (\vac, Y(a_n,\uz_n^{-1})\dots Y(a_1,\uz_1^{-1})\vac) & \text{(} a_j \text{ are Hermite)}\\
&=\prod_{j=1}^n (-1)^{h_j-\h_j} z_j^{-2h_j} \z_j^{-2\h_j} \cdot \<\vac, Y(a_n,\uz_n^{-1})\dots Y(a_1,\uz_1^{-1})\vac\>
  & \text{(}\phi\vac = \vac, \text{definition of } \<\bullet, \bullet\>\text{)}\\
\end{align*}

Now we evaluate the equality with the complex numbers $\zeta_1, \cdots, \zeta_n$ satisfying
$|\zeta_1| > \cdots > |\zeta_n|$ and obtain
\begin{align*}
 S_n^{\pmb{a}}(\ze_1, \cdots, \ze_n)
 &= \prod_{j=1}^n (-1)^{h_j-\h_j}\ze_j^{-2h_j}\zee_j^{-2\h_j} \cdot S_n^{\sigma\pmb{a}}(\ze_n^{-1}, \cdots, \ze_1^{-1}) \\
 &= \prod_{j=1}^n (-1)^{h_j-\h_j}\ze_j^{-2h_j}\zee_j^{-2\h_j} \cdot S_n^{\pmb{a}}(\ze_1^{-1}, \cdots, \ze_n^{-1}),
\end{align*}
where $\sigma (1\cdots n) = (n\cdots 1)$ and we used the symmetry of $S^{\pmb{a}}_n$ in the second equality.
Since $\{(\ze_1,\dots,\ze_n)\in X_n(\bbC)\mid |\ze_1|>\cdots>|\ze_n|\}$ is an open subset of $X_n(\bbC)$, the assertion follows from the identity theorem for real analytic functions.
\end{proof}

\begin{proof}[proof of Proposition \ref{prop_global_conf}]
By Lemma \ref{lem_inverse_cor}, $S_n^{\pmb{a}}(\ze_1, \cdots, \ze_n)$
 is in $C_{\pmb{h_a}}^{\om}(X_n(\CP))$.
Hence, it suffices to show that the invariance holds for the generators of $\mathrm{PSL}_2 (\bbC)$.
By Corollary \ref{cor_inf_conf}, $S_n^{\pmb{a}}(\ze_1, \cdots, \ze_n)$ is invariant
under translations, rotations and dilations, that is, for any $\ga = \begin{pmatrix}
a & b \\
0 & d \\
\end{pmatrix} \in \bbC \rtimes \bbC \subset \PSLtwoC$.
In particular, for $\ga=\begin{pmatrix}
i & 0 \\
0 & -i \\
\end{pmatrix}$, $S_n^{\pmb{a}}(\ze_1, \cdots, \ze_n) = \prod_{j=1}^n (-1)^{h_j-\h_j}
S_n^{\pmb{a}}(-\ze_1, \cdots, -\ze_n)$ holds.
Combining this with Lemma \ref{lem_inverse_cor}, we have
\begin{align*}
S_n^{\pmb{a}}(\ze_1, \cdots, \ze_n) = \prod_{j=1}^n \ze_j^{-2h_j}\zee_j^{-2\h_j}
S_n^{\pmb{a}}(-\ze_1^{-1}, \cdots, -\ze_n^{-1}).
\end{align*}
Hence, the assertion holds for $\ga=\begin{pmatrix}
0& 1 \\
-1 & 0 \\
\end{pmatrix} \in \PSLtwoC$. Since $\PSLtwoC$ is generated by translations, rotations, dilations and $\begin{pmatrix}
0& 1 \\
-1 & 0 \\
\end{pmatrix}$, the assertion holds.
\end{proof}

\subsection{OS axioms on \texorpdfstring{$\bbR^2$}{R2}}\label{OS}
We follow the conventions of \cite{OS73, OS75} for general Euclidean fields (the original papers are written for the four-dimensional case,
but it is straightforward to adapt them to two dimensions),
and \cite{FFK89, Schottenloher08} for things specific to two-dimensional CFT\footnote{Unfortunately, the ``axioms'' of \cite{FFK89, Schottenloher08}
do not include regularity conditions corresponding to ($\check E0$) or ($E0'$), in particular,
they are not strong enough to reconstruct Wightman fields.}, up to some changes of notations.

A Euclidean quantum field field theory can be formulated in terms of Schwinger functions satisfying
the Osterwalder-Schrader axioms \cite{OS73, OS75}.
We are interested in two-dimensional models, and we identify $\bbR^2 = \bbC$, which is natural for conformal field theory.
We denote $\zeta = \tau+\ri\xi \in \bbC$, or equivalently $\zeta = (\tau, \xi) \in \bbR^2$.
For a test function $f$ on $\bbC^n = \bbR^{2n}$ and a multi-index $\alpha = (\alpha_1, \cdots, \alpha_{2n})$,
we denote
\begin{align*}
 \partial^\alpha f := \frac{\partial^{|\alpha|}}{\partial \tau_1^{\alpha_1}\partial \xi_1^{\alpha_2}\cdots \partial \tau_n^{\alpha_{2n-1}}\partial \xi_n^{\alpha_2n}}f,
\end{align*}
where $\alpha_k \ge 0, |\alpha| = \sum \alpha_k$.

Let $\Upsilon$ be a set of indices, which will be eventually any finite set of Hermite quasi-primary fields in the theory\footnote{Below,
to compare with \cite{OS73}, one can take $\Upsilon$ as a set with a single element $a$ and assume that $h_a = \bar h_a$,
then it is a theory with a single scalar field, and one can ignore all the indices $\pmb{a}$, leaving only
the number of variables $n$. With $\Upsilon$ we just introduce more fields with non-zero spin.}.
Any such $a \in \Upsilon$ is associated with $h_a, \bar h_a > 0$ called \textbf{conformal dimensions}
(they are possibly different positive numbers, not the complex conjugate of each other) where $h_a - \bar h_a \in \bbZ$.
An $n$-tuple of indices with $n \in \bbN$ is denoted by $\pmb{a} = (a_1, \cdots, a_n)$ and we denote
the set of such $n$-tuples $\Upsilon_n$.

Apart from this, following \cite{OS73}, we introduce
\begin{itemize}
 \item For $\pmb{a} \in \Upsilon_n$, let $\scS^{\pmb{a}}(\bbR^{2n})$ the space of Schwartz functions (smooth and rapidly decaying).
 For this definition, there is no distinction for different $\pmb{a}$, but we will consider different actions
 of $\PSLtwoC$.
 \item For $f \in \scS^{\pmb{a}}(\bbR^{2n})$, define the Schwartz norms
 \begin{align}\label{eq:schwartznorms}
  |f|_p = \underset{|\alpha| \le p}{\sup_{\pmb{\zeta} \in \bbR^{2n}}} |(1+|\pmb{\zeta}|^2)^{\frac p2}(\partial^\alpha f)(\pmb{\zeta})|,
 \end{align}
 where $\pmb{\zeta} = (\zeta_1,\cdots, \zeta_n) \in \bbR^{2n}$ and $|\pmb{\zeta}|^2 = \sum_{j=1}^n |\zeta_j|^2,
 |\zeta_j| = \sqrt{\tau_j^2 + \xi_j^2}$ for $\zeta_j = \tau_j + \ri\xi_j \in \bbC$.
 \item $\scS^{\pmb{a}}_{\neq}(\bbR^{2n}) = \{f \in \scS(\bbR^{2n}): \partial^\alpha f(\zeta_1, \cdots\!, \zeta_n) = 0 \text{ for all } \alpha
 \text{ if } \zeta_j = \zeta_k \text{ for some } j < k\}$.
 (in \cite{OS73} the notation $^0\scS$ is used. We take this notation from \cite{Simon74}.)
 \item For $-\infty \le s_1 < s_2 \le \infty$,
 \[
  \scS^{\pmb{a}}_{s_1 s_1}(\bbR^{2n}) = \{f \in \scS^{\pmb{a}}(\bbR^{2n}): \partial^\alpha f(\zeta_1, \cdots\!, \zeta_n) = 0 \text{ for all } \alpha
  \text{ unless } s_1 < \tau_1 < \cdots < \tau_n < s_2\}
 \]
 \item $\scS^{\pmb{a}}_+(\bbR^{2n}) = \scS^{\pmb{a}}_{0,\infty}(\bbR^{2n}), \scS^{\pmb{a}}_<(\bbR^{2n}) = \scS^{\pmb{a}}_{-\infty, \infty}(\bbR^{2n})$.
 \item For $\scS^{\pmb{a}}_*$ any of these spaces of test functions, let $(\scS^{\pmb{a}}_*)'$ be the continuous dual. 
 \item For $\gamma \in \PSLtwoC$, we define the local action on $f$ with compact support\footnote{This is defined
as long as $\gamma$ does not send any point in $\supp f$ to $\infty$. In particular, for any $f$ with compact support,
there is a neighbourhood of the unit element in $\PSLtwoC$ with that property, and the action is defined there.}:
\begin{align}\label{eq:fgamma}
f_\gamma(\zeta_1, \cdots, \zeta_n) := \prod_{j=1}^n  \left(\frac{d\gamma^{-1}}{d\zeta}(\zeta_j)\right)^{-h_{a_j}}
 \left(\overline{\frac{d\gamma^{-1}}{d\zeta}(\zeta_j)}\right)^{-\bar h_{a_j}} \left|\frac{d\gamma^{-1}}{d\zeta}(\zeta_j)\right|
f(\gamma^{-1}(\zeta_1),\ldots, \gamma^{-1}(\zeta_n)). 
\end{align}
This can be extended to all $f \in \scS^{\pmb{a}}(\bbR^{2n})$ if $\gamma \in \euclidean$,
as such $\gamma$ does not take any point $\bbR^2$ to $\infty$.

\end{itemize}

We call a set $\{S_n^{\pmb{a}}\}_{n \in \bbN, \pmb{a} \in \Upsilon}$ of distributions $S_n^{\pmb{a}} \in (\scS^{\pmb{a}}_{\neq})'(\bbR^{2n})$
\textbf{Schwinger functions on $\bbR^2$} if they satisfy the following \textbf{Osterwalder-Schrader axioms}:
\begin{enumerate}[{(OS}1{)}]
 \item[(OS\textlabel{0$^\prime$}{ax:lineargrowth})] (Linear growth) $S_0 = 1$ (for $n=0$ there is no upper index)
 and there exist $s \in \bbN, C_1, C_2 > 0$ such that
 \begin{align}\label{eq:lineargrowth}
  |S_n^{\pmb{a}}(f)| \le C_1 (n!)^{C_2} |f|_{sn}
 \end{align}
 for all $n \in \bbN, f \in \scS^{\pmb{a}}_{\neq}(\bbR^{2n}), \pmb{a} \in \Upsilon_n$.
 \item\label{ax:invariance} (Euclidean invariance)
 For any $\gamma \in \euclidean$ and $f \in \scS^{\pmb{a}}_{\neq}(\bbR^{2n})$, it holds that
 \begin{align*}
  S_n^{\pmb{a}}(f_\gamma) = S_n^{\pmb{a}}(f)
 \end{align*}
 \item\label{ax:RP} (Reflection positivity)
 For any finite set $A$ of finite sequences\footnote{Here the lower index $n$ is not distinguishing the elements of $A$, but rather
 indicating the length of the sequence $(a_{n,1}, \cdots, a_{n,n})$.
 In other words, elements of $A$ are written as $\pmb{a}_n, \pmb{b}_m...$ and so on.}
 $\pmb{a}_n = (a_{n,1}, \cdots, a_{n,n}) \in \Upsilon_n, n \in \bbN$ and $\{f_n^{\pmb{a}_n}\}, f_n^{\pmb{a}_n} \in \scS^{\pmb{a}_n}_+(\bbR^{2n})$,
 it holds that
 \begin{align}\label{eq:RP}
  \sum_{\pmb{a}_n, \pmb{b}_m, \in A} S_{m+n}^{(\phi\pmb{b}_m, \pmb{a}_n)}(\Theta (f_m^{\pmb{b}_m})^* \otimes f_n^{\pmb{a}_n}) \ge 0,
 \end{align}
  where\footnote{We choose $b_j$ Hermite, so it is natural that $\phi$ has no effect on single $b$'s.} $\phi\pmb{b}_m = (b_{m,m}, \cdots, b_{m,1})$
  and $(\phi\pmb{b}_m, \pmb{a}_n)$ is the finite sequence in $\Upsilon_{m+n}$ given by concatenation and
 \begin{align*}
  f^*(\zeta_1, \cdots, \zeta_n) &= \overline{f(\zeta_n, \cdots, \zeta_1)} \quad \text{(note the inversion of the variables)}\\
  \Theta f(\zeta_1, \cdots, \zeta_n) &= f(\theta \zeta_1, \cdots, \theta \zeta_n), \quad \text{with } \theta \zeta = -\bar \zeta
 \end{align*}

 \item\label{ax:symmetry} (Symmetry) For any $\sigma \in \frS_n$ where $\frS_n$ is the symmetric group with $n$ elements, it holds that
 $S_n^{\sigma\pmb{a}}(f^\sigma) = S_n^{\pmb{a}}(f)$,
 where $f^\sigma(\zeta_1, \cdots, \zeta_n) = f(\zeta_{\sigma^{-1}(1)}, \cdots, \zeta_{\sigma^{-1}(n)})$
 and $\sigma\pmb{a} = (a_{\sigma^{-1}(1)}, \cdots, a_{\sigma^{-1}(n)})$.
 \item\label{ax:clustering} (Clustering)
 For any finite sets $A, B$ of finite sequences of indices and finite sets of test functions $\{f^{\pmb{a}_m}_m\},\{g^{\pmb{b}_n}_n\},
 f^{\pmb{a}_m}_m \in \scS^{\pmb{a}_m}_+(\bbR^{2m}), g^{\pmb{b}_n}_n \in \scS^{\pmb{b}_n}_+(\bbR^{2n})$,
 indices $\pmb{a}_m \in A, \pmb{b}_n \in B$ as in (OS\ref{ax:RP})
 and $y \in \euclidean$ is a translation of the form $y = (0, y_2) \in \bbR^2$, it holds that
  \[
   \lim_{\lambda \to \infty} \sum_{\pmb{a}_m \in A, \pmb{b}_n \in B}
   S_{m+n}^{(\phi\pmb{a}_m, \pmb{b}_n)}(\Theta (f^{\pmb{a}_m}_m)^* \otimes g^{\pmb{b}_m}_{n,\lambda y})
   = \sum_{\pmb{a}_m \in A, \pmb{b}_n \in B} S_m^{\phi\pmb{a}_m}(\Theta (f^{\pmb{a}_m}_m)^*)S_n^{\pmb{b}_n}(g^{\pmb{b}_n}_n),
  \]
  where $\phi\pmb{a}_m$ is as above.
 \item[(OS\textlabel{C}{ax:invarianceC})] (Conformal invariance)
 For any $f \in \scS^{\pmb{a}}_{\neq}(\bbR^{2n})$ with compact support and $\gamma \in \PSLtwoC$ sufficiently close to
 the unit element such that $\gamma$ does not take any point of $\supp f$ to $\infty$, it holds that
 \begin{align*}
  S_n^{\pmb{a}}(f_\gamma) = S_n^{\pmb{a}}(f)
 \end{align*}
  \setcounter{OS}{\value{enumi}}
\end{enumerate}

\begin{remark}
 (OS\ref{ax:invarianceC}) is a generalization of (OS\ref{ax:invariance}),
 where $\euclidean$ is considered as a subgroup of $\PSLtwoC$ as in Section \ref{riemann}.
 For $\gamma \in \euclidean$, the factor $\frac{d\gamma}{d\zeta}$ simplifies
 ($1$ for translations and $\re^{\ri\lambda}$ for rotations, giving the ``spin'', both constants in $\zeta$).

 Our (OS\ref{ax:invariance})--(OS\ref{ax:clustering}) are natural generalizations of E1--E4 of \cite{OS73} to
 multiple fields with spin, while (OS\ref{ax:lineargrowth}) corresponds to E0$'$ of \cite{OS75}. 
 It should not be confused with (OS1$'$) of \cite{Simon74}, which corresponds to $\check{\rE0}$ of \cite{OS75},
 the assumption that $S_n^{\pmb{a}}$ are Laplace transforms of some distributions.
\end{remark}

\section{Proof of the OS axioms}\label{proofOS}
Hereafter we assume that $F$ is a full VOA satisfying local $C_1$-cofiniteness.

Recall that $Y(a,z,\z)=\sum_{r,s \in \bbR}a(r,s)z^{-r-1}\z^{-s-1}$ is a formal series.
The series
\begin{align}\label{eq:standardregion}
 S_n^{\pmb{a}}(\zeta_1, \cdots, \zeta_n)
 &= \<\vac, Y(a_1, \zeta_1, \zee_1)\cdots Y(a_n, \zeta_n, \zee_n)\vac\> \nonumber \\
 &=
\sum_{\substack{r_2, \cdots, r_n \in \bbR \\ s_2, \cdots, s_n \in \bbR}}
 \<\vac, a_1(-r_2\cdots -r_n, -s_2\cdots -s_n)a_2(r_2,s_2)\cdots a_2(r_n,s_n)\vac\> \nonumber \\
 &\qquad\qquad \times \left(\frac{\zeta_2}{\zeta_1}\right)^{-r_2\cdots -r_n}\left(\frac{\zee_2}{\zee_1}\right)^{-s_2\cdots -s_n}\cdots
 \left(\frac{\zeta_n}{\zeta_{n-1}}\right)^{-r_n}\left(\frac{\zee_n}{\zee_{n-1}}\right)^{-s_n}\frac1{|\zeta_1\cdots \zeta_n|^2}
\end{align}
is convergent if $|\zeta_1| > \cdots > |\zeta_n|$.
By Theorem \ref{thm_A}, this extends to $X_n(\bbC)$ and \eqref{eq:invarianceC-function} holds.

To simplify the notations, we denote $\bbR^{2n}_{\neq} = X_n(\bbC) = \{(\zeta_1, \cdots, \zeta_n) : \zeta_j \neq \zeta_k \text{ for } j\neq k\}$
the set of $n$ ordered distinct points and
$\bbR^{2n}_= :=\bbR^{2n} \setminus \bbR^{2n}_{\neq}$ the set of coinciding points.
For test functions $f \in \scS^{\pmb{a}}_{\neq}(\bbR^2)$
with compact support in $\bbR^{2n}_{\neq}$,
we define $S^{\pmb{a}}_n(f)$ by the integral
\begin{align}\label{eq:defSn}
 \int S^{\pmb{a}}_n(\zeta_1, \cdots, \zeta_n) f(\zeta_1, \cdots, \zeta_n) d\tau_1 d\xi_1\cdots d\tau_n d\xi_n,
\end{align}
where $\zeta = (\tau, \xi)$ is considered as a point in $\bbR^2$.
By this definition, it is clear from Theorem \ref{thm_B} that $\{S^{\pmb{a}}_n\}$ satisfy (OS\ref{ax:symmetry})
for test functions $f = f^{\pmb{a}}_n$ with compact support in $\bbR^{2n}_{\neq}$.

\begin{theorem}\label{th:fulltoOS}
 Let $F$ be a unitary simple full VOA satisfying local $C_1$-cofiniteness, polynomial energy bounds \eqref{eq:PEB} and polynomial spectral density \eqref{eq:density},
 and $\Upsilon$ be a finite family of quasi-primary vectors in $F$.
 Then $\{S^{\pmb{a}}_n\}$ extends to $(\scS^{\pmb{a}}_{\neq})'(\bbR^{2n})$ and satisfy the conformal OS axioms
 (OS\ref{ax:lineargrowth})--(OS\ref{ax:invarianceC}). 
\end{theorem}

The proof extends in the following Sections. It is not that one axiom is proved in one section,
but they are rather interrelated. The proof of (OS\ref{ax:lineargrowth}) is given in Section \ref{LG},
which uses the translation-invariance for $f$ compactly supported in $\bbR^{2n}_{\neq}$.
The proofs of (OS\ref{ax:invarianceC}), (OS\ref{ax:symmetry}) are concluded in Section \ref{LG} as well,
while (OS\ref{ax:RP}) is proved in Section \ref{RP} and (OS\ref{ax:clustering}) is given in Section \ref{clustering}.

\subsection{Conformal invariance for compactly supported functions}\label{invariance}
Let us first prove (OS\ref{ax:invariance}) and (OS\ref{ax:invarianceC}) for test functions compactly supported in $\bbR^{2n}_{\neq}$.
We need this result in the proof of (OS\ref{ax:lineargrowth}), which in turn implies the continuity
of $S^{\pmb{a}}_n(f)$ in $f$, in particular
the invariance for all test functions in $\scS^{\pmb{a}}_{\neq}(\bbR^2)$ follows by continuity.

The conformal invariance for $f$ compactly supported in $\bbR^{2n}_{\neq}$ goes as follows:
writing $\omega = \gamma^{-1}(\zeta) = \upsilon + \ri\eta$, thus $\zeta = \gamma(\omega)$,
\begin{align*}
&S^{\pmb{a}}_n(f_\gamma) \\
&=\int S^{\pmb{a}}_n(\zeta_1, \cdots, \zeta_n) f_\gamma(\zeta_1, \cdots, \zeta_n) d\tau_1 d\xi_1\cdots d\tau_n d\xi_n\\
&= \int S^{\pmb{a}}_n(\zeta_1, \cdots, \zeta_n) \\
&\quad\times\prod_{j=1}^n  \left(\frac{d\gamma^{-1}}{d\zeta}(\zeta_j)\right)^{-h_{a_j}} \left(\overline{\frac{d\gamma^{-1}}{d\zeta}(\zeta_j)}\right)^{-\bar h_{a_j}} \left|\frac{d\gamma^{-1}}{d\zeta}(\zeta_j)\right|
f(\gamma^{-1}(\zeta_1),\ldots, \gamma^{-1}(\zeta_n)) d\tau_1 d\xi_1\cdots d\tau_n d\xi_n\\
&= \int S^{\pmb{a}}_n(\gamma(\omega_1), \cdots, \gamma(\omega_n)) \\
&\quad\times\prod_{j=1}^n  \left(\frac{d\gamma^{-1}}{d\zeta}(\gamma(\omega_j))\right)^{-h_{a_j}}
\left(\overline{\frac{d\gamma^{-1}}{d\zeta}(\gamma(\omega_j))}\right)^{-\bar h_{a_j}} \left|\frac{d\gamma^{-1}}{d\zeta}(\gamma(\omega_j))\right|
\left|\frac{d\gamma}{d\omega}(\omega_j)\right| \\
&\quad\times f(\omega_1,\ldots, \omega_n) d\upsilon_1 d\eta_1\cdots d\upsilon_n d\eta_n\\
&= \int \prod_{j=1}^n  \left(\frac{d\gamma}{dw}(\omega_j)\right)^{h_{a_j}} \left(\overline{\frac{d\gamma}{d\zeta}(\omega_j)}\right)^{\bar h_{a_j}}
S^{\pmb{a}}_n(\gamma(\omega_1), \cdots,\gamma(\omega_n)) f(\omega_1,\ldots, \omega_n) d\upsilon_1 d\eta_1\cdots d\upsilon_n d\eta_n\\
&= \int  S^{\pmb{a}}_n(\omega_1, \cdots, \omega_n) f(\omega_1,\ldots, \omega_n) d\upsilon_1 d\eta_1\cdots d\upsilon_n d\eta_n=S^{\pmb{a}}_n(f),
\end{align*}
where we substitute $f_\gamma$ \eqref{eq:fgamma} in the second equation,
we change the variable $\zeta_j = \gamma(\omega_j)$ with the Jacobian $\prod_{j=1}^n |\frac{d\gamma}{d\omega}(\omega_j)|$ in the third equation
and used the relation $\frac{d\gamma^{-1}}{d\zeta}(\gamma(\omega)) = (\frac{d\gamma}{d\omega}(\omega))^{-1}$ in the fourth equation
and rewrote using \eqref{eq:invarianceC-function} in the fifth equation.

\subsection{Linear growth + extension to distributions}\label{LG}
The linear growth condition (OS\ref{ax:lineargrowth}) is an estimate of $S^{\pmb{a}}_n(f)$
in terms of the Schwartz norm of $f \in \scS^{\pmb{a}}_{\neq}(\bbR^{2n})$ with a good behaviour in $n$.
As all $S^{\pmb{a}}_n$, first defined on functions compactly supported in $\bbR^{2n}_{\neq}$, have finite Schwartz norm,
they extend to distributions in $(\scS^{\pmb{a}}_{\neq})'(\bbR^{2n})$.
Then (OS\ref{ax:symmetry}) and (OS\ref{ax:invarianceC}) follow by continuity.
We will use here the translation invariance of $S^{\pmb{a}}_n(f)$ for $f$ compactly supported in $\bbR^{2n}_{\neq}$.

We prove linear growth as follows.
\begin{itemize}
 \item[(a)] We restrict $f$ to those with compact support in $\bbR^{2n}_{\neq}$.
 With this, we may assume that $\supp f$ has radius $R > 0$
 and has distance (see below, not the Euclidean distance) $\epsilon > 0$ from the set of coinciding points.
 They both depend on $f$.
 This allows us (see the next point) to decompose $\supp f$ into finitely many small pieces,
 each of which can be translated to a set where a simple representation of $S^{\pmb{a}}_n$ is available,
 cf.\! \eqref{eq:standardregion}.
 \item[(b)] We cut $f$ into small pieces by a smooth partition of unity.
 This requires estimates of all derivatives of the bump functions.
 \item[(c)] We estimate each piece in a translated region,
 and collect them together. This gives an estimate of $S^{\pmb{a}}_n(f)$ in terms of a Schwartz norm
 that depends on $R, \epsilon$.
 \item[(d)] We write the above estimates in terms of stronger Schwartz norms of $f$ without $R, \epsilon$.
 The dependence on $n$ satisfies (OS\ref{ax:lineargrowth}).
\end{itemize}
In the course of the proof, we estimate $S^{\pmb{a}}_n(f)$ in various steps
and we reduce the estimate to some parts up to a factor of the form $C_1 (n!)^{C_2}$.
As the actual value of $C_1, C_2$ do not matter but it is only important that
they are independent of $n$, we do not specify these values.
Instead, we mark every point where we extract a factor of the form $C_1 (n!)^{C_2}$ with $\bigstar$.
Every time we update such a factor, we indicate it with $\bigstar'$.
This makes it clear that we can gather them and rename it as $C_1 (n!)^{C_2}$ and obtain estimates in the desired form.

Moreover, for $q$ independent of $n$, we will use the following fact without further remark:
\begin{itemize}
 \item $2^{qn} = (2^q)^n \le C_qn! = \bigstar$,
 \item $(qn)! = qn \cdots (qn-1) \cdots 1 \le q^{qn}(n!)^q \le C_q(n!)^{q+1} = \bigstar$
 \item $(n^2)^{qn} = (n^n)^{2q} \le (n!)^{4q} = \bigstar$ (the Stirling formula)
 \item $a^n \le Cn! = \bigstar$
\end{itemize}

\paragraph{(a) Finding good directions.}
Here we use $y_\bullet$ for points in $\bbR^2$.

Let $f \in \scS^{\pmb{a}}_{\neq}(\bbR^{2n})$ be a test function compactly supported in $\bbR^{2n}_{\neq}$.
For each pair $j\neq k$, the set
\begin{align*}
 \{y_j-y_k \in \bbR^2 : (y_1, \cdots, y_n) \in \supp f\} 
\end{align*}
is compact and does not contain $0$ and there are $\frac{n(n-1)}2$ (thus finitely many) such pairs.
Therefore, there are $\epsilon > 0, R > 0$ such that
\begin{align}\label{eq:epsilon}
 \{\|y_k-y_j\| \in \bbR : (y_1, \cdots, y_n) \in \supp f\} &> \epsilon, \quad \text {for all }j \neq k, \\
 \left\{\sqrt{\sum_{j=1}^n \|y_j\|^2} : (y_1, \cdots, y_n) \in \supp f\right\} &< R. \nonumber
\end{align}
We may assume that $\epsilon < 1 < R$. They clearly depend on $f$.

We need the following variation of a purely geometric lemma in $\bbR^2$ \cite[Lemma II.11]{Simon74}
(the difference is that we specify the constant $c$ explicitly).
\begin{lemma}\label{lm:circle-cover}
 Let $n \in \bbN$. For each $y=(y_1,\ldots, y_n)\in \bbR^{2n}_{\neq}$,
 there is a unit vector $\hat e_y \in S^1 \subset \bbR^2$
 such that
 \begin{align*}
  |\hat e_y \cdot (y_j - y_k)| \ge \frac{\pi}{4n^2}\|y_j - y_k\|, \quad j \neq k,
 \end{align*}
 where $y_j - y_k \in \bbR^2$. 
 Moreover, for $y, y' \in \supp f \subset \bbR^{2n}$ and $\|y'-y\| < \frac{\pi\epsilon}{32n^2}$, it holds that
 \begin{align}\label{eq:preepsilon}
  |\hat e_y \cdot (y'_j - y'_k)| \ge \frac{\pi}{8n^2}\|y'_j - y'_k\|, \quad j \neq k.
 \end{align}
\end{lemma}
\begin{proof}
 Let $\{e_{jk}\}_{1 \le j < k \le n}$ be an arbitrary set of unit vectors in $S^1\subset \bbR^2$.
 For each $e_{jk}$, the union $\mathcal{E}_{jk}$ of two arcs made of unit vectors $\hat e$ such that $|\hat e \cdot e_{jk}| < \frac\pi{4n^2}$
 has measure $4 \cdot \arcsin \frac\pi{4n^2} < \frac{2\pi}{n^2}$.
 Indeed, for every $j,k$, it is the union of two arcs whose vectors span an angle $\vartheta$ with $e_{jk}$,
 where $\frac{\pi}2 - \frac\pi{4n^2} < \sin\pm\vartheta < \frac{\pi}2 + \frac\pi{4n^2}$,  see Figure \ref{fig:slice}.
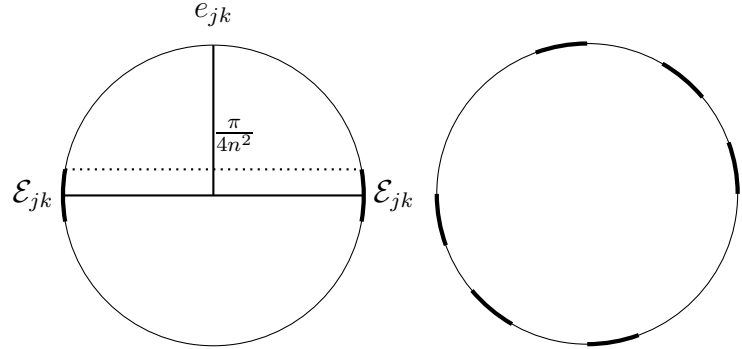
\begin{figure}[ht]\centering
\begin{tikzpicture}[scale=0.4]
  \draw[color=black](0,0) circle (5);
  \node at (0,6) {$e_{jk}$} ;
  \draw[thick] (0,0)--(0,5);
  \draw[thick] (-5,0)--(5,0);
  \draw[dotted, thick] (-4.924,0.8682)--(4.924,0.8682);
  \node at (0.7, 2) {$\frac{\pi}{4n^2}$};
  \draw[ultra thick] (-10:5) arc(-10:10:5);
  \draw[ultra thick] (170:5) arc(170:190:5);
  \node at (6,0) {$\cE_{jk}$} ;
  \node at (-6,0) {$\cE_{jk}$} ;
\end{tikzpicture}
\begin{tikzpicture}[scale=0.4]
  \draw[color=black] (0,0) circle (5);
  \draw[ultra thick] (0:5) arc(0:20:5);
  \draw[ultra thick] (180:5) arc(180:200:5);
  \draw[ultra thick] (40:5) arc(40:60:5);
  \draw[ultra thick] (220:5) arc(220:240:5);
  \draw[ultra thick] (90:5) arc(90:110:5);
  \draw[ultra thick] (270:5) arc(270:290:5);
\end{tikzpicture}
\caption{Left: for a given unit vector $e_{jk}$, the ``bad'' directions $\cE_{jk}$ are indicated.
Right: after excluding ``bad'' directions, there are nonempty intervals of ``good'' directions.}
\label{fig:slice}
\end{figure}

There are $\frac{n(n-1)}2$ such unit vectors $e_{jk}$, thus the union of the arcs in $\mathcal{E}_{jk}$ has the measure not larger than $\pi$. In particular, its complement is nonempty.
We can take $\hat e_y$ from the complement.

To show \eqref{eq:preepsilon}, fix $j,k\in\{1,\ldots,n\}$, $j\neq k$. We consider the map $g_{jk}:\mathbb{R}^{2n}_{\neq}\to S^1 \subset \mathbb{R}^2$ defined as $y'=(y'_1,\ldots,y'_n)\mapsto \frac{y'_j - y'_k}{\|y'_j - y'_k\|}$. As $y'_j - y'_k \neq 0$, the partial derivatives are well defined.
Let $y'_\ell=(s'_\ell,t'_\ell)$ for $\ell=1,\ldots,n$ for $s'_\ell,t'_\ell\in\mathbb{R}$. For every $j\neq \ell\neq k$, then $\partial_{s'_\ell}g_{jk}=0=\partial_{t'_\ell}g_{jk}$. Thus only 4 partial derivatives are non-trivial, namely $\partial_{s'_j}g_{jk}$, $\partial_{t'_j}g_{jk}$, $\partial_{s'_k}g_{jk}$ and $\partial_{t'_k}g_{jk}$.

Due to the high symmetry of $g_{jk}$, without loss of generality, we compute the partial derivative $\partial_{s'_j}$, as the other non-trivial partial derivatives go similarly. We have
$$\partial_{s'_j}g_{jk}(s'_j,t'_j,s'_k,t'_k)=\left(\frac{\|y'_j - y'_k\|^2-2(s'_j-s'_k)^2}{\|y'_j - y'_k\|^3}, \frac{2(t'_j-t'_k)(s'_j-s'_k)}{\|y'_j - y'_k\|^3}\right).$$
The square of the norm of $\partial_{s'_j}g_{jk}(s'_j,t'_j,s'_k,t'_k)$ is given by
\begin{align*}
\left\|\partial_{s'_j}g_{jk}(s'_j,t'_j,s'_k,t'_k)\right\|^2&=\frac{\left(\|y'_j - y'_k\|^2-2(s'_j-s'_k)^4\right)^2}{\|y'_j - y'_k\|^6}+\frac{4(t'_j-t'_k)^2(s'_j-s'_k)^2}{\|y'_j - y'_k\|^6}\\
&= \frac{\|y'_j - y'_k\|^4-4(s'_j-s'_k)^2\|y'_j - y'_k\|^2+4(s'_j-s'_k)^4+4(t'_j-t'_k)^2(s'_j-s'_k)^2}{\|y'_j - y'_k\|^6}\\
&= \frac{\|y'_j - y'_k\|^4}{\|y'_j - y'_k\|^6}=\frac1{\|y'_j - y'_k\|^2}
\end{align*}
Thus $\left\|\partial_{s'_j}g_{jk}(s'_j,t'_j,s'_k,t'_k)\right\|=\frac1{\|y'_j - y'_k\|}$. If we choose $y'\in \mathrm{supp} f$,
then by \eqref{eq:epsilon} $\|y'_j-y'_k\|>\epsilon$ and thus we conclude that $\left\|\partial_{s'_j}g_{jk}(s'_j,t'_j,s'_k,t'_k)\right\|=\frac1{\|y'_j - y'_k\|}<\frac1{\epsilon}$.

For a fixed $y\in\supp f$, we are going to apply Lemma \ref{lm:derivative} below with $g_{jk}:\mathbb{R}^{2n}_{\neq}\to\mathbb{R}$ defined by
$g_{jk}(y'):= \hat e_y\cdot  \frac{y'_j-y'_k}{\|y'_j - y'_k\|}, \eta:=\frac\pi{4n^2}, H:=\frac4\epsilon, M:=4$. 
By the first part of the present lemma, we know that for $y'=y$, we have
\begin{align*}
 \left| g_{jk}(y)\right|= \left|\hat e_y\cdot  \frac{y_j-y_k}{\|y_j - y_k\|}\right|\geq\frac{\pi}{4n^2}=\eta.
\end{align*}
Without loss of generality, among the non-trivial partial derivatives, we consider the partial derivative $\partial_{s_j}$. It holds that $\partial_{s_j}  g_{jk}(y') = \hat e_y\cdot \partial_{s_j}\frac{y'_j-y'_k}{\|y'_j - y'_k\|}$. Thus
\begin{align*}
 \left|\partial_{s_j}  g_{jk}(y')\right|= \left|\hat e_y\cdot \partial_{s_j}\frac{y'_j-y'_k}{\|y'_j - y'_k\|}\right|\leq \left|\partial_{s_j}\frac{y'_j-y'_k}{\|y'_j - y'_k\|}\right|=\frac1{\|y'_j - y'_k\|}<\frac1{\epsilon}=\frac HM.
\end{align*}
Therefore, as long as $\|y'-y\| < \frac{\pi\epsilon}{32n^2}=\frac{\eta}{2H}$, we have by Lemma \ref{lm:derivative}
$$\left|g_{jk}(y')\right|=\left|\hat e_y\cdot  \frac{y'_j-y'_k}{\|y'_j - y'_k\|}\right| \ge \frac{\pi}{8n^2},$$
namely $|\hat e_y\cdot  (y'_j-y'_k)| \ge \frac{\pi}{8n^2}\|y'_j - y'_k\|$ as desired.
\end{proof}

The point of the following lemma is that the estimate is independent of the dimension $l$,
because of the fixed number of non-trivial partial derivatives.
\begin{lemma}\label{lm:derivative}
Let $h:U\subseteq\mathbb{R}^l\to\mathbb{R}$ be a differentiable function. Assume that there exits a point $x\in U$ such that $|h(x)|\geq \eta$ for $\eta> 0$ and $|\partial_{x'_i}h(x')|<\frac{H}{M}$ for $H>0$, where $M$ is the number of non-trivial partial derivatives. Then, for all $x'\in U$ such that $\|x'-x\|<\frac\eta{2H}$, then $|h(x')|>\frac\eta2$.  
\end{lemma}
\begin{proof}
By the gradient theorem, one has
$$|h(x')-h(x)|=\left|\int_{\gamma}\nabla h(x')\cdot dr\right|$$
for every differentiable curve $\gamma:U\subseteq\mathbb{R}^l\to\mathbb{R}$ which starts at $x$ and ends at $x'$. If we choose the parametrization $r:[0,1]\to U$ defined by $t\mapsto tx'+(1-t)x$ of $\gamma$, then one has $r'(t)=x'-x$
\begin{align*}
|h(x')-h(x)|&=\left|\int_0^1 \nabla h(r(t))\cdot (x'-x)dt\right|=\left|\int_0^1 \sum_{j=1}^l \partial_{x'_j}h(r(t)) (x'_i-x_i)dt\right|\\
&\leq \int_0^1 \sum_{j=1}^l \left| \partial_{x'_j}h(r(t)) (x'_i-x_i) \right| dt\leq \sum_{j=1}^M \frac HM\|x'-x\|<\frac\eta2,
\end{align*} 
for all $x'\in U$ such that $\|x'-x\|<\frac\eta{2H}$. Observe that in the last sum, we considered only the non-trivial partial derivatives. Therefore, we have
\begin{align*}
|h(x')|&
\geq|h(x)|-|h(x')-h(x)|>\eta-\frac\eta2=\frac\eta2.
\end{align*}
which is our conclusion.
\end{proof} 

In the next Lemma, recall the identification $\bbR^2 \cong \bbC$,
with coordinates $y = (s,t) \in \bbR^2$, corresponding to $\zeta = s+\ri t \in \bbC$,
and we use these two representations interchangeably.
\begin{lemma}\label{lm:translateddomain}
Let $\vartheta \in [0,2\pi)$ and $\hat e_y = (\cos \vartheta, \sin \vartheta) \in S^1 \subset \bbR^2$.
For any compact subset $K \subset X_n(\bbC) \cap U_{n,\vartheta} (=\bbR^{2n}_{\neq} \cap U_{n,\vartheta})$, where
\begin{align*}
 U_{n,\vartheta}
 &=\{(\zeta_1,\cdots, \zeta_n) \in \bbC^n : 0<\Re (\re^{-\ri\vartheta}\zeta_1) <\dots < \Re (\re^{-\ri\vartheta}\zeta_r)\} \\
 &=\{(y'_1,\cdots, y'_n) \in \bbR^{2n} : 0<\hat e_y \cdot y'_1 < \cdots < \hat e_y \cdot y'_n\} \\
\end{align*}
there exists $N > 0$ such that
for any $\la \in \bbR$ with $\la>N$,
\begin{align*}
K+\la \hat e^y \subset X_n(\bbC) \cap U_{n,<} &=\{(\zeta_1,\cdots, \zeta_n) \in \bbC^n : 0<|\zeta_1| <\dots <|\zeta_r|\}\\
K -\la\hat e^y \subset X_n(\bbC) \cap U_{n,>} &=\{(\zeta_1,\cdots, \zeta_n) \in \bbC^n : |\zeta_1| >\dots >|\zeta_r|>0\}.
\end{align*}

Furthermore, if
\begin{align*}
K\subset  \{(y'_1,\cdots, y'_n) \in \bbR^{2n} : \textstyle{\frac{\hat e_y \cdot (y'_\ell - y'_{\ell+1})}{\|y'_\ell - y'_{\ell+1}\|}} > \eta, \ell = 1,\cdots, n-1\} \cap
\{y' \in \bbR^{2n} : \|y'\| \le R\}
\end{align*}
for some $0 < \eta < 1 < R$,
then we can take $\lambda = \frac {2R^2}{\eta}$ and it results in
\begin{align*}
K+\la \hat e^y \subset X_n(\bbC) \cap U_{n,<}
&=\{(\zeta_1,\cdots, \zeta_n) \in \bbC^n : \textstyle{\frac{|\zeta_\ell|}{|\zeta_{\ell+1}|} < 1- \frac{\eta^2}{4R^2}}\}\\
K-\la \hat e^y \subset X_n(\bbC) \cap U_{n,>}
&=\{(\zeta_1,\cdots, \zeta_n) \in \bbC^n : \textstyle{\frac{|\zeta_{\ell+1}|}{|\zeta_\ell|} < 1- \frac{\eta^2}{4R^2}}\}.
\end{align*}
\end{lemma}
\begin{proof}
By rotation symmetry, we may assume that $\vartheta=0$, or $\hat e_y = (1,0)$. We have
\begin{align*}
|\zeta_{\ell+1}+\lambda|^2-|\zeta_\ell + \lambda|^2 
&= |\zeta_{\ell+1}|^2 - |\zeta_\ell|^2 + 2\lambda \Re (\zeta_{\ell+1}-\zeta_\ell),
\end{align*}
and by taking a large enough $|\lambda|$ we can show the first claim.

As for the second claim, as $ ||\zeta_\ell|^2 - |\zeta_{\ell+1}|^2| < 2R^2$ and $\Re (\zeta_{\ell+1}-\zeta_\ell) > \eta$,
by taking $\lambda > \frac {2R^2}{\eta}$, we have $|\zeta_{\ell+1}+\lambda|^2-|\zeta_\ell + \lambda|^2 > 2R^2$.
Furthermore, by noting that $|\zeta_{\ell+1} + \lambda|^2 < \frac{4R^4}{\eta^2}$,
we obtain
\begin{align*}
 \frac{|\zeta_\ell + \lambda|}{|\zeta_{\ell+1} + \lambda|} < \sqrt{1 - 2R^2\cdot \frac{\eta^2}{4R^4}} < 1 - \frac{\eta^2}{4R^2}. 
\end{align*}
\end{proof}

\paragraph{(b) Preparing the partition of unity}
For each $y \in \supp f$, there is $\hat e_y$ as in Lemma \ref{lm:circle-cover}.
Actually, the vector $\hat e_y$ can be used for $y'$ with $\|y'-y\| \le \frac{\pi\epsilon^2}{8n^2}$
in the sense of Lemma \ref{lm:translateddomain}.
This ball includes a $2n$-dimensional hypercube with side $\frac{\pi\epsilon^2}{4\sqrt 2 n^\frac52}$.

Therefore, by taking gluing margins of $\frac{\pi\epsilon^2}{16\sqrt 2 n^\frac52}$ on each side,
we can find an open cover of $\supp f$ with less than $\left(\frac{32\sqrt 2 n^\frac52 R}{\pi\epsilon^2}\right)^{2n}$
hypercubes\footnote{This is not optimal. Indeed, it should be possible to gather
these open sets to the corresponding $\hat e_y$. But we do not need it for the sake of linear growth.},
which by the Stirling formula can be estimated in the form of $\bigstar \frac{R^{2n}}{\epsilon^{4n}}$.

We take a smooth partition of unity subordinate to this open cover.
We can take a concrete partition of unity as follows. First, we take the product
of $2n$ bump functions in different variables, and we may assume that
the $m$-th derivative of each one-dimensional bump function do not grow too fast.
One can obtain a partition of unity with $2n$ real variables as a product of $2n$ single variable bump functions. 

\paragraph{(b-1)  Construction and estimates of a 1-dimensional partition of unity}
For the one-variable case, we start by considering the bump function 
$$\phi(s)=\begin{cases}
\re^{\frac1{s^2-1}} & s\in (-1,1)\\
0 & \text{elsewhere}
\end{cases}.
$$
Consider now the function $\psi(s) = \int_{-\infty}^s\phi(s_1)ds_1 /\int_{-\infty}^\infty \phi(s_1)ds_1$ and its translation $\psi(s-3)$.
It is straightforward to see that the function $\tilde \psi(s) = \psi(s)-\psi(s-3)$ is supported in $[-1,4]$ and $\sum_{\ell \in \bbZ} \tilde\psi(s - 3 \ell) = 1$ for all $s \in \bbR$.
As this family is clearly locally finite, it is a partition of unity on $\bbR$ up to an overall constant.

Since $\psi(s-3)$ is a translation of $\psi(s)$, one needs to estimate only $\psi(s)$ and can estimate $\psi(s-3)$ with the same argument. The first derivative is given by $\psi^{(1)}(s)=\phi(s)$ ($=\re^{\frac1{s^2-1}}$ on $(-1,1)$), for which one has $|\psi^{(1)}(s)|\leq 1$ for every $s\in\mathbb{R}$. Moreover, $\psi^{(1)}$ is an even function, and thus $\psi^{(2)}$ is an odd function. Therefore, $\psi^{(2k+1)}$ and $\psi^{(2k)}$ are respectively even and odd functions for every $k=0,1,\ldots$. Hence, without loss of generality we consider only $s\in (-1, \infty)$. 

To estimate the $m$-th derivative for $m=2,3,\ldots$, one rewrites $\psi^{(1)}$ as 
$$\psi^{(1)}(s)=\re^{\frac1{s^2-1}}=\re^{-\frac12 \frac 1{s+1}}\re^{\frac12 \frac 1{s-1}}.$$
Let us define $g_{\pm}(s):=\re^{\mp \frac12 \frac 1{s\pm1}}$, then $\psi^{(1)}(s)=g_+(s)g_-(s)$. We observe that $g_+(-s)=\re^{- \frac12 \frac 1{-s+1}}=\re^{\frac12 \frac 1{s-1}}=g_-(s)$ for all $s\in\mathbb{R}$. If one defines $h_k(s):= (h(s))^k =\frac1{(s+1)^k}$ for every $k=1,2,3,\ldots$, then one has $g_+(s)=\re^{-\frac12h_1(s)}$.
By induction, it holds that $h^{(j)}=(-1)^j j!\,h_{j+1}$ and $h_j'=-j\,h_{j+1}$, for all $j=1,2,3,\ldots$. 

One can compute the $l$-th derivative of $g_+$ by the chain rule and derivatives of $h_1(s)$, but we do not need explicit expressions.
\begin{lemma}\label{lm:der_g}
We have the following.
\begin{itemize}
 \item[(1)] The $l$-th derivative of $g_+$ is a linear combination of $2^{l-1}$ terms of the form $h_jg_+, j \le l$
 (if we do not gather the proportional terms).
 \item[(2)] The highest power of $\frac1{s+1}$ that appears in the $l$-th derivative is $2l$.
 \item[(3)] In the sense of 1), all the absolute values of the coefficient of $h_jg_+, j \le l$ are less than or equal to
 $(2l)!! = (2l)(2l-2)\cdots 2$.
\end{itemize}
\end{lemma}
\begin{proof}
(1) By induction for $l$ we need to show, for some coefficients $\{\alpha_k\}$,
\begin{align}\label{2m_pieces_der}
g_+^{(l)}(s)=\sum_{k=1}^{2^{l-1}}\alpha_k h_{j_k}(s)g_+(s),
\end{align}
where $j_k \le l$ (they are not necessarily distinct).
The case $l=1$ is clear. By the Leibniz rule ($(fg)' = f'g + fg'$),
the derivative of each term produces two terms, and as
$h_{j_k}'(s) = -j_k\,h_{j_k+1}(s)$ and $g_+'(s) = \frac12 h_2(s)g_+(s)$
and $h_{j_k}(s)h_2(s) = h_{j_k+2}(s)$, it remains in this form with $2^l$ terms
(if we do not gather proportional terms).

(2) We show this by induction. For $l=1$
$$g_+^{(1)}(s)=-\frac12h^{(1)}(s)\re^{-\frac12 h(s)}=\frac12\cdot h_2(s) g_+(s),$$   
which has 1 term of the form $h_2g_+$, the highest power is $h_{2}$ with a coefficient $1!=1$.

Suppose it is true for $l$, and we prove it for $l+1$. Then the highest power of $h$ that appears in the $l$-th derivative is $h_{2l}g_+$. Then, taking the first derivative, one gets the highest power of $h$ that appears in the $(l+1)$-derivative, since the powers of $h$ increase computing the derivatives. One has $(h_{2j}g_+)^{(1)}=h_{2j}^{(1)}g_++h_{2j}g_+^{(1)}=-2j\,h_{2j+1}g_++h_{2j+2}g_+$, and thus the highest power of $h$ that appears in the $(l+1)$-derivative is $2(l+1)$. 

(3) Again in the expansion \eqref{2m_pieces_der} (the $l$-th derivative of $g_+$),
we know that there are $2^{l-1}$ terms and the highest power of $\frac1{s+1}$ is $2l$.
Upon a further derivative, each term gives two terms by the Leibniz rule, and we obtain $2^l$ terms
(we do not collect the same functions, just count them separately),
and the coefficients get multiplied either by $-h_k$ or $1$, whose absolute values are less than or equal to $2l$.
Thus, after $l$ derivatives, the coefficients are bounded by $(2l)!!$.
\end{proof}

The estimate of the $l$-th derivative of $g_+(s)$ boils down to estimating terms of the form $h_j(s)g_+(s)=\frac1{(s+1)^j}\re^{-\frac12\frac1{s+1}}$. Thus, for a fixed $j$, we have
\begin{align*}
\sup_{s\in(-1,1)} h_j(s)g_+(s)
&\le \sup_{s\in(-1,\infty)} h_j(s)g_+(s)=\sup_{s\in(-1,\infty)}\frac1{(s+1)^j}\re^{-\frac12\frac1{s+1}}
=\sup_{s\in(0,\infty)}\frac1{s^j}\re^{-\frac12\frac1s} \\
&=\sup_{\frac sj\in(0,\infty)}\frac1{s^j}\re^{-\frac12\frac1{s}}=\sup_{s\in(0,\infty)}\frac{j^j}{s^j}\re^{-\frac12\frac js} \\
&=j^j\sup_{s\in(0,\infty)}\frac1{s^j}\re^{-\frac12\frac j{s}}=j^j\left(\sup_{s\in(0,\infty)}\frac1{s}\re^{-\frac12\frac 1{s}}\right)^j\leq j^j.
\end{align*}

Let us estimate the $l$-th derivative of $g_+(s)$. Using Lemma \ref{lm:der_g}, 1), in particular \eqref{2m_pieces_der}, one has
\begin{align}\label{eq:estimateg+}
 |g^{(l)}_+(s)|=\left|\sum_{i=1}^{2^{l-1}}\alpha_ih_{j_i}(s)g_+(s)\right|\leq\sum_{i=1}^{2^{l-1}}|\alpha_i|\left|h_{j_i}(s)g_+(s)\right|
 \leq \sum_{i=1}^{2^{l-1}} (2l)!!\,j_i^{j_i}\leq 2^{l-1} (2l)!!\, (2l)^{2l},
\end{align}
where in the last equality we used Lemma \ref{lm:der_g}, 2) and 3). 

We need a similar estimate for $g_-$ on $(-\infty, 1)$.
Since one has $g_-(s)=g_+(-s)$, for $k=1,2,\ldots$, by induction it is straightforward to show that $g^{(2k-1)}_-(s)=-g^{(2k-1)}_+(-s)$ and $g^{(2k)}_-(s)=g^{(2k)}_+(-s)$. Thus, by \eqref{eq:estimateg+},
\begin{align*}
 |g^{(l)}_-(s)| \leq 2^{l-1} (2l)!!\, (2l)^{2l}.
\end{align*}

We are now ready to estimate $\tilde\psi^{(m)}$ for $m=1,2\ldots$.
The $m$-th derivative produces $2^m$ terms by the Leibniz rule, thus
combining all the above computations together, we have
\begin{align}\label{eq:boundPOU1dim}
\sup_{s \in \bbR}|\tilde\psi^{(m)}(s)|
&= \sup_{s \in \bbR} |\psi^{(m)}(s)-\psi^{(m)}(s-3)|\leq \sup_{s \in \bbR}|\psi^{(m)}(s)|+\sup_{s \in \bbR}|\psi^{(m)}(s-3)| \nonumber \\
&\leq 2\sup_{s \in (-1,1)}|\psi^{(m)}(s)| \nonumber \\
&\leq2\sum_{j=0}^{m-1}\binom{m-1}j |g_+^{(j)}(s)|\,|g_-^{(m-1-j)}(s)| \nonumber \\
&\leq 2^m \cdot 2^m (2m)!!\, (2m)^{2m}.
\end{align}

\paragraph{(b-2)  Multidimensional partition of unity}
A $2n$-dimensional partition of unity is given by the product of $2n$ one-dimensional partitions of unity given by the family $C_\psi\tilde\psi(s_j)$, for $j=1,2\ldots,2n$, with a normalization constant $C_\psi$ and $m_j \in \bbZ$, by
$$\Phi_{m_1, \cdots, m_{2n}}(s_1,\ldots,s_{2n}) = C_\psi^q\tilde\psi(4s_1-m_1)\tilde\psi(4s_2-m_2)\ldots\tilde\psi(4s_{2n}-m_{2n})$$
gives a $2n$-dimensional partition of unity. Let $\alpha=(\alpha_1,\alpha_2,\ldots,\alpha_{2n})$ be a multi-index
with $|\alpha| = \sum_j \alpha_j$. Then one has
\begin{align*}
\sup_{s_j \in \bbR}|\partial^{\alpha}\Phi_{m_1, \cdots, m_{2n}}(s_1,\ldots,s_{2n})|
&= C_\psi^{2n}\sup_{s_j \in \bbR}|(\partial_{s_1}^{\alpha_1}\ldots\partial_{s_{2n}}^{\alpha_{2n}})
\tilde\psi(4s_1-m_1)\ldots\tilde\psi(4s_{2n}-m_{2n})| \\
&= C_\psi^{2n} 4^{|\alpha|}\sup_{s_j \in \bbR}|(\partial_{s_1}^{\alpha_1}\ldots\partial_{s_{2n}}^{\alpha_{2n}})
\tilde\psi(4s_1-m_1)\ldots\tilde\psi(4s_{2n}-m_{2n})| \\
&\leq C_\psi^{2n} 4^{|\alpha|}\prod_{j=1}^{2n} 4^{\alpha_j} (2{\alpha_j})!!\, (2{\alpha_j})^{2{\alpha_j}} \\
&\le C_\psi^{2n} 8^{|\alpha|} (2|\alpha|)!\cdot (2|\alpha|)^{2|\alpha|}.
\end{align*}

By scaling such partition of unity by $\frac{32\sqrt 2 n^\frac52}{\pi\epsilon^2}$,
we conclude that each piece has the $\alpha$-derivative bounded by
\begin{align*}
 C_\psi^{2n} 8^{|\alpha|} (2|\alpha|)!\cdot (2|\alpha|)^{2|\alpha|} \cdot \left(\frac{32\sqrt 2 n^\frac52}{\pi\epsilon^2}\right)^{|\alpha|}. 
\end{align*}
Note that, at a linear order, say $\alpha$-derivative with $|\alpha| = qn$, it is
\begin{align}\label{eq:boundPOU}
 \bigstar \frac{1}{\epsilon^{2qn}}. 
\end{align}

\paragraph{(c) Estimate of the integral}
Now we estimate $S^{\pmb{a}}_n(f)$.
First we multiply $f$ with the smooth partition $\{g_\sigma\}$ of unity from (b).
We know the number of elements in the partition of unity, $\left(\frac{32\sqrt 2 n^\frac52 R}{\pi\epsilon^2}\right)^{2n} = \bigstar \frac{R^{2n}}{\epsilon^{4n}}$,
and we have an estimate of the Schwartz norm of each piece, $|f_\sigma|_m \le 2^m \cdot |f|_m |g_\sigma|_m$ by the Leibniz rule.
As each such piece can be translated by a multiple of $\hat e_y$, by Lemma \ref{lm:translateddomain}
and by the translation and permutation invariance of $S^{\pmb{a}}_n$, we may assume that
each $f_\sigma$ is supported in $\{(\zeta_1, \cdots, \zeta_n) : \frac{|\zeta_{j+1}|}{|\zeta_j|} < 1 - \eta', j = 1, \cdots n-1\}$,
where $\eta' = \left(\frac\pi{16Rn^2}\right)^2$. Here we estimate $S^{\pmb{a}_n}(f_\sigma)$.

For simplicity, let us write $f$ instead of $f_\sigma$. We expand $f$ into a multiple Fourier series for fixed $\rho_j$
in the polar coordinates $\zeta_j = \rho_j \re^{\ri\vartheta_j}$:
\begin{align}\label{eq:Fourierdef}
 f(\zeta_1, \cdots, \zeta_n)
 &= \sum_{k_1, \cdots, k_n \in \bbZ} \re^{\ri(k_1 \vartheta_1 + \cdots + k_n \vartheta_n)} f_{k_1,\cdots, k_n}(\rho_1, \cdots, \rho_n), \nonumber \\
 f_{k_1, \cdots, k_n}(\rho_1, \cdots, \rho_n)
 &= \frac1{(2\pi)^n}\int f(\rho_1, \vartheta_1, \cdots, \rho_n, \vartheta_n) \re^{-\ri k_1\vartheta_1 \cdots -\ri k_n \vartheta_n}d\vartheta_1\dots d\vartheta_n.
\end{align}

On the other hand, \eqref{eq:standardregion} can be written as
\begin{align*}
 S_n^{\pmb{a}}(\zeta_1, \cdots, \zeta_n)
 &=
\sum_{\substack{r_2, \cdots, r_n \in \bbR \\ s_2, \cdots, s_n \in \bbR}}
 \<\vac, a_1(-r_2\cdots -r_n, -s_2\cdots -s_n)a_2(r_2,s_2)\cdots a_2(r_n,s_n)\vac\> \\
 &\qquad\qquad \times \left(\frac{\rho_2}{\rho_1}\right)^{-r_2-s_2-\cdots -r_n-s_n}\cdots \left(\frac{\rho_n}{\rho_{n-1}}\right)^{-r_n-s_n} \\
 &\qquad\qquad \times \re^{\ri(r_2+\cdots+r_n-s_2-\cdots-s_n) \vartheta_1 + (-r_2+s_2)\vartheta_2 + \cdots + (-r_n+s_n) \vartheta_n)}
 \frac1{\rho_1^2\cdots \rho_n^2}
\end{align*}

Now we calculate \eqref{eq:defSn} using \eqref{eq:standardregion}:
\begin{align*}
 &\int_{\frac{\rho_{j+1}}{\rho_j} < 1 - \eta'} S_n^{\pmb{a}}(\zeta_1, \cdots, \zeta_n) f(\zeta_1, \cdots, \zeta_n)\rho_1 \cdots \rho_n d\rho_1 d\vartheta_1 \cdots d\rho_n  d\vartheta_n \\
 &= \int_{\frac{\rho_{j+1}}{\rho_j} < 1 - \eta'} \sum_{r_2, \cdots, r_n, s_2, \cdots, s_n \in \bbR}
 \<\vac, a_1(-r_2\cdots -r_n, -s_2\cdots -s_n)a_2(r_2, s_2)\cdots a_2(r_n, s_n)\vac\> \\
 &\quad\quad\times\left(\frac{\rho_2}{\rho_1}\right)^{-r_2\cdots -r_n -s_2\cdots -s_n}\cdots \left(\frac{\rho_n}{\rho_{n-1}}\right)^{-r_n-s_n}
 \frac1{\rho_1^2\cdots\rho_n^2}
 f_{- r_2 \cdots - r_n + s_2 +\cdots + s_n, r_2 - s_2, \cdots, r_n - s_n}(\rho_1, \cdots, \rho_n) \\
 &\quad \quad \times \rho_1 \cdots \rho_n d\rho_1 \cdots d\rho_n.
\end{align*}
Here, the sum $\sum_{r_2, \cdots, r_n, s_2, \cdots, s_n \in \bbR}$
can be equally written as
$\sum_{r_2 + s_2, \cdots, r_n + s_n \in \bbR}\sum_{r_2 - s_2, \cdots, r_n - s_n \in \bbZ}$,
and the sum over $r_j - s_j$ are restricted to $\bbZ$ because 
all the vectors $a_1(r_1, s_1)\cdots a_2(r_n, s_n)\vac$ vanish unless $\sum_j s_j - r_j \in \bbZ$ by FO1).

First we calculate the sum $\sum_{r_2 - s_2, \cdots, r_n - s_n \in \bbZ}$, noting that the indeces of $f$ are of this form.
On one hand, as $f$ is a test function, for each $m \in \bbZ$, the Schwartz norm
\begin{align}\label{eq:schwartz}
 |f|_m := \sup_{\substack{\zeta_j \in \bbR^2 \\ |\alpha| \le m}} |(1+|\pmb{\zeta}|^2)^\frac m2 \partial^\alpha f(\zeta_1, \cdots, \zeta_n)|
\end{align}
is finite, where $|\pmb{\zeta}|^2 = \sum_{j=1}^n |\zeta_j|^2$ and $\partial^\alpha f$ denotes the derivative of $f$
with a multi-index $\alpha$ and $|\alpha|$ is the order of the derivative.

By using \eqref{eq:Fourierdef} and the formula that relates the Fourier transform of $f$ and its derivative, one has
\begin{align}\label{eq:der_FT}
 \ri^{|\beta|}k_1^{\beta_1} \cdots k_n^{\beta_n} f_{k_1, \cdots, k_n}(\rho_1, \cdots, \rho_n)
 = \frac1{(2\pi)^n}\int D_{\pmb{\vartheta}}^\beta f(\rho_1, \vartheta_1, \cdots, \rho_n, \vartheta_n) \re^{-\ri k_1\vartheta_1 \cdots -\ri k_n \vartheta_n}d\vartheta_1\dots d\vartheta_n.
 \end{align}
Therefore, it holds that
\begin{align*}
 (|k_1| + 1)^{\beta_1} \cdots (|k_n|+1)^{\beta_n}| f_{k_1, \cdots, k_n}(\rho_1, \cdots, \rho_n)| 
 &=\sum_{j=1}^n\sum_{\ell=0}^{\beta_j}\binom{\beta_j}{\ell}|k_j|^\ell|f_{k_1, \cdots, k_n}(\rho_1, \cdots, \rho_n)| \\
 &\le \frac{2^{|\beta|}}{(2\pi)^n}
 \int |D_{\pmb{\vartheta}}^\beta f(\rho_1, \vartheta_1, \cdots, \rho_n, \vartheta_n)|d\vartheta_1\dots d\vartheta_n \\
 &\le 2^{|\beta|}2^n|f|_{2n + |\beta|},
\end{align*}
where we have estimated the integral using \eqref{eq:schwartz} and the fact that
the $\vartheta_j$-derivative is a linear combination of partial derivatives in $\tau_j, \xi_j$ with
coefficients $\le 1$. Thus,
\begin{align}\label{eq:der_FTplus}
 (|k_1| + 1)^{\beta_1+2} \cdots (|k_n|+1|)^{\beta_n+2}| f_{k_1, \cdots, k_n}(\rho_1, \cdots, \rho_n)|
 &\le 2^{|\beta|+2n}2^n|f|_{4n + |\beta|}.
\end{align}

By Lemma \ref{lm:PEBestimate},
\begin{align*}
 &|\<\vac, a_1(r_1, s_n)a_2(r_2, s_2)\cdots a_2(r_n, s_n)\vac\>| \\
  &\le \bigstar (|r_1 + s_1|+1 +\cdots |r_n + s_n|+1)^{nQ_\Upsilon} (|r_1 - s_1|+1 +\cdots |r_n - s_n|+1)^{nQ_\Upsilon}
\end{align*}
and by expanding the right-hand side,
we have $(n!)^{2(Q_\Upsilon+1)} = \bigstar$ terms of the following form, with $\sum_j \beta^+_j = \sum_j \beta^-_j = nQ_\Upsilon$,
\begin{align*}
 \prod_{j=1}^n (|r_j + s_j|+1)^{\beta^+_j} (|r_j - s_j|+1)^{\beta^-_j}
\end{align*}
By \eqref{eq:der_FTplus},
\begin{align*}
 &\sum_{r_1-s_1, r_n-s_n \in \bbZ} \prod_{j=1}^n (|r_j-s_j|+1)^{\beta^-_j} |f_{r_1-s_1, \cdots, r_n-s_n}(\rho_1, \cdots, \rho_n)| \\
 &\le \sum_{r_1-s_1, r_n-s_n \in \bbZ} \prod_{j=1}^n (|r_j-s_j|+1)^{-2} \cdot 2^{nQ_\Upsilon+2n}2^n |f|_{(4+Q_\Upsilon)n} \\
 &= \bigstar |f|_{(4+Q_\Upsilon)n}.
\end{align*}
Therefore, we have
\begin{align}\label{eq:estimate-}
 &\left|\int_{\frac{\rho_{j+1}}{\rho_j} < 1 - \eta'} S_n^{\pmb{a}}(\zeta_1, \cdots, \zeta_n) f(\zeta_1, \cdots, \zeta_n)\rho_1 \cdots \rho_n d\rho_1 d\vartheta_1 \cdots d\rho_n  d\vartheta_n\right| \nonumber \\
 &\le \bigstar |f|_{(4+Q_\Upsilon)n}\int_{\frac{\rho_{j+1}}{\rho_j} < 1 - \eta'}
 \underset{r_2+s_2, \cdots, r_n+s_n \in \bbR}{\overset{\circ}\sum} \left(\frac{\rho_2}{\rho_1}\right)^{-r_2\cdots -r_n -s_2\cdots -s_n}\cdots \left(\frac{\rho_n}{\rho_{n-1}}\right)^{-r_n-s_n} \nonumber \\
  &\qquad\qquad\qquad\qquad\qquad\qquad\qquad\qquad \times \frac{\prod_{j=1}^n (|r_j + s_j|+1)^{\beta^+_j}}{\rho_1\cdots\rho_n}
   d\rho_1 \cdots d\rho_n,
\end{align}
where we wrote $\overset{\circ}\sum$ the sum
where $\<\vac, a_1(-r_2\cdots -r_n, -s_2\cdots -s_n)a_2(r_2, s_2)\cdots a_2(r_n, s_n)\vac\>$ do not vanish.

Next, note that
\begin{align*}
 |r_j + s_j| + 1 &= |(r_j + s_j + \cdots + r_n + s_n) - (r_{j+1} + s_{j+1} + \cdots + r_n + s_n)| + 1 \\
 &\le |r_j + s_j + \cdots + r_n + s_n| + 1 + |r_{j+1} + s_{j+1} + \cdots + r_n + s_n| + 1,
\end{align*}
and $r_1 + s_1 = -r_2-\cdots-r_n - s_2 - \cdots -s_n$ at a non-zero term in the sum $\overset{\circ}\sum$,
thus by writing $t_j := r_j + s_j + \cdots + r_n + s_n$,
\begin{align*}
 \prod_{j=1}^n (|r_j + s_j| + 1)^{\beta^+_j} 
 &\le (|t_2|+1)^{\beta^+_1}(|t_n|+1)^{\beta^+_n}\prod_{j=2}^{n-1} (|t_j| + 1 + |t_{j+1}| + 1)^{\beta^+_j}
\end{align*}
and by expanding it, there are less than $2^{nQ_\Upsilon}$ terms each of which is bounded by
$\prod_{j=2}^{n} (|t_j| + 1)^{2Q_\Upsilon}$.

To estimate the remaining $\rho$-integral of \eqref{eq:estimate-},
the number of nontrivial terms in $\overset{\circ}\sum_{t_2,\cdots, t_n \ge 0}$ when
$N_j \le t_j < N_{j+1}$ is less than the number of $\{(r_j, s_j)\}$
for which this scalar product does not vanish,
$\<\vac, a_1(-r_2\cdots -r_n, -s_2\cdots -s_n)a_2(r_2, s_2)\cdots a_2(r_n, s_n)\vac\> \neq 0$,
which is by (PSD) \eqref{eq:density} bounded by
$C^n\prod_{j=1}^n (N_j+2)^L$.
We change the variable $\rho'_{j+1} = \rho_{j+1}/\rho_j$ for $j=1,\cdots,n-1$
and calculate:
\begin{align*}
 &\int_{\frac{\rho_{j+1}}{\rho_j} < 1- \eta'} \overset{\circ}{\underset{r_2 + s_2, \cdots, r_n + s_n \in \bbR}{\sum}}
 \prod_{j=1}^n (|r_j+s_j| + 1)^{\beta^+_j}
 \left(\frac{\rho_2}{\rho_1}\right)^{-r_2\cdots -r_n -s_2\cdots -s_n}\cdots \left(\frac{\rho_n}{\rho_{n-1}}\right)^{-r_n-s_n}
 \frac{d\rho_1 \cdots d\rho_n}{\rho_1\cdots \rho_n} \\
 &\le 2^{nQ_\Upsilon} \int_{\rho'_{j+1} < 1- \eta'} \overset{\circ}\sum_{t_2, \cdots, t_n \ge 0}
 \prod_{j=1}^n (|t_j| + 1)^{2Q_\Upsilon}
 (\rho'_2)^{t_2}\cdots (\rho'_n)^{t_n}
 d\rho'_2\cdots d\rho'_n \\
 &\le \bigstar \int_{\rho'_{j+1} < 1- \eta'} \sum_{\substack{t_j \in \bbZ\\t_2, \cdots, t_n \ge 0}}
 \prod_{j=1}^n (|t_j| + 1)^{2Q_\Upsilon+L}
 (\rho'_2)^{t_2}\cdots (\rho'_n)^{t_n}
 d\rho'_2\cdots d\rho'_n \\
 &= \bigstar \sum_{\substack{t_j \in \bbZ\\t_2, \cdots, t_n \ge 0}}
 \prod_{j=1}^n (t_j + 1)^{2Q_\Upsilon+L}
 \frac{(1- \eta')^{t_j+1}}{t_j + 1} \\
 &\le \bigstar \left(\int_0^\infty t^{2Q_\Upsilon+L-1} (1- \eta')^t dt\right)^n \\
 &\le \bigstar \Gamma(2Q_\Upsilon+L)^n\left(\frac1{-\log(1-\eta')}\right)^{(2Q_\Upsilon+L)n} \\
 &\le \bigstar'\left(\frac1{\eta'}\right)^{(2Q_\Upsilon+L)n}
 = \bigstar \left(\frac{16Rn^2}\pi\right)^{(2Q_\Upsilon+L)n}
 = \bigstar'R^{(2Q_\Upsilon+L)n},
\end{align*}
where
$\Gamma$ is the Gamma-function and
after the change of variables there is no dependence on $\rho_1$, so
we integrated it out and obtained the radius of the domain which is small (because $f = f_\sigma$ is cut by a partition of unity)
so replaced by $1$,
while we ignored $\rho_n$ in the denominator as we may assume that it is larger than $1$.

This is an estimate of single $f_\sigma$. We cut the original $f$ into
$\left(\frac{16\sqrt 2 n^\frac52 R}{\pi\epsilon^2}\right)^{2n} = \bigstar\frac{R^{2n}}{\epsilon^{4n}}$ pieces.
Recall that $|f_\sigma|_m \le 2^m |f|_m |g_\sigma|_m$ by the Leibniz rule.
By collecting them all with \eqref{eq:boundPOU}, we have
\begin{align}\label{eq:beforeabsorption}
 |S^{\pmb{a}}_n(f)| \le \sum_\sigma |S^{\pmb{a}}_n(f_\sigma)| \le \bigstar' |f|_{(4+Q_\Upsilon)n}\frac{R^{(2Q_\Upsilon+L+2)n}}{\epsilon^{(4+2(4+Q_\Upsilon))n}}.
\end{align}

\paragraph{(d) Extending to functions vanishing at the coinciding points}
The estimate \eqref{eq:beforeabsorption} is obtained for a test function $f$ with $\supp f$ contained
in the disk with radius $R$ and $\min_{j,k} |z_j - z_k| > \epsilon$.
By taking a larger Schwartz norm, we can eliminate the dependence on $R, \epsilon$.

\paragraph{(d-1) Absorbing $R$}
We show that the estimate $|S^{\pmb{a}}_n(f)| \le \bigstar |f|_{(4+Q_\Upsilon)n} \frac{R^{N_1 n}}{\epsilon^{N_2 n}}$
implies $|S^{\pmb{a}}_n(f)| \le \bigstar' |f|_{(4+Q_\Upsilon+N_1)n+3}\frac1{\epsilon^{N_2 n}}$.
Concretely, we will take $N_1 = 2Q_\Upsilon+L+2$
and $N_2 = 4+2(4+Q_\Upsilon)$.

To remove the $R$ dependence with the cost of larger Schwartz norm,
note that any test function $f$ supported in $\bbR^{2n}_{\neq}$ can be written as
$f = f_0 + f_1 + \cdots + f_{\tilde R}, \tilde R \in \bbN$, where $R < \tilde R$ and 
\begin{align*}
 f_j(\pmb{z}) = f(\pmb{z}) \cdot h_j(|\pmb{z}|),
\end{align*}
where $\{h_j\}$ is a smooth partition of unity on $\bbR$ (a scaled version of the one in part (b-1)) such that $\supp h_j \subset [j-\frac13, j]$
and $h_j = 1$ on $[j+\frac13,j+\frac23]$.
By the estimate there, we know that there are constants $C_{h,1}, C_{h,2}$ such that
$\sup_{s \in \bbR} |\partial^m h_j(s)| \le C_{h,1}(m!)^{C_{h,2}}$.

For arbitrary $N$, we have
\begin{align*}
 |f_j|_m
 &= \sup_{\substack{\pmb{\zeta} \in \bbR^{2n} \\ |\alpha| \le m}} |(1 + |\pmb{\zeta}|^2)^\frac m2 |\partial^\alpha f_j(\pmb{\zeta})| \\
 &= \sup_{\substack{\pmb{\zeta} \in \bbR^{2n} \\ |\alpha| \le m}} \left((1 + |\pmb{\zeta}|^2)^\frac m2 |\partial^\alpha f_j(\pmb{\zeta})|
 \cdot \frac{(1+|\pmb{\zeta}|^2)^\frac {N+3}2}{(1+|\pmb{\zeta}|^2)^\frac {N+3}2}\right)\\
 &\le \sup_{\substack{\pmb{\zeta} \in \bbR^{2n} \\ |\alpha| \le m}} (1 + |\pmb{\zeta}|^2)^\frac {m+N+3}2 |\partial^\alpha f_j(\pmb{\zeta})|
 \cdot \frac{1}{(1+(j-\frac13)^2)^\frac {N+3}2}\\
 &\le \frac{|f_j|_{m+N+3}}{(1+(j-\frac13)^2)^\frac {N+3}2},
\end{align*}
where the case $j=0$ requires a slightly different estimate (the third equality changes)
but it holds that $|f_0|_m \le |f_0|_{m+N+3}$.

Furthermore, for arbitrary $M$ it holds that
\begin{align*}
  |f_j|_M
  &= \sup_{\substack{\pmb{\zeta} \in \bbR^{2n} \\ |\alpha| \le M}} (1 + |\pmb{\zeta}|^2)^\frac M2 |\partial^\alpha (f(\pmb{\zeta})h_j(\pmb{\zeta}))| \\
  &\le 2^M\sup_{\substack{\pmb{\zeta} \in \bbR^{2n} \\ |\alpha| \le M}} (1 + |\pmb{\zeta}|^2)^\frac M2 |\partial^\alpha f(\pmb{\zeta})|
 \cdot \sup_{\substack{\pmb{\zeta} \in \bbR^{2n} \\ |\beta| \le M-|\alpha|}}|\partial^\beta h_j(|\pmb{\zeta}|)| \\  
  &\le 6^M M! C_{h,1}(M!)^{C_{h,2}}|f|_M,
\end{align*}
as $h_j$ is a one-dimensional partition of unity and its derivatives
$\partial^\beta h_j(|\pmb{\zeta}|)$ can be estimated by the derivatives of $h_j$
and the maximum of $\partial^\beta|\pmb{\zeta}|$ can be bounded by
$|\beta|! 2^{|\beta|}(\frac32)^{|\beta|}$ for $|\pmb{\zeta}|\ge \frac23$, and for $|\pmb{\zeta}| < \frac23$ we may assume that $h_0(\pmb{\zeta}) = 1$,
thus the derivatives vanish.

Therefore, by taking $M = m + N + 3$ with $m = (4+Q_\Upsilon)n, N = N_1$,
by noting that $\supp f_j$ has radius $j+\frac32$,
\begin{align*}
 |S^{\pmb{a}}_n(f)| &\le \sum_j |S^{\pmb{a}}_n(f_j)| \le \bigstar\sum_j |f_j|_{(4+Q_\Upsilon)n} \frac{(j+\frac32)^{N_1 n}}{\epsilon^{N_2 n}} \\
 &\le \bigstar\left(|f_0|\frac{(\frac32)^{N_1 n}}{\epsilon^{N_2 n}}
 + \sum_{j \ge 1} \frac{|f_j|_{(4+Q_\Upsilon)n+N_1 n+3}}{(1+(j-\frac13)^2)^\frac {N_1 n+3}2} \frac{(j+\frac32)^{N_1 n}}{\epsilon^{N_2 n}}\right) \\
 &\le \bigstar' |f|_{(4+Q_\Upsilon)n+N_1 n+3}
 \left({\textstyle{(\frac32)^{N_1 n}}} + \sum_{j\ge 1} \frac{(j+2)^{Qn}}{(1+(j-\frac13)^2)^\frac {Qn+3}2}\right) \cdot \frac1{\epsilon^{N_2 n}}\\
 &\le \bigstar' |f|_{(4+Q_\Upsilon+N_1) n+3} \frac1{\epsilon^{N_2 n}}.
\end{align*}

\paragraph{(d-2) Absorbing $\epsilon$}

Let $\check h_0$ be a smooth function on $\bbR$ such that $\check h_0(t) = 0$ for $t \le 2^{-1}$,
$\check h_0(t) = 1 = 2^0$ for $1 \le t$. As we have seen in (b-1), we can take $\check h_0$ in such a way
that $|\frac{d^m}{dt^m} \check h_0|_0 \le C_{h,1}(m!)^{C_{h,2}}$.
We define $\check h_\ell(t) = \check h_0(2^\ell t)$.
Then $|\frac{d^m}{dt^m} \check h_\ell|_0 \le (2^\ell)^m C_{h,1}(m!)^{C_{h,2}}$.

Let $h_\ell(\pmb{\zeta}) = \prod_{j\neq k} \check h_{\ell+1}(|\zeta_j - \zeta_k|) - \prod_{j\neq k} \check h_\ell(|\zeta_j - \zeta_k|)$.
Then for any $\pmb{\zeta} \in \supp h_\ell$, $2^{-\ell-2} < |\zeta_j - \zeta_k|$,
but there is a pair $j,k$ such that $|\zeta_j - \zeta_k| < 2^{-\ell}$.
Moreover, note that for any partial derivative $\partial^\alpha h_\ell(\pmb{\zeta})$
there are at most $(2n-1)^{|\alpha|}$ terms (because each variable is involved only in $2n-1$ factors).
With the scaling factor $2^{\ell|\alpha|}$ for the $|\alpha|$-th derivative
and the estimate of partition of unity, we have
\begin{align}\label{eq:suph}
 \sup_{\pmb{\zeta} \in \bbR^{2n}} |\partial^\alpha h_\ell(\pmb{\zeta})| \le (2n)^{|\alpha|} \cdot 2^{\ell|\alpha|} \cdot C_{h,1}(|\alpha|!)^{C_{h,2}}.
\end{align}

For any test function $f$ with compact support in $\bbR^{2n}_{\neq}$, we put $f_\ell = f\cdot h_\ell$.
Then $f = f_0 + \cdots + f_L$ for some $L \in \bbN$ and
\begin{align*}
 |S_n^{\pmb{a}}(f)| \le  \sum_{\ell} |S_n^{\pmb{a}}(f_\ell)|
 \le \bigstar \sum_\ell |f_\ell|_{(4+Q_\Upsilon+N_1) n+3}2^{(\ell+2)N_2 n}
 = \bigstar' \sum_\ell |f_\ell|_{(4+Q_\Upsilon+N_1) n+3}2^{\ell N_2 n},
\end{align*}
because of part (d-1) applied to $f_\ell$ and the fact that for $\pmb{\zeta} \in \supp f_\ell$
it holds that $\epsilon = 2^{-\ell - 2} < |\zeta_k - \zeta_j|$.

Next recall the multi-variable Taylor formula with remainder
(cf.\! \href{https://sites.math.washington.edu/~folland/Math425/taylor2.pdf}{these lecture notes}):
\begin{align*}
 g(\pmb{\zeta}) = \sum_{|\alpha| \le m} \frac{\partial^\alpha g(\pmb{\omega})}{\alpha!}(\pmb{\zeta} - \pmb{\omega})^\alpha
 + \sum_{|\alpha| = m+1} \frac{\partial^\alpha g(\pmb{\omega} + c(\pmb{\zeta} - \pmb{\omega}))}{\alpha!}(\pmb{\zeta} - \pmb{\omega})^\alpha
\end{align*}
where $0 < c < 1$ and $m$ arbitrary. In particular, if $\partial^\alpha g(\pmb{\omega}) = 0$ for all $\alpha$, we have
\begin{align}\label{eq:multitaylor0}
 |g(\pmb{\zeta})| \le 
 \sum_{|\alpha| = m+1} \frac{|\partial^\alpha g(\pmb{\omega} + c(\pmb{\zeta} - \pmb{\omega}))|}{\alpha!}|(\pmb{\zeta} - \pmb{\omega})^\alpha|.
\end{align}
Let $\pmb{\zeta} \in \supp f_\ell$ and $j<k$ be a pair for which $|\zeta_j - \zeta_k| < 2^{-\ell}$.
We set
\begin{align*}
 \pmb{\omega} = (\zeta_, \cdots, \zeta_j, \cdots, \zeta_k = \zeta_j, \cdots, \zeta_n) \in \bbR^{2n}_=.
\end{align*}
Applying \eqref{eq:multitaylor0} to $\partial^\beta f_\ell$ with $\pmb{\omega} \in \bbR^{2n}_=$, we have for any $m \in \bbN$
and $|\beta| \le M$ and $ \pmb{\zeta} \in \supp f_\ell$,
\begin{align}\label{eq:zinsuppfl}
 (1 + |\pmb{\zeta}|^2)^\frac M2|\partial^\beta f_\ell(\pmb{\zeta})|
 &\le 
 (1 + |\pmb{\zeta}|^2)^\frac M2 \sum_{|\alpha| = m+1} \frac{|\partial^{\alpha+\beta}f_\ell(\pmb{\omega} + c(\pmb{\zeta} - \pmb{\omega}))|}{\alpha!}|(\pmb{\zeta} - \pmb{\omega})^\alpha| \nonumber \\
 &\le 2^M (1 + |\pmb{\zeta} + c(\pmb{\zeta} - \pmb{\omega})|^2)^\frac M2
 \sum_{|\alpha| = m+1} \frac{|\partial^{\alpha+\beta}f_\ell(\pmb{\omega} + c(\pmb{\zeta} - \pmb{\omega}))|}{\alpha!} 2^{-\ell(m+1)},
\end{align}
where the derivatives $\partial^\alpha$ act only on $\zeta_k = (y_{k,1}, y_{k,2})$.
With $m = N_2 n$ and $M = (4+Q_\Upsilon+N_1) n+3$,
and using the fact that the above summation contains $2^{|\alpha|}$ terms while the denominator
is at least $((\frac{|\alpha|}2)!)^2$, thus their ratio is less than a universal constant,
\begin{align*}
 &|f_\ell|_{(4+Q_\Upsilon+N_1) n+3} \\
 &= \sup_{\substack{\pmb{\zeta} \in \bbR^{2n} \\ |\beta| \le (4+Q_\Upsilon+N_1) n+3}}
 |(1 + |\pmb{\zeta}|^2)^\frac {(4+Q_\Upsilon+N_1) n+3}2 |\partial^\beta f_\ell(\pmb{\zeta})| \\
 &= 
 2^{(-\ell+1)(N_2 n+1)} \bigstar \sup_{\substack{\pmb{\zeta} \in \bbR^{2n} \\ |\beta| \le (4+Q_\Upsilon+N_1) n+3 \\ |\alpha|= N_2 n+1}}
 \left((1 + |\pmb{\zeta} + c(\pmb{\zeta} - \pmb{\omega})|^2)^\frac {(4+Q_\Upsilon+N_1) n+3}2
 \frac{|\partial^{\alpha+\beta}f_\ell(\pmb{\omega} + c(\pmb{\zeta} - \pmb{\omega}))|}{\alpha!}\right)\\
 &\le \bigstar 2^{(-\ell)(N_2 n+1)} |f_\ell|_{(4+Q_\Upsilon+N_1+N_2) n+4}
\end{align*}

Thus, by putting $4+Q_\Upsilon+N_1+N_2 = Q_\Upsilon'$ we have
\begin{align*}
 |S^{\pmb{a}}_n(f)| \le
 \bigstar' \sum_\ell 2^{-\ell} |f_\ell|_{Q_\Upsilon' n+4}.
\end{align*}

Furthermore, we have
\begin{align*}
 |f_\ell|_{Q_\Upsilon' n+4}
 &= \sup_{\substack{\pmb{\zeta} \in \bbR^{2n} \\ |\beta| \le Q_\Upsilon' n+4}}
 |(1 + |\pmb{\zeta}|^2)^\frac {Q_\Upsilon' n+4}2 |\partial^\beta (f h_\ell)(\pmb{\zeta})| \\
 &= \sup_{\substack{\pmb{\zeta} \in \bbR^{2n} \\ |\beta| \le Q_\Upsilon' n+4}} \sum_{\alpha \le \beta}
 |(1 + |\pmb{\zeta}|^2)^\frac {Q_\Upsilon' n+4}2 |\partial^{\beta - \alpha} f(\pmb{\zeta})| \cdot |\partial^{\alpha} h_\ell(\pmb{\zeta})| \\
 &= \bigstar 2^{\ell(Q_\Upsilon' n+4)}\sup_{\substack{\pmb{\zeta} \in \supp f_\ell \\ |\beta| \le Q_\Upsilon' n+4 \\ \alpha \le \beta}}
 |(1 + |\pmb{\zeta}|^2)^\frac {Q_\Upsilon' n+4}2
 |\partial^{\beta-\alpha}f(\pmb{\zeta})| \\
 &= \bigstar' 2^{\ell(Q_\Upsilon' n+4)}|f|_{2Q_\Upsilon' n+8} 2^{-\ell(Q_\Upsilon' n+4+1)} \\
 &\le \bigstar |f|_{2Q_\Upsilon' n+8}  \\
\end{align*}
where we used \eqref{eq:suph} and $|\alpha| \le Q_\Upsilon' n+4$ in the third line
(the number of terms apparing from the Leibniz rule is bounded by $(2n)^{Q_\Upsilon'n+4} = \bigstar$),
and \eqref{eq:zinsuppfl} with $m = Q_\Upsilon' n+4$ in the fourth line by noting the range of $\pmb{\zeta}$.

Altogether,
\begin{align*}
 |S^{\pmb{a}}_n(f)| \le
 \bigstar' \sum_\ell |f|_{2Q_\Upsilon' n+8}2^{-\ell} \le \bigstar' |f|_{2Q_\Upsilon' n+8} \le \bigstar |f|_{(2Q_\Upsilon'+8) n}
\end{align*}
holds for $n\neq 0$. For $n=0$, $|S_0| = 1$ by definition.
Thus (OS\ref{ax:lineargrowth}) holds.

\subsection{Reflection positivity}\label{RP}
We prove (OS\ref{ax:RP}) in two steps:
\begin{itemize}
 \item[(a)] We prove the reflection positivity for a set of test functions satisfying a slightly different support condition
 \item[(b)] We show that a general set of test functions can be brought to the previous set by a large dilation
\end{itemize}

\paragraph{(a) Reflection positivity for a subset of functions.}
For a quasi-primary vector $a$, we consider the formal series $\varphi(a, \uz) = Y(a, \uz)\uz^{\uh_a}$.
Let $a$ be Hermite,
then it holds that $\<u, \varphi(a, \uz)v\> = \<\varphi(a, \uz^{-1})u, v\>$ by \eqref{eq:hermite-sp}.

Let us consider the equality as formal series
\begin{align*}
 \<\vac, \varphi(a_1, \uz_1)\cdots \varphi(a_{m+n}, \uz_{m+n})\vac\> 
 =\<\varphi(a_m, \uz_m^{-1})\cdots \varphi(a_1, \uz_1^{-1})\vac, \varphi(a_{m+1}, \uz_{m+1})\cdots \varphi(a_{m+n}, \uz_{m+n})\vac\>.
\end{align*}
Both sides are convergent if this is evaluated at $|\zeta_1| > \cdots > |\zeta_{m+n}| > 0$.
Recall the convention $\<zu, v\> = z\<u, v\>$ for a formal variable $z$.
This implies that, when we evaluate a formal variable $z$ by $\zeta$ on the left of the scalar product, we have to substitute it by $\bar \zeta$.
Therefore,
\begin{align}\label{eq:conjugatephi}
 \<\vac, \varphi(a_1, \uze_1)\cdots \varphi(a_{m+n}, \uze_{m+n})\vac\> 
 = \< \varphi(a_m, r\uze_m)\cdots \varphi(a_1, r\uze_1)\vac, \varphi(a_{m+1}, \uze_{m+1})\cdots \varphi(a_{m+n}, \uze_{m+n})\vac\>,
\end{align}
where $r\zeta = \bar\zeta^{-1}$ is the reflection with respect to the unit circle.

We claim that an expression $\Psi(\uze_1, \cdots, \uze_m) = \varphi(a_1, \uze_1)\cdots \varphi(a_m, \uze_m)\vac$ defines a vector in the completion
of $F$ with respect to the norm if $1 > |\zeta_1| > \cdots > |\zeta_m| > 0$.
Indeed, for each pair $h, \bar h \in \bbR$, $F_{h, \bar h}$ is a finite-dimensional vector space,
and the projection $P_{h, \bar h}\Psi(\uze_1, \cdots, \uze_m)$ onto $F_{h, \bar h}$ is an absolutely convergent series
in $\uze_1, \cdots, \uze_m$ if $|\zeta_1| > \cdots > |\zeta_m|$.
Then for any finite set $\frH$ of pairs $(h, \bar h)$,
\begin{align*}
 \sum_{(h, \bar h) \in \frH} \|P_{h, \bar h}\Psi(\uze_1, \cdots, \uze_m)\|^2
 &= \sum_{(h, \bar h) \in \frH}
 \left.\<P_{h, \bar h}\Psi(\uze_{m+1}, \cdots, \uze_{2m}), P_{h, \bar h}\Psi(\uze_1, \cdots, \uze_m)\>\right|_{\zeta_{m+j} = \zeta_j}.
\end{align*}
This is a power series in $\uze_1, \cdots, \uze_{2m}$ evaluated at $\zeta_{m+j} = \zeta_j$ for all $j = 1,\cdots, m$.
The power series is given by extracting terms of the form $A_{s_1, \cdots, s_m, t_1,\cdots, t_m} \zeta^{-s_1}\zee^{-t_1}\cdots \zeta^{-s_{2m}}\zee^{-t_{2m}}$
satisfying $\sum_{j=1}^m s_j = h, \sum_{j=1}^m t_j = \bar h$ for some $(h, \bar h) \in \frH$ from
\begin{align*}
 \<\vac, \varphi(a_{2m}, \uze_{2m})\cdots\varphi(a_{m+1}, \uze_{m+1}) \varphi(a_1, \uze_1)\cdots \varphi(a_m, \uze_m)\vac\>
\end{align*}
It is absolutely convergent when evaluated at $r(\zeta_{m+j}) = \zeta_j$ for all $j=1,\cdots, m$
if $1 > |\zeta_1| > \cdots > |\zeta_m| > 0$, as in this case it holds that $|\zeta_{2m}| > \cdots |\zeta_{m+1}| > 1$.
This implies that the norm $\|\Psi(\uze_1, \cdots, \uze_m)\|^2 = \sum_{(h, \bar h)} \|P_{h, \bar h}\Psi(\uze_1, \cdots, \uze_m)\|^2$ is convergent.

We are interested in the reflection $\theta(\omega) = -\bar \omega$ (the reflection with respect to the imaginary axis),
rather than $r\omega = \bar\omega^{-1}$.
Using a slight variation of the Cayley transform $\frC \in \PSLtwoC$, where $\frC(\omega) = \frac{\frac1{\sqrt2}(\ri\omega+\ri)}{\frac1{\sqrt2}(-\ri\omega+\ri)} = \frac{1+\omega}{1-\omega}$,
we map $\ri\bbR$ to $S^1$. Other important points are:
\begin{align*}
 \frC(\omega) = \begin{cases}
         \infty & \text{ if } \omega = 1 \\
        0 & \text{ if } \omega = -1 \\
         1 & \text{ if } \omega = 0 \\
    -1 & \text{ if } \omega = \infty
        \end{cases}.  
\end{align*}
We also have $\frC^{-1}(\zeta) = \frac{\zeta-1}{\zeta+1}$.
It holds that 
\begin{align}\label{eq:CinvRC}
 \frC^{-1} \circ r \circ \frC(\omega) = -\bar \omega = \theta \omega.
\end{align}

We know that the power series
\begin{align*}
 \<\vac, Y(a_1, \underline \zeta_1)\cdots Y(a_n, \underline \zeta_n)\vac\>
\end{align*}
converge for $|\zeta_1| > \cdots > |\zeta_n|$.
As $\frac{d\frC}{d\omega}(\omega) = \frac{2}{(1-\omega)^2}$,
by Proposition \ref{prop_global_conf}, for $\omega_1, \cdots, \omega_n \in \bbC$ such that $\omega_j \neq 1, -1$ and
$|\frC(\omega_1)| > \cdots > |\frC(\omega_n)|$ we have
\begin{align*}
 S_n^{\pmb{a}}(\omega_1,\cdots, \omega_n)
 = \prod_{j=1}^n \left(\frac{2}{(1-\omega_j)^2}\right)^{h_{a_j}}\overline{\left(\frac{2}{(1-\omega_j)^2}\right)}^{\bar h_{a_j}}\<\vac, Y(a_1, \underline{\frC(\omega_1)})\cdots Y(a_n, \underline{\frC(\omega_n)})\vac\>,
\end{align*}
where the scalar product has a convergent expansion in $\frC(\omega_j)$.
In terms of the field $\varphi$ above, this amounts to
\begin{align}\label{eq:SJphi}
 S_n^{\pmb{a}}(\omega_1,\cdots, \omega_n)
 = \prod_{j=1}^n \frJ(\omega_j)^{h_{a_j}}\overline{\frJ(\omega_j)}^{\bar h_{a_j}}
 \<\vac, \varphi(a_1, \underline{\frC(\omega_1)})\cdots \varphi(a_n, \underline{\frC(\omega_n)})\vac\>,
\end{align}
where we introduced $\frJ(\omega) = \frac{2}{(1-\omega)^2}\cdot \frac1{\frC(\omega)} = \frac{2}{(1-\omega)(1+\omega)}$.
It holds that $\overline{\frJ(\omega)} = \frJ(\bar \omega)$ and $\frJ(-\omega) = \frJ(\omega)$, thus
\begin{align}\label{eq:Jtheta}
 \overline{\frJ(\omega)} = \frJ(\bar\omega) =  \frJ(- \bar\omega) = \frJ(\theta \omega).
\end{align}

By smearing the correlation functions $S^{\pmb{a}_j}_{j}$ with a finite set of test functions $f^{\pmb{a}_j}_j(\zeta_1,\cdots, \zeta_j)$
supported in the set of $\omega_1, \cdots, \omega_n \in \bbC$ such that $\omega_j\neq \pm 1$ for all $j$ and
$1 > |\frC(\omega_1)| > \cdots > |\frC(\omega_n)|$ we considered above,
the following expression gives a vector in the completion of $F$ with respect to the norm:
\begin{align*}
 \Psi^{\pmb{a}_j} &= \int f_j(\omega_1, \cdots, \omega_j)
 \prod_{\ell=1}^j \frJ(\omega_j)^{h_{a_j}}\overline{\frJ(\omega_j)}^{\bar h_{a_j}}
  \varphi(a_1, \frC(\omega_1))\cdots \varphi(a_j, \frC(\omega_j)) \vac d\tau_1\xi_1\cdots d\tau_j d\xi_j,
\end{align*}
where $\omega_j = \tau_j + \ri\xi_j$.

Let $A$ be as in (OS\ref{ax:RP}).
Then, the positive-definiteness of the scalar product tells that
\begin{align*}
 0 &\le \<\sum_{\pmb{b}_j \in A} \Psi^{\pmb{b}_j}, \sum_{\pmb{a}_k \in A} \Psi^{\pmb{a}_k}\> \\
 &= \sum_{\pmb{b}_j, \pmb{a}_k \in A}
 \int d\tau_1d\xi_1\cdots d\tau_j d\xi_j d\tau_{j+1}d\xi_{j+1}\cdots d\tau_{j+k} d\xi_{j+k} \;\overline{f^{\pmb{b_j}}_j(\omega_1, \cdots, \omega_j)} f^{\pmb{a_k}}_k(\omega_{j+1}, \cdots, \omega_{j+k}) \\
 &\qquad\prod_{\ell=1}^j \overline{\frJ(\omega_\ell)}^{h_{b_\ell}} \frJ(\omega_\ell)^{\bar h_{b_\ell}}
 \prod_{m=1}^k \frJ(\omega_{j+m})^{h_{a_m}} \overline{\frJ(\omega_{j+m})}^{\bar h_{a_m}}\cdot \\
 &\qquad\<\varphi(b_1, \underline{\frC(\omega_1)})\cdots \varphi(b_j, \underline{\frC(\omega_j)}) \vac,
   \varphi(a_{1}, \underline{\frC(\omega_{j+1})})\cdots \varphi(a_{k}, \underline{\frC(\omega_{j+k})}) \vac\> \\
 &= \sum_{\pmb{b}_j, \pmb{a}_k \in A}
 \int d\tau_1d\xi_1\cdots d\tau_j d\xi_j d\tau_{j+1}d\xi_{j+1}\cdots d\tau_{j+k} d\xi_{j+k} \;\overline{f^{\pmb{b_j}}_j(\omega_1, \cdots, \omega_j)} f^{\pmb{a_k}}_k(\omega_{j+1}, \cdots, \omega_{j+k}) \\
 &\qquad\prod_{\ell=1}^j \overline{\frJ(\omega_\ell)}^{h_{b_\ell}} \frJ(\omega_\ell)^{\bar h_{b_\ell}}
 \prod_{m=1}^k \frJ(\omega_{j+m})^{h_{a_m}} \overline{\frJ(\omega_{j+m})}^{\bar h_{a_m}}\cdot \\
 &\qquad\<\vac, \varphi(b_j, \underline{r\frC(\omega_j)})\cdots \varphi(b_1, \underline{r\frC(\omega_1)})
   \varphi(a_{1}, \underline{\frC(\omega_{j+1})})\cdots \varphi(a_{k}, \underline{\frC(\omega_{j+k})}) \vac\> \\
 &= \sum_{\pmb{b}_j, \pmb{a}_k \in A}
 \int d\tau_1d\xi_1\cdots d\tau_j d\xi_j d\tau_{j+1}d\xi_{j+1}\cdots d\tau_{j+k} d\xi_{j+k} \;\overline{f^{\pmb{b_j}}_j(\omega_1, \cdots, \omega_j)} f^{\pmb{a_k}}_k(\omega_{j+1}, \cdots, \omega_{j+k}) \\
 &\quad\qquad\times\prod_{\ell=1}^j \frJ(\theta \omega_\ell)^{h_{b_\ell}} \overline{\frJ(\theta \omega_\ell)}^{\bar h_{b_\ell}} \prod_{m=1}^k \frJ(\omega_{j+m})^{h_{a_m}}\overline{\frJ(\omega_{j+m})}^{\bar h_{a_m}} \\
 &\quad\qquad\times\<\vac, \varphi(b_j, \underline{\frC(\theta \omega_j)})\cdots \varphi(b_1, \underline{\frC(\theta \omega_1)})
 \varphi(a_{j+1}, \underline{\frC(\omega_{j+1})})\cdots \varphi(a_{j+k}, \underline{\frC(\omega_{j+k})}) \vac\> \\
 &= \sum_{\pmb{b}_j, \pmb{a}_k \in A}
 \int d\tau_1d\xi_1\cdots d\tau_j d\xi_j d\tau_{j+1}d\xi_{j+1}\cdots d\tau_{j+k} d\xi_{j+k}
 \;\overline{f^{\pmb{b_j}}_j(\theta \omega_j, \cdots, \theta \omega_1)} f^{\pmb{a_k}}_k(\omega_{j+1}, \cdots, \omega_{j+k}) \\
 &\qquad\quad\times\prod_{\ell=1}^j \frJ(\omega_\ell)^{h_{b_\ell}} \overline{\frJ(\omega_\ell)}^{\bar h_{b_\ell}} \prod_{m=1}^k \frJ(\omega_{j+m})^{h_{a_m}}\overline{\frJ(\omega_{j+m})}^{\bar h_{a_m}} \\
 &\qquad\quad\times\<\vac, \varphi(b_j, \underline{\frC(\omega_1)})\cdots \varphi(b_1, \underline{\frC(\omega_j)})
 \varphi(a_{1}, \underline{\frC(\omega_{j+1})})\cdots \varphi(a_{k}, \underline{\frC(\omega_{j+k})}) \vac\> \\
 &= \sum_{\pmb{b}_j, \pmb{a}_k \in A}
 \int d\tau_1d\xi_1\cdots d\tau_j d\xi_j d\tau_{j+1}d\xi_{j+1}\cdots d\tau_{j+k} d\xi_{j+k}
 \;\Theta (f^{\pmb{b_j}}_j)^*(\omega_1, \cdots, \omega_j) f^{\pmb{a_k}}_k(\omega_{j+1}, \cdots, \omega_{j+k}) \\
 &\qquad\quad\times S^{(\phi \pmb{b}_j, \pmb{a}_k)}_{j+k}(\omega_1,\cdots,\omega_{j+k}),
\end{align*}
where in the 3rd line we used the relation \eqref{eq:conjugatephi},
in the 4th line we used the relations \eqref{eq:CinvRC} and \eqref{eq:Jtheta},
in the 5th line we used the invariance of the measure under $\omega \mapsto \theta \omega$
and made the relabelling of the variables $(\omega_1,\cdots, \omega_j) \mapsto (\omega_j,\cdots, \omega_1)$
and in the 6th line we used the definition $(\Theta f_j^*)(\omega_1,\cdots, \omega_n) = \overline{f_j(\theta \omega_j, \cdots \theta \omega_1)}$
and \eqref{eq:SJphi}.
This is reflection positivity (OS\ref{ax:RP}) for the test functions chosen above.

\paragraph{(b) The general case.}
Note that the support properties of $\{f^{\pmb{a}_j}_j\}$ are not the ones required in (OS\ref{ax:RP}).
Let $\{f^{\pmb{a}_j}_j\}$ be test functions compactly supported in $\bbR^{2n}_{\neq}$ with $f^{\pmb{a}_j}_j \in \scS^{\pmb{a}_j}_+(\bbR^{2j})$.
Our goal is to prove \eqref{eq:RP} in this setting.
For $\lambda > 0$ large, by conformal invariance (Proposition \ref{prop_global_conf}) we have
$S_m^{\pmb{a}}(\omega_1,\cdots, \omega_m) = \prod_{j=1}^m \re^{\lambda(h_j+\bar h_j)}S_m^{\pmb{a}}(\re^\lambda \omega_1,\cdots, \re^\lambda \omega_m)$.
Therefore,
 \begin{align}\label{eq:lambda}
 &\sum_{\pmb{b}_j, \pmb{a}_k \in A} S_{j+k}^{(\phi\pmb{b}_j, \pmb{a}_k)}(\Theta (f^{\pmb{b}_j}_j)^* \otimes f^{\pmb{a}_k}_k) \nonumber \\
 &= \sum_{\pmb{b}_j, \pmb{a}_k \in A}
 \int d\tau_1d\xi_1\cdots d\tau_{j+k} d\xi_{j+k} \;\overline{f^{\pmb{b}_j}_j(\theta \omega_j, \cdots, \theta \omega_1)} f^{\pmb{a}_k}_k(\omega_{j+1}, \cdots, \omega_{j+k}) 
 S_{j+k}^{(\phi\pmb{b}_j, \pmb{a}_k)}(\omega_1, \cdots, \omega_{j+k}) \nonumber \\
 &= \sum_{\pmb{b}_j, \pmb{a}_k \in A}
 \int d\tau_1d\xi_1\cdots d\tau_{j+k} d\xi_{j+k} \;\overline{f^{\pmb{b}_j}_j(\theta \re^\lambda \omega_j, \cdots, \theta \re^\lambda \omega_1)} f^{\pmb{a}_k}_k(\re^\lambda \omega_{j+1}, \cdots, \re^\lambda \omega_{j+k}) \nonumber \\
 &\qquad\qquad S_{j+k}^{(\phi\pmb{b}_j, \pmb{a}_k)}(\re^\lambda \omega_1, \cdots, \re^\lambda \omega_{j+k}) \nonumber \\
 &= \sum_{\pmb{b}_j, \pmb{a}_k \in A}
 \int \prod_{\ell=1}^{j+k} \re^{2\lambda}\cdot d\tau_1d\xi_1\cdots d\tau_{j+k} d\xi_{j+k} \;\overline{f^{\pmb{b}_j}_j(\theta \re^\lambda \omega_j, \cdots, \theta \re^\lambda \omega_1)} f^{\pmb{a}_k}_k(\re^\lambda \omega_{j+1}, \cdots, \re^\lambda \omega_{j+k}) \nonumber \\
 &\qquad\qquad \prod_{\ell=1}^{j+k} \re^{-\lambda(h_\ell+\bar h_\ell)}\cdot S_{j+k}^{(\phi\pmb{b}_j, \pmb{a}_k)}(\re^\lambda \omega_1, \cdots, \re^\lambda \omega_{j+k}) \nonumber \\
 &= \sum_{\pmb{b}_j, \pmb{a}_k \in A}
 \int d\tau_1d\xi_1\cdots d\tau_{j+k} d\xi_{j+k} \;\overline{f^{\pmb{b}_j}_{\lambda,j}(\theta \omega_j, \cdots, \theta \omega_1)}f^{\pmb{a}_k}_{\lambda,k}(\omega_{j+1}, \cdots, \omega_{j+k})
 S_{j+k}^{(\phi\pmb{b}_j, \pmb{a}_k)}(\omega_1 \cdots, \omega_{j+k}) \nonumber \\
 &=\sum_{\pmb{b}_j, \pmb{a}_k \in A} S_{j+k}^{(\phi\pmb{b}_j, \pmb{a}_k)}(\Theta (f^{\pmb{b}_j}_{\lambda,j})^* \otimes f^{\pmb{a}_k}_{\lambda,k}),
 \end{align}
 where we introduced $f^{\pmb{a}_j}_{\lambda,j}(\omega_1, \cdots, \omega_j) = \re^{j\lambda - \sum_{\ell=1}^j \lambda(h_\ell+\bar h_\ell)} f^{\pmb{a}_j}_j(\re^\lambda \omega_1, \cdots, \re^\lambda \omega_j)$
 and used that $\theta \re^\lambda = \re^\lambda \theta$.
 The functions $\{f^{\pmb{a}_j}_{\lambda,j}\}$ still satisfy the support condition, but their supports are scaled by $\re^{-\lambda}$.
 
 As the supports of $\{f^{\pmb{a}_j}_j\}$ are compact, we may assume that there is (small) $\epsilon > 0$ and (large) $R > 0$ such that
 the support of $f^{\pmb{a}_j}_j$ is contained in the set
 \begin{align}\label{eq:domain1}
  \{(\omega_1, \cdots, \omega_j): \tau_1 + j\epsilon > \tau_2 + (j-1)\epsilon > \cdots > \tau_j, |\omega_\ell| < R \text{ for all } \ell\},
 \end{align}
 where $\omega_\ell = \tau_\ell + \ri\xi_\ell$.
 Then the support of $f^{\pmb{a}_j}_{\lambda,j}$ is contained in
 \begin{align}\label{eq:domain2}
  \{(\omega_1, \cdots, \omega_j): \tau_1 + \re^{-\lambda}j\epsilon > \tau_2 + \re^{-\lambda}(j-1)\epsilon > \cdots > \tau_j,
  |\omega_\ell| < \re^{-\lambda}R \text{ for all } \ell\}.
 \end{align}
 
 We claim that, for a sufficiently large $\lambda$ and sufficiently small $\eta' > 0$, this support satisfies also the condition 
 \begin{align*}
  \{(\omega_1, \cdots, \omega_j): 1 > |\frC(\omega_1)| > |\frC(\omega_2)|  > \cdots > |\frC(\omega_j)|\}.
 \end{align*}
 To see this, note that, for $\omega = \tau + \ri\xi$,
 \begin{align*}
  |\frC(\omega)|^2 = \left|\frac{1 + \tau + \ri\xi}{1 - \tau - \ri\xi}\right|^2 = \frac{1+\tau^2 + \xi^2 + 2\tau}{1+\tau^2 + \xi^2 - 2\tau}.
 \end{align*}
 For any point in the set \eqref{eq:domain1}, $\tau_\ell$ and $\tau_{\ell+1}$ differ by more than $\re^{-\lambda}\epsilon$,
 while $\tau_\ell^2 + \xi_\ell^2 = |\omega_\ell|^2 < \re^{-2\lambda}R^2, \tau_\ell^2 + \xi_\ell^2 < \re^{-2\lambda}R^2$,
 thus the latter is negligible, by taking $\lambda$ sufficiently large.
 Then it holds that $|\frC(\omega_\ell)|^2 - |\frC(\omega_{\ell+1})|^2 > \frac12(\tau_\ell - \tau_{\ell+1})$,
 as desired.
 Moreover, with a large enough $\lambda$, it is clear that
 $\supp f^{\pmb{a}_j}_{\lambda, j}$ do not contain $\pm 1$.
 
 From Step 1, we know that (OS\ref{ax:RP}) is satisfied for $\{f^{\pmb{a}_j}_{\lambda,j}\}$ and by \eqref{eq:lambda} this is equivalent to
 (OS\ref{ax:RP}) for $\{f^{\pmb{a}_j}_j\}$.

\subsection{Clustering}\label{clustering}

In this section, we will show the clustering by using Proposition \ref{prop_cluster_tree}. We first explain our idea of the proof.

We apply Proposition \ref{prop_cluster_tree} to the case $(m+1,n+1)$ where $a_{m+1}$ and $a_{m+n+2}$ are the vacuum vectors. That is, for example, consider the following  compositions of vertex operators:
\begin{align*}
Y(Y(a_1,\zeta_{1,v})Y(a_2,\zeta_{2,v})Y(a_3,\zeta_{3,v})\vac, \zeta_{v,0}) Y(a_4,\zeta_{4,0})Y(a_5,\zeta_{5,0})\vac,
\end{align*}
which absolutely convergent if
\begin{align}
\left|\frac{\zeta_{2,v}}{\zeta_{1,v}}\right|, \left|\frac{\zeta_{3,v}}{\zeta_{2,v}}\right| < 1 \quad \text{ and } \left|\frac{\zeta_{5,0}}{\zeta_{4,0}}\right| < 1 \label{eq_cluster_inequality1}
\end{align}
and
\begin{align}
\left|\frac{\zeta_{1,v}}{\zeta_{v,0}}\right| + \left|\frac{\zeta_{4,0}}{\zeta_{v,0}}\right| < 1\label{eq_cluster_inequality2}
\end{align}
by Proposition \ref{prop_cluster_tree}.

We shall think of $\{a_1,a_2,a_3\}$ as a first cluster around $\zeta_v$ and $a_4,a_5$ as a second cluster around $\zeta_0$. Then $\zeta_{v,0}$ corresponds to the distance between the two clusters.
When $\zeta_1,\dots,\zeta_5$ are fixed (generally when moving in a compact set in \eqref{eq_cluster_inequality1}), if we move the distance between the two clusters away ($\zeta_{v,0} \to \infty$), then \eqref{eq_cluster_inequality2} is automatically satisfied.
Using this fact and the property Proposition \ref{prop_unitary_spec}
and Proposition \ref{prop_unitary_vacuum} about the spectrum of the unitary full VOA, the cluster decomposition follows.

Recall that $U_n= \{|\zeta_1|>|\zeta_2|>\cdots >|\zeta_n|\}$.
\begin{proposition}
\label{prop_nm_conv}
Let $K_m \subset U_m$ and $K_n \subset U_n$ be compact subsets.
Then, $S_{m+n}^{(\pmb{a_m},\pmb{b_n})}(\ze_1+\la, \dots,\ze_m+\la, \ze_{m+1},\dots,\ze_{m+n})$ uniformly
converge to $S_m^{\pmb{a_m}}(\ze_1,\dots,\ze_m)S_n^{\pmb{b_n}}(\ze_{m+1},\dots,\ze_{m+n})$ in $K_m \times K_n$ as $\la \to \infty$.
\end{proposition}
\begin{proof}
By Theorem \ref{thm_B}, the vacuum property and symmetry, for any $\lambda \in \bbC$, we have
\begin{align*}
&S_{m+n}^{(\pmb{a_m}, \pmb{b_n})}(\ze_1+\lambda,\dots,\ze_m+\lambda,\ze_{m+1},\ze_{m+2}, \dots,\ze_{m+n})\\
&=S_{m+n+2}^{(\pmb{a_m},\vac,\pmb{b_n},\vac)}(\ze_1+\lambda,\dots,\ze_m+\lambda, \ze_v, \ze_{m+1},\ze_{m+2},\dots, \ze_{n+m},\ze_{0}),
\end{align*}
where we insert the vacuum vector at $\ze_v$ and $\ze_0$, and thus, the left-hand-side is independent of $\ze_v$ and $\ze_0$. We set $\ze_v=\lambda$ and $\ze_0=0$.

Then, by Proposition \ref{prop_cluster_tree}, we have
\begin{align}
  &S_{m+n+2}^{(\pmb{a_m},\vac,\pmb{b_n},\vac)}(\ze_1+\lambda, \dots,\ze_m+\lambda, \lambda, \ze_{m+1},\dots, \ze_{m+n},0)|_{U_{n+1,m+1}}\nonumber \\
&=\<\vac, Y\left(
Y(a_{1},\uze_{1,v})\dots Y(a_{m},\uze_{m,v})\vac,\uze_{v,0} \right)
 Y(a_{m+1},\uze_{m+1,0})\dots Y(a_{m+n},\uze_{m+n,0})\vac\> \label{eq_cluster_expansion}
\end{align}
where the right-hand-side is absolutely and locally uniformly convergent 
in
\begin{align}
\left|\frac{\zeta_{i+1,v}}{\zeta_{i,v}}\right| <1,\quad\quad \left|\frac{\zeta_{j+1,0}}{\zeta_{j,0}}\right|<1
\label{eq_conv_cluster1}
\end{align}
and
\begin{align}
\left|\frac{\zeta_{1,0}}{\zeta_{v,0}}\right|+\left|\frac{\zeta_{m+1,0}}{\zeta_{v,0}}\right| <1 \label{eq_conv_cluster2}
\end{align}
and coincides with the left-hand-side after substituting 
\begin{align*}
\zeta_{i,v}=(\ze_i+\lambda)-\lambda=\ze_i,\quad\quad \zeta_{j,0}=\ze_j, \quad\quad \zeta_{v,0}=\lambda
\end{align*}
 for $i \in\{1,\dots,m\}$ and $j \in \{m+1,\dots,m+n\}$.

Let $(\ze_1,\dots,\ze_m) \in K_m$ and $(\ze_{m+1},\dots,\ze_{m+n})\in K_n$.
Then, since $K_m \subset U_m$ and $K_n \subset U_n$, 
\eqref{eq_conv_cluster1} holds.
Moreover, if $\la \to \infty$, then
\begin{align*}
\left|\frac{\zeta_{1,v}}{\zeta_{v,0}}\right|+\left|\frac{\zeta_{m+1,0}}{\zeta_{v,0}}\right| =\left|\frac{\ze_1}{\lambda}\right|+\left|\frac{\ze_{m+1}}{\lambda}\right| <1.
\end{align*}
Hence, the series \eqref{eq_cluster_expansion} converges uniformly on $K_m \times K_n$ for sufficiently large $\la$.

Let $\{v_{h,\h}^i\}_{i\in I_{h,\h}}$ be a basis of $F_{h,\h}$ and $\{v_i^{h,\h}\}_{i\in I_{h,\h}}$ be the dual basis.
For $(h,\h)=(0,0)$, take $\vac$ as the basis of $F_{0,0}=\bbC\vac$ (see 
Proposition \ref{prop_unitary_spec}).
Set $\Delta h' = h_1+\dots+h_m$ and $\Delta h = h_{m+1}+\dots +h_{m+n}$.\\
Since 
$\zeta_{1,v}^{\Delta h'-h'}\zee_{1,v}^{\Delta \h'-\h'}\langle v_j^{h',\h'}, Y(a_{1},\uze_{1,v})\dots Y(a_{m},\uze_{m,v})\vac\>$
is scale and rotation invariant, 
\begin{align*}
F_j = \zeta_{1,v}^{\Delta h'-h'}\zee_{1,v}^{\Delta\h'-\h'}\langle v_j^{h',\h'}, Y(a_{1},\uze_{1,v})\dots Y(a_{m},\uze_{m,v})\vac\>
\end{align*}
is a power series of $\frac{\zeta_{i+1,v}}{\zeta_{i,v}}$ and $\frac{\zee_{i+1,v}}{\zee_{i,v}}$ for $i\in\{1,\dots,m\}$, i.e., depending only on the ratio.
The same holds for $G_i = \zeta_{m+1,0}^{\Delta h-h}\zee_{m+1,0}^{\Delta h-\h}\langle v_i^{h,\h}, Y(a_{m+1},\uze_{m+1,0})\dots Y(a_{m+n},\uze_{m+n,0})\vac\>$.
Then,
\begin{align*}
\<\vac, &Y\left(
Y(a_{1},\uze_{1,v})\dots Y(a_{m},\uze_{m,v})\vac,\underline{\zeta_{v,0}} \right)
 Y(a_{m+1},\uze_{m+1,0})\dots Y(a_{m+n},\uze_{m+n,0})\vac\>\\
 &=
\sum_{h,\h,h',\h'} \Big(\sum_{i\in I_{h,\h},j\in I_{h',\h'}}
 \<\vac, Y\left(v_{h'\h'}^j,\underline{\zeta_{v,0}}\right)v_{h,\h}^i \> \< v_j^{h',\h'}, Y(a_{1},\uze_{1,v})\dots Y(a_{m},\uze_{m,v}) \vac \> \\
 &\qquad\qquad\qquad \<v_i^{h,\h}, Y(a_{m+1},\uze_{m+1,0})\dots Y(a_{m+n},\uze_{m+n,0})\vac\>\Big) \\
&=\zeta_{1,v}^{-\Delta h'}\zee_{1,v}^{-\Delta \h'}
\zeta_{m+1,0}^{-\Delta h}\zee_{m+1,0}^{-\Delta \h}
\sum_{h,\h,h',\h'} \left(\sum_{i\in I_{h,\h},j\in I_{h',\h'}}
\zeta_{1,v}^{h'}\zee_{1,v}^{\h'}\zeta_{m+1,0}^{h}\zee_{m+1,0}^{\h} \<\vac, Y(v_{h',\h'}^j,\underline{\zeta_{v,0}})v_{h,\h}^i \rangle F_i G_j\right).
\end{align*}
Since
$\zeta_{v,0}^{-h-h'}\zee_{v,0}^{-\h-\h'}\<\vac,Y(v_{h',\h'}^j,\underline{\zeta_{v,0}})v_{h,\h}^i \rangle$ is scale and rotation invariant, it is independent of $\zeta_{v,0}$ and $\bar{\zeta}_{v,0}$.
Set $c_{i,j} = \zeta_{v,0}^{-h-h'}\zee_{v,0}^{-\h-\h'}\<\vac,Y(v_{h',\h'}^j,\underline{\zeta_{v,0}})v_{h,\h}^i \rangle \in \bbC$.
Then, we have
\begin{align*}
\zeta_{1,v}^{h'}\zee_{1,v}^{\h'}\zeta_{m+1,0}^{h}\zee_{m+1,0}^{\h} \<\vac, Y(v_{h',\h'}^j,\underline{\zeta_{v,0}})v_{h,\h}^i \rangle 
= c_{i,j} \left(\frac{\zeta_{1,v}}{\zeta_{v,0}}\right)^{h'}\left(\frac{\zee_{1,v}}{\zee_{v,0}}\right)^{\h'}
\left(\frac{\zeta_{m+1,0}}{\zeta_{v,0}}\right)^{h} \left(\frac{\zee_{m+1,0}}{\zee_{v,0}}\right)^{\h}.
\end{align*}
Since $L(1)F_{1,0}=\Ld(1)F_{0,1}=0$ and 
\begin{align*}
\<\vac, Y(v_{h',\h'}^j,\underline{\zeta_{v,0}})v_{h,\h}^i \rangle = \zeta_{v,0}^{-2h}\zee_{v,0}^{-2\h} \langle \exp(-L(1)\zeta_{v,0}-\Ld(1)\zee_{v,0})v_{h',\h'}^j,  \exp(-L(1)\zeta_{v,0}^{-1}-\Ld(1)\zee_{v,0}^{-1}) v_{h,\h}^i\rangle,
\end{align*}
we have
$\<\vac, Y(v_{h',\h'}^j,\underline{\zeta_{v,0}})v_{h,\h}^i \rangle=0$ if $(h,\h)=(0,0), (h',\h')\neq (0,0)$ or $(h,\h) \neq (0,0), (h',\h') = (0,0)$.
Hence, 
\begin{align}
\<\vac, &Y\left(
Y(a_{1},\uze_{1,v})\dots Y(a_{m},\uze_{m,v}),\underline{\zeta_{v,0}} \right)
 Y(a_{m+1},\uze_{m+1,0})\dots Y(a_{m+n},\uze_{m+n,0})\vac\>\nonumber \\
 &=  \<\vac, Y(a_1,\uze_{1,v})\dots Y(a_m,\uze_{m,v})\vac\> \<\vac, Y(a_{m+1},\uze_{m+1,0})\dots Y(a_{m+n},\uze_{m+n,0})\vac\>\nonumber \\
 &\qquad + \zeta_{1,v}^{-\Delta h'}\zee_{1,v}^{-\Delta \h'} \zeta_{m+1,0}^{-\Delta h}\zee_{m+1,0}^{-\Delta \h} \nonumber \\
 &\qquad \quad\times
 \left(\sum_{\substack{(h,\h)\neq (0,0),\\(h',\h')\neq (0,0)}} \sum_{i\in I_{h,\h},j\in I_{h',\h'}}
c_{i,j} \left(\frac{\zeta_{1,v}}{\zeta_{v,0}}\right)^{h'}\left(\frac{\zee_{1,v}}{\zee_{v,0}}\right)^{\h'}
\left(\frac{\zeta_{m+1,0}}{\zeta_{v,0}}\right)^{h} \left(\frac{\zee_{m+1,0}}{\zee_{v0}}\right)^{\h}
F_i G_j\right). \label{eq_cluster_hh}
\end{align}
The sum $\sum_{\substack{(h,\h)\neq (0,0),\\(h',\h')\neq (0,0)}} \left|\sum_{i\in I_{h,\h},j\in I_{h',\h'}}
c_{i,j} \left(\frac{\zeta_{1,v}}{\zeta_{v,0}}\right)^{h'}\left(\frac{\zee_{1,v}}{\zee_{v,0}}\right)^{\h'}
\left(\frac{\zeta_{m+1,0}}{\zeta_{v,0}}\right)^{h} \left(\frac{\zee_{m+1,0}}{\zee_{v,0}}\right)^{\h}F_i G_j\right|$ converges locally uniformly in $U_{m+1,n+1}$,
and $(h,\h),(h',\h')$ run through the spectrum of $F$, i.e., $\{(h,\h)\in \bbR^2\mid F_{h,\h}\neq 0\}$.
By Proposition \ref{prop_unitary_spec} and Proposition \ref{prop_unitary_vacuum},
$h+\h \geq 0$ and $h+\h =0$ if and only if $h=\h=0$.
Hence, by the polynomial spectrum density \eqref{eq:density}, the series \eqref{eq_cluster_hh} does not have singularity at 
$\frac{\zeta_{1,v}}{\zeta_{v,0}}=0$ and $\frac{\zeta_{m+1,0}}{\zeta_{v,0}}=0$
 by Proposition \ref{prop_unitary_spec} and Proposition \ref{prop_unitary_vacuum},
we can take the limit of $\la \to \infty$, i.e., $\frac{\zeta_{1,v}}{\zeta_{v,0}}, \frac{\zeta_{m+1,0}}{\zeta_{v,0}} \to 0$, which proves the assertion.
\end{proof}

Now we prove (OS\ref{ax:clustering}).
By translation-invariance (OS\ref{ax:invariance}), it suffices to show
\begin{align*}
   \lim_{\lambda \to \infty} S_{m+n}^{(\pmb{a}_m, \pmb{b}_n)}(\Theta(f^{\pmb{a}_m}_{m,\pm \ri\la})^* \otimes g^{\pmb{b}_n}_{n})
   = S_m^{\pmb{a}_m}(\Theta (f^{\pmb{a}_m}_m)^*)S_n^{\pmb{b}_n}(g^{\pmb{b}_n}_n),
\end{align*}
for each $\pmb{a}_m, {\pmb{b}_n}$ and $f^{\pmb{a}_m}_m \in \scS^{\pmb{a}_m}_+(\bbR^{2m})$ and $g^{\pmb{b}_n}_n \in \scS^{\pmb{b}_n}_+(\bbR^{2n})$
with compact support $\bbR^{2n}_{\neq}$.

By Lemma \ref{lm:translateddomain}, we may assume that 
\begin{align*}
\mathrm{supp}(f^{\pmb{a}_m}_m) \subset U_m \quad\quad\text{and}\quad\quad \mathrm{supp}\, g^{\pmb{b}_n}_n \subset  U_n,
\end{align*}
since (OS\ref{ax:clustering}) is invariant under the simultaneous translation.

Set $K_{\pmb{a}_m}= \mathrm{supp}(\Theta (f^{\pmb{a}_m}_m)^*)$ and $K_{\pmb{b}_n}= \mathrm{supp}(g^{\pmb{b}_n}_n)$.
Let $\la \in \bbR$ be sufficiently large so that $(K_{\pmb{a}_m}\pm \ri\la) \times K_{\pmb{b}_n}  \subset X_{m+n}(\bbC)$.
By the definition of the distribution,
\begin{align*}
&S_{m+n}^{(\pmb{a}_m, \pmb{b}_n)}(\tau_{\mp \ri\lambda} \Theta(f^{\pmb{a}_m}_m)^* \otimes g^{\pmb{b}_n}_{n})\\
&= \int  S_{m+n}^{(\pmb{a}_m, \pmb{b}_n)}(\ze_1,\dots,\ze_{m+n}) 
(\tau_{\mp \ri\la} \Theta (f^{\pmb{a}_m}_m)^*)(\zeta_1,\cdots, \zeta_m)  g^{\pmb{b}_n}_{n}(\zeta_1,\cdots, \zeta_n)
d\tau_1d\xi_1\cdots d\tau_{m+n} d\xi_{m+n}\\
&= \int_{K_m \times K_n} S_{m+n}^{(\pmb{a}_m, \pmb{b}_n)}(\ze_1\pm \ri\la,\dots,\ze_{m}\pm \ri\la,\ze_{m+1},\dots,\zeta_{m+n}) 
\Theta(f^{\pmb{a}_m}_m)^*(\zeta_1,\cdots, \zeta_m) g^{\pmb{b}_n}_{n}(\zeta_1,\cdots, \zeta_n) \\
&\qquad\quad d\tau_1d\xi_1\cdots d\tau_{m+n} d\xi_{m+n}.
\end{align*}
Hence, by Proposition \ref{prop_nm_conv}, and by noticing that its assumption is rotation-invariant,
thus we can apply it to $\pm \ri\lambda$ instead of $\lambda \in \bbR$,
\begin{align*}
&\lim_{\lambda \to \infty}S_{m+n}^{(\pmb{a}_m, \pmb{b}_n)}(\tau_{\mp \ri\lambda} \Theta(f^{\pmb{a}_m}_m)^* \otimes g^{\pmb{b}_n}_{n})\\
&= \int_{K_m\times K_n} 
S_{m}^{\pmb{a}_m}(\ze_1,\dots,\ze_{m}) S_{n}^{\pmb{b}_n}(\ze_{m+1},\dots,\ze_{m+n}) 
\Theta(f^{\pmb{a}_m}_m)^*(\zeta_1,\cdots, \zeta_m) g^{\pmb{b}_n}_{n}(\zeta_1,\cdots, \zeta_n) \\
&\qquad\quad d\tau_1d\xi_1\cdots d\tau_{m+n} d\xi_{m+n}\\
&= S_{m}^{\pmb{a}_m}(\Theta(f^{\pmb{a}_m}_m)^*) S_{n}^{\pmb{b}_n}(g^{\pmb{b}_n}_{n}),
\end{align*}

This completes the proof of (OS\ref{ax:clustering}), and thus all the conformal OS axioms (OS\ref{ax:lineargrowth})--(OS\ref{ax:invarianceC}).

Note that in OS\ref{ax:clustering},
only translations in spatial directions are considered, but the above proof shows the clustering for translations in arbitrary directions.

\section{Examples}\label{examples}
In this section we will see that the family of conformal field theories (full vertex operator algebras) introduced in \cite{Moriwaki21}
satisfies the assumptions of the main theorem of this paper, that is, unitarity, local $C_1$-cofiniteness, and polynomial energy bounds.
We begin with a brief review of the construction of full vertex operator algebras in \cite[Section 6]{Moriwaki21}.

Let $L$ be an even lattice, that is, $L$ is a free abelian group of finite rank $n$ equipped with non-degenerate symmetric bilinear form,
$$(-,-)_\lat:L\times L \rightarrow \bbZ,$$
such that $(\al,\al)_\lat \in 2\bbZ$ for any $\al \in L$.
Note that we do not assume that $L$ is positive-definite.

Let $\{\al_i\}_{i=1,\dots,n}$ be a basis of $L$ and $\ep:L\times L \rightarrow \bbZ_2 = \bbR^\times$ be a (non-symmetric) bilinear form defined by
\begin{align*}
\ep(\al_i,\al_i)&= (-1)^{(\al,\al)_\lat/2}\quad\quad\text{ for all $i$}\\
\ep(\al_i,\al_j)&= (-1)^{(\al_i,\al_j)_\lat}\quad\quad\text{ if $i>j$}\\
\ep(\al_i,\al_j)&= 1 \quad\quad\quad\quad\quad\quad\text{ if $j > i$}.
\end{align*}
Then, $\ep(-,-)$ is a 2-cocycle $Z^2(L,\bbR^\times)$.
Let $\bbR[\hat{L}]= \bigoplus_{\al \in L} \bbR e_\al$ be an $\bbR$-algebra with the multiplication defined by
\begin{align*}
e_\al \cdot e_\be = \ep(\al,\be)e_{\al+\be}
\end{align*}
for $\al,\be\in L$. This algebra is called a \textit{twisted group algebra} introduced in \cite{FLM88VertexOperatorAlgebras}.
It is easy to show that
\begin{align}
\ep(\al,\be)=\ep(\al,-\be)=\ep(-\al,\be) 
\quad\text{ and }\quad \ep(\al,\al)=(-1)^{(\al,\al)_\lat/2}
\label{eq_minus_ep}
\end{align}
holds for any $\al,\be \in L$.
Define a linear map $\phi:\bbR[\hat{L}] \rightarrow \bbR[\hat{L}]$ by
\begin{align}
\phi(e_\al)= e_{-\al},\label{eq_grouplattice_theta}
\end{align}
which is an $\bbR$-algebra automorphism by \eqref{eq_minus_ep}.

Set
\begin{align*}
H= L \otimes_\bbZ \bbR,
\end{align*}
which is equipped with the symmetric bilinear form induced from $L$.
Let $P(H)$ be a set of $\bbR$-linear maps $p\in \End_\bbR H$ such that:
\begin{enumerate}
\item[P1)]
$p^2=p$, that is, $p$ is a projection;
\item[P2)]
The subspaces $\ker(1-p)$ and $\ker(p)$ are orthogonal to each other.
\end{enumerate}
Let $P_>(H)$ be a subset of $P(H)$
consisting of $p \in P(H)$ such that:
\begin{enumerate}
\item[P3)]
$\ker(1-p)$ is positive-definite and $\ker(p)$ is negative-definite.
\end{enumerate}
For $p\in P(H)$, set $\p=1-p$ and $H_l=\ker(\p)$ and $H_r=\ker(p)$.
We will construct a full vertex algebra $F_{L,p}$ for each $p\in P(H)$
and show that it is unitary if $p \in P_>(H)$.

Let $(n_L,m_L)$ be the signature of $H$. Then, the orthogonal group $O(H) \cong O(n_L,m_L,\bbR)$ acts on $P(H)$ and $P_>(H)$.
Then, as an $O(H)$-set,
\begin{align}
P_>(H) \cong O(n_L,m_L;\bbR) / O(n_L;\bbR)\times O(m_L;\bbR).\label{eq_ohset}
\end{align}

Let $p\in P(H)$.
Define the new bilinear forms $(-,-)_p: H\times H \rightarrow \bbR$
by 
\begin{align*}
(h,h')_p=(ph,ph')_\lat-(\p h,\p h')_\lat
\end{align*}
for $h,h' \in H$.
By (P1) and (P2), $(-,-)_p$ is non-degenerate.
Note that $(-,-)_p$ is positive-definite if and only if $p \in P_>(H)$.

Let $\hat{H}^p=\bigoplus_{n \in \bbZ} H\otimes t^n \oplus \bbR c$ be the affine Heisenberg Lie algebra associated with $({H},(-,-)_p)$
and $\hat{H}_{\geq 0}^p=\bigoplus_{n \geq 0} H\otimes t^n \oplus \bbR c$ a subalgebra of $\hat{H}^p$.
Define the action of $\hat{H}_{\geq 0}^p$ on the twisted group algebra
$\bbR[\hat{L}]= \bigoplus_{\al \in L} \bbR e_\al$ by
\begin{align*}
c e_\alpha&=e_\al \\
h\otimes t^n e_\alpha &=
\begin{cases}
0, &n \geq 1,\cr
(h,\al)_p e_\alpha, &n = 0
\end{cases}
\end{align*}
for $\al \in H$.
Let $F_{L,p}$ be the $\hat{H}^p$-module induced from $\bbR[\hat{L}]$.
Denote by $h(n)$ the action of $h \otimes t^n$ on $F_{L,p}$ for $n \in \bbZ$.
For $h \in H$, set 
\begin{align*}
h(\uz) &= \sum_{n \in \bbZ}( (ph)(n)z^{-n-1}+ (\p h)(n)\z^{-n-1}) \in
\End F_{L,p}[[z^\pm,\z^\pm]]\\
h^+(\uz) &=\sum_{n \geq 0}( (ph)(n)z^{-n-1}+ (\p h)(n)\z^{-n-1}) \\
h^-(\uz) &=\sum_{n \geq 0}( (ph)(-n-1)z^{n}+ (\p h)(-n-1)\z^{n}).\\
E^+(h,\uz)&=
\exp\biggl(-\sum_{n\geq 1}\left(\frac{p h(n)}{n}z^{-n}+\frac{\p h(n)}{n}\z^{-n}\right)
\biggr)\\
E^-(h,\uz)&=\exp\biggl(\sum_{n\geq 1}\left(\frac{p h(-n)}{n}z^{n} +\frac{\p h(-n)}{n}\z^{n}\right)
\biggr).
\end{align*}
For $h_r \in H_r$ and $h_l \in H_l$,
$h_r(\uz)$ and $h_l(\uz)$ are denoted by $h_l(z)$ and $h_r(\z)$, respectively.

Let $\al \in H$.
Denote by $l_{e_\al} \in \End \bbR[\hat{L}]$ the left multiplication by ${e_\al}$ and define the linear map
$z^{p\alpha} \z^{\p \alpha} :\bbR[\hat{L}] \rightarrow  \bbR[\hat{L}][z^\bbR,\z^\bbR]$
by $z^{p\alpha} \z^{\p \alpha}e_\beta= z^{(p\alpha,p\beta)_p}\z^{(\p\alpha,\p\beta)_p}e_\beta$ 
for $\beta \in L$.
Then, set
\begin{align*}
e_\al(\uz)&=E^-(\al,\uz)E^+(\al,\uz)l_{e_\al} z^{p \alpha} \z^{\p \alpha} \in \End\, F_{L,p}[[z^\pm,\z^\pm]][z^\bbR,\z^\bbR].
\end{align*}

By Poincar{\'e}-Birkhoff-Witt theorem,
$F_{L,p}$ is spanned by 
\begin{align*}
\{h_l^1(-n_1-1)\dots h_l^l(-n_l-1)h_r^1(-m_1-1)\dots h_r^k(-m_k-1)e_\al \},
\end{align*}
where $h_l^i \in H_l$, $n_i \in \bbZ_{\geq 0}$ and $h_r^j \in H_r$, $m_j \in \bbZ_{\geq 0}$
 for any $1\leq i \leq l$ and $1\leq j \leq k$ and $\al \in H$.
Then, a map
$Y:F_{L,p} \rightarrow \End\, F_{L,p}[[z^\pm,\z^\pm]][z^\bbR,\z^\bbR]$ is defined inductively as follows:
For $\al \in L$, define $Y({e_\al},\uz)$ by $Y({e_\al},\uz)={e_\al}(\uz)$.
Assume that $Y(v,\uz)$ is already defined for $v \in F_{L,p}$.
Then, for $h_r \in H_r$ and $h_l \in H_l$ and $n,m \in \bbZ_{\geq 0}$,
$Y(h_l(-n-1)v,\uz)$ and $Y(h_r(-m-1)v,\uz)$ is defined by 
\begin{align*}
Y(h_l(-n-1)v,\uz)&=\Bigl(\frac{1}{n!}\frac{d}{dz}^n h_l^-(z)\Bigr)Y(v,\uz)+Y(v,\uz)\Bigl(\frac{1}{n!}\frac{d}{dz}^n h_l^+(z)\Bigr) \\
Y(h_r(-m-1)v,\uz)&=\Bigl(\frac{1}{m!}\frac{d}{d\z}^n h_r^-(\z)\Bigr)Y(v,\uz)+Y(v,\uz)\Bigl(\frac{1}{m!}\frac{d}{d\z}^n h_r^+(\z)\Bigr).
\end{align*}

Set
\begin{align*} 
\vac=1\otimes e_0, \quad
\nu  = \frac{1}{2}\sum_{i=1}^{\dim H_l} h_l^i (-1)h_l^i,\quad
\bar{\nu} = \frac{1}{2} \sum_{j=1}^{\dim H_r} h_r^j (-1)h_r^j,
\end{align*}
where $h_l^i$ and $h_r^j$ is an orthonormal basis of $H_l$ and $H_r$ with respect to the bilinear form $(-,-)_p$.
Set $F=F_{L,p}$ and $$F_{h,\h}^\al=\{v \in G\;|\; \nu(1,-1)v=h v, \bar{\nu}(-1,1)v=\h v, h(0)v=(\al,h)_p v \text{ for all } h \in H \}$$
for $h,\h \in \bbR$ and $\al \in H$. Then,
$e_\al \in F_{\frac{1}{2}(p\al,p\al)_p,\frac{1}{2}(\p\al,\p\al)_p}^\al$.

Then, by \cite[Theorem 4.14, Proposition 5.11, Theorem 6.5]{Moriwaki21}, we have:
\begin{proposition}\label{standard}
For $p \in P(H)$, $(F_{L,p},Y,\vac)$ is a full vertex algebra over $\bbR$.
If $p \in P_>(H)$, then $(L(0),\Ld(0))$-eigenvalues are non-negative, and $(F_{L,p},Y,\vac,\nu,\bar{\nu})$ is a full vertex operator algebra over $\bbR$.
Moreover, the isomorphism classes of full vertex operator algebras are parametrized by
\begin{align*}
\mathrm{Aut}L \backslash P_>(H) \cong \mathrm{Aut}L \backslash O(n_L,m_L;\bbR) / O(n_L;\bbR)\times O(m_L;\bbR),
\end{align*}
where $\mathrm{Aut}L \subset O(n_L,m_L;\bbR)$ is the automorphism group of the lattice $L$.
\end{proposition}

Hereafter we assume that $p \in P_>(H)$.
The $\bbR$-algebra involution \eqref{eq_grouplattice_theta} can extend to 
an automorphism of the full vertex algebra $\phi:F_{L,p} \rightarrow F_{L,p}$ by
\begin{align*}
\phi: &h_l^1(-n_1-1)\dots h_l^l(-n_l-1)h_r^1(-m_1-1)\dots h_r^k(-m_k-1)e_\al\\
&\mapsto (-1)^{l+k} h_l^1(-n_1-1)\dots h_l^l(-n_l-1)h_r^1(-m_1-1)\dots h_r^k(-m_k-1)e_{-\al}
\end{align*}
(see \cite[Proposition 5.13]{Moriwaki21} applied to the automorphism $\phi$ \eqref{eq_grouplattice_theta} of the AH pair
(the twisted group algebra) giving an automorphism of the G-full VA and by \cite[Theorem 5.8]{Moriwaki21}
an automorphism of $(F_{L,p}, Y, \vac)$ by functoriality).
Note that $F_{L,p}^\bbC = F_{L,p} \otimes_{\bbR}\bbC$ is naturally a full vertex algebra over $\bbC$.
For $i=0,1$, set $F_i =\{a\in F_{L,p}\mid \phi(a)=(-1)^i a\}$,
which are regarded as subspaces of $F_{L,p}^\bbC$,
 and
\begin{align*}
\tilde{F} = F_0 \oplus \ri F_1 \subset F_{L,p}^\bbC.
\end{align*}
Then, $\tilde{F}$ is a real subalgebra of $F_{L,p}^\bbC$.
It is clear that $\nu, \bar\nu \in F_0 \subset \tilde{F}$
and $\tilde{F}$ satisfy the assumption of Proposition \ref{prop_inv_bilinear}.
Thus, $\tilde{F}$ has a unique non-degenerate symmetric invariant bilinear form:
\begin{align*}
\langle-,-\rangle:\tilde{F} \otimes \tilde{F} \rightarrow \bbR
\end{align*}
with $\langle \vac,\vac\rangle=1$, which is a restriction of that of $F_{L,p} \otimes_{\bbR}\bbC$.
\begin{proposition}\label{prop_positive_definite}
The bilinear form on $\tilde{F}$ is positive-definite.
\end{proposition}
\begin{proof}
We first show that the bilinear form is positive-definite on
\begin{align*}
\bbR (e_\al+e_{-\al}) \oplus \bbR \ri(e_\al-e_{-\al})
\end{align*}
for any $\al \in L$.
We have
\begin{align*}
&\langle e_\al+e_{-\al}, e_\al+e_{-\al}\rangle \\
&=\lim_{z\to 0} \langle Y(e_\al+e_{-\al},\uz)\vac, e_\al+e_{-\al}\rangle\\
&=\lim_{z\to 0} (-1)^{(p\al,p\al)_\lat^2/2+(\p\al,\p\al)_\lat^2/2}z^{-(p\al,p\al)_\lat^2}\z^{(\p\al,\p\al)_\lat^2} \langle \vac, Y(e_\al+e_{-\al},\uz^{-1})(e_\al+e_{-\al})\rangle\\
&= (-1)^{(\al,\al)_\lat/2} \langle \vac, l_{e_\al}e_{-\al}+l_{e_{-\al}}e_{\al})\rangle \\
&= (-1)^{(\al,\al)_\lat/2}(\ep(\al,-\al)+\ep(-\al,\al)) \langle \vac,\vac \rangle =2,
\end{align*}
where we used \eqref{eq_minus_ep} in the last line.
Similarly,
\begin{align*}
\langle \ri (e_\al-e_{-\al}), \ri(e_\al-e_{-\al})\rangle =2\quad
\text{ and }\quad
\langle (e_\al+e_{-\al}), \ri(e_\al-e_{-\al})\rangle=0.
\end{align*}
It is clear that $\tilde{F}$ is spanned by
\begin{align*}
V_\al^+ &= \mathrm{Span}_\bbR \{ \ri^{l+k} h_l^1(-n_1-1)\dots h_l^l(-n_l-1)h_r^1(-m_1-1)\dots h_r^k(-m_k-1)(e_\al+e_{-\al})\}\\
V_\al^- &= \mathrm{Span}_\bbR \{ \ri^{l+k+1} h_l^1(-n_1-1)\dots h_l^l(-n_l-1)h_r^1(-m_1-1)\dots h_r^k(-m_k-1)(e_\al-e_{-\al})\}.
\end{align*}
Then, $V_\al^{\eta_1}$ and $V_\be^{\eta_2}$ are orthogonal if 
$\eta_1 \neq \eta_2$ or $\al \neq \pm \be$.
Since $(-,-)_p$ is positive-definite for any $p \in P_>(H)$,
by the commutator relation of affine Heisenberg Lie algebra,
$V_\al^+$ and $V_\al^-$ are both positive-definite.
\end{proof}
There is a unique anti-linear involution $\phi:F_{L,p}^\bbC \rightarrow F_{L,p}^\bbC$ whose fixed point real subalgebra is $\tilde{F}$.
Hence, $\tilde{F}$ is a unitary full VOA.

The (chiral) polynomial energy bounds for intertwining operators among VOA modules are studied by \cite{Toledano-LaredoThesis, Gui19}.
The following proposition is clear from the definition of the polynomial energy bounds:
\begin{proposition}
\label{prop_chiral_full_growth}
Let $F$ be a locally $C_1$-cofinite unitary full VOA. Assume that
\begin{enumerate}
\item
$M_i$ and $\overline{M}_i$ in the definition of local $C_1$-cofiniteness are unitary modules of the VOA $V$ and $W$;
\item
All the VOA intertwining operators in \eqref{eq_int_C1_decomp} satisfy the chiral polynomial energy bounds.
\end{enumerate}
Then, $F$ satisfies the polynomial energy bounds.
\end{proposition}

Then, we have:
\begin{theorem}\label{thm_unitarity_lattice}
For any $p \in P_>(H)$, $(F_{L,p}^\bbC,Y,\vac,\nu,\bar\nu,\phi)$ is a unitary full vertex operator algebra and satisfies
the local $C_1$-cofiniteness condition, the polynomial energy bounds and the polynomial spectral density.
\end{theorem}
\begin{proof}
The canonical subvertex operator algebra $\ker \Ld(-1) \subset F$ contains the Heisenberg vertex algebra $M_{H_l}(0)$, which is generated by the weight one subspace $H_l$, and similarly $\overline{M}_{H_r}(0) \subset \ker L(-1)$.
For any $\al \in L$, $F^\al = \bigoplus_{n,m\in \bbZ_{\geq 0}}F_{n+\frac{1}{2}(p\al,p\al)_p,m+\frac{1}{2}(\p\al,\p\al)_p}^\al$ is the Verma module of $M_{H_l}(0)\otimes M_{H_r}(0)$ and is local $C_1$-cofinite.

The chiral polynomial energy bounds for Heisenberg modules are shown in \cite[Proposition 1.2.1 and Proposition 1.3.1]{Toledano-LaredoThesis} and \cite[Theorem A.2]{Gui19}. Hence, the polynomial energy bounds follows from Proposition \ref{prop_chiral_full_growth}.

It is clear that
\begin{align*}
\#\{(h,\h) &\in \bbR^2\mid n-1 \leq h \leq n, m-1 \leq \h \leq m, F_{h,\h}\neq 0 \}\\
&=\#\{\al \in L \mid (p\al,p\al)_p \leq n \text{ and } (\p\al,\p\al)_p \leq m\}
\end{align*}
for any $n,m\in \bbZ$.
Since $\ker p$ and $\ker \p$ are orthogonal,
\begin{align*}
\#\{\al \in L \mid (p\al,p\al)_p \leq n \text{ and } (\p\al,\p\al)_p \leq m\}
\leq \#\{\al \in L \mid (\al,\al)_p \leq n+m\}.
\end{align*}
Since $(\bullet,\bullet)_p$ is positive-definite, the right-hand-side can be estimated by the lattice points (the volume) of the $n_L+m_L=\mathrm{rank} H$ dimensional sphere.
Hence, there is a constant $C>0$ such that
\begin{align*}
\#\{\al \in L \mid (\al,\al)_p \leq n+m\} <C (n+m)^{\mathrm{rank} H}.
\end{align*}
From this it is straightforward that
$\#\{(h,\h) \in \bbR^2\mid N \leq h + \h \leq N+1, F_{h,\h}\neq 0 \}$
is bounded by a polynomial in $N$ as in \eqref{eq:density}.
\end{proof}

\begin{remark}
Note that $\ker \Ld(-1)$ generically coincides with the Heisenberg vertex operator algebra, however, at the rational points of 
\begin{align*}
\mathrm{Aut}L \backslash O(n_L,m_L;\bbR) / O(n_L;\bbR)\times O(m_L;\bbR),
\end{align*}
$\ker \Ld(-1)$ is a lattice vertex operator algebra. For example, for $L=\tw$, the unique even unimodular lattice of signature $(1,1)$, $\bbR_{> 0} \cong O(1,1;\bbR) / O(1;\bbR)\times O(1;\bbR)$
and $\mathrm{Aut}\tw \cong \bbZ_2\times \bbZ_2$ acts as $R \mapsto R^{-1}$ for $R\in \bbR_{>0}$,
which is the T-duality of string theory.
Then, we have
\begin{align*}
\ker \Ld(-1) =
\begin{cases}
 M_1(0) & \text{ if } R^2 \notin \bbQ\\
 V_{\sqrt{2pq}\bbZ} & \text{ if } R^2=p/q \text{ with }p,q\in \bbZ \text{ coprime},
\end{cases}
\end{align*}
where $M_1(0)$ is the rank one Heisenberg vertex operator algebra and $V_{\sqrt{2pq}\bbZ}$ is the lattice vertex operator algebra associated with the rank one lattice $\bbZ\al$, $(\al,\al)=2pq$ (for more detail see \cite[Section 6.3]{Moriwaki21}).
\end{remark}

\section{Outlook}\label{outlook}
We plan to investigate more examples of full VOAs such as the Ising model and framed algebras \cite{Moriwaki21Code},
and compare them with direct constructions of Wightman fields \cite{AGT23Pointed}.

We are partly motivated by the possibility to deform full CFT to massive models \cite{Zamolodchikov89-1}.
To do that, it would be a great help to have stronger tools that characterize the CFT such as functional integration measures \cite{GJ87}.
It would be interesting to see which unitary CFTs are associated with such measures.
In this regard, our proof shows the linear growth condition (OS\ref{ax:lineargrowth}) under (PEB) and (PSD),
yet we are unable to show another variant (E0$^{\prime\prime}$) of \cite{OS75} which asks
that the Schwinger functions $S_n$ extend to $\scS(\bbR^{2n})$, including the coinciding points.
To have an functional integration measure, it seems necessary that such extensions are possible,
cf.\! \cite[Proposition 5.1]{GH21APDEConstruction}, and considering fields with small scaling dimensions might help.

Related with this, such a deformation might be more conveniently performed on the Riemann sphere.
As a result, the deformed QFT should be defined on the de Sitter space \cite{BJM23, JT23Towards}.
For this purpose, it is desirable to have an analogue of the OS reconstruction for the sphere, cf.\! \cite{Schlingemann99Sphere}.
In addition, similar to the OS axiom framework, reflection positivity plays a role in the representation theory of symmetric Lie groups,
where by analytic continuation one aims to construct representations of the dual group by starting from those of the group,
by providing a new Hilbert structure to the original presentation Hilbert space.
In this way, it is interesting to study reflection positivity representations on their own,
cf.\! \cite{NO18ReflectionPositivity} for a general account on this topic and \cite{ANS22ReflectionPositivity, ANS24ReflectionPositivity}
for the prototypical examples $\bbZ, \bbR$, and the circle group $T_\beta$ for $\beta>0$.

\subsection*{Acknowledgements}
M.S.A. is a Humboldt Research Fellow supported by the Alexander von Humboldt Foundation. Early stages of this work have been carried on while M.S.A. was a JSPS International Research Fellow and she received support by the Grant-in-Aid Kakenhi n. 22F21312.
She is partially supported by INdAM--GNAMPA Project, codice CUP E53C23001670001 and GNAMPA--INdAM.
Y.M. is supported by Grant-in Aid for Early-Career Scientists (24K16911) and FY2023 Incentive Research Projects (Riken).
Y.T. is partially supported by the MUR Excellence Department Project MatMod@TOV awarded to the Department of Mathematics,
University of Rome ``Tor Vergata'' CUP E83C23000330006, the University of Rome ``Tor Vergata'' funding OAQM
CUP E83C22001800005, INdAM--GNAMPA Project codice CUP E53C23001670001 ``Operator algebras and infinite quantum systems''
and GNAMPA--INdAM.

{\small
\def\polhk#1{\setbox0=\hbox{#1}{\ooalign{\hidewidth
  \lower1.5ex\hbox{`}\hidewidth\crcr\unhbox0}}} \def\cprime{$'$}
}
\end{document}